\pdfoutput=1

\documentclass[11pt]{article}
\usepackage{graphicx}
\usepackage{fullpage}

\usepackage{amsmath, amssymb, amsthm, bbm}
\usepackage{amsmath,amssymb}
\usepackage{thmtools}
\usepackage{mathtools}
\usepackage[shortlabels]{enumitem}
\usepackage{hyperref}
\usepackage{stmaryrd}
\usepackage{tikz}
\usepackage{subfig}
\hypersetup{breaklinks, urlcolor=blue, colorlinks, citecolor=green!50!black, linkcolor=blue}
\usepackage[capitalize,noabbrev,nameinlink]{cleveref}
\usepackage[ruled,vlined,linesnumbered,algo2e]{algorithm2e}

\declaretheorem[numberwithin=section,refname={Theorem,Theorems},Refname={Theorem,Theorems}]{theorem}
\declaretheorem[numberlike=theorem]{lemma}
\declaretheorem[numberlike=theorem]{proposition}
\declaretheorem[numberlike=theorem]{corollary}
\declaretheorem[numberlike=theorem]{claim}
\declaretheorem[numberlike=theorem]{fact}
\declaretheorem[numberlike=theorem,style=definition]{definition}
\declaretheorem[numberlike=theorem,style=definition]{problem}
\declaretheorem[numberlike=theorem,style=remark]{remark}
\declaretheorem[numberlike=theorem]{observation}

\def\final{0}
\ifnum\final=0
\newcommand{\jiale}[1]{{\color{red} [\textbf{Jiale}: #1]}}
\newcommand{\tawei}[1]{{\color{purple} [\textbf{Ta-Wei}: #1]}}
\newcommand{\sidford}[1]{{\color{green} [\textbf{Aaron}: #1]}}
\newcommand{\todo}[1]{{\color{blue} [\textbf{TODO}: #1]}}
\else
\newcommand{\jiale}[1]{}
\newcommand{\tawei}[1]{}
\newcommand{\sidford}[1]{}
\newcommand{\todo}[1]{}
\fi

\newcommand{\expnear}{41}
\newcommand{\expalmost}{17}
\DeclarePairedDelimiter{\norm}{\lVert}{\rVert}
\newcommand{\OPT}{\mathsf{OPT}}
\newcommand{\bb}{\boldsymbol{b}}
\newcommand{\bc}{\boldsymbol{c}}
\newcommand{\bd}{\boldsymbol{d}}
\newcommand{\Bf}{\boldsymbol{f}}
\newcommand{\bg}{\boldsymbol{g}}
\newcommand{\bx}{\boldsymbol{x}}
\newcommand{\bw}{\boldsymbol{w}}
\newcommand{\by}{\boldsymbol{y}}
\newcommand{\byz}{\boldsymbol{s}}
\newcommand{\bp}{\boldsymbol{p}}
\newcommand{\br}{\boldsymbol{r}}
\newcommand{\bz}{\boldsymbol{z}}
\newcommand{\bB}{\boldsymbol{B}}
\newcommand{\poly}{\mathrm{poly}}
\newcommand{\supp}{\mathrm{supp}}
\newcommand{\eps}{\varepsilon}
\newcommand{\nnz}{\mathrm{nnz}}
\newcommand{\entropy}{H}
\newcommand{\R}{\mathbb{R}}
\newcommand{\N}{\mathbb{N}}
\newcommand{\Z}{\mathbb{Z}}
\newcommand{\M}{\mathcal{M}}
\newcommand{\B}{\mathcal{B}}
\newcommand{\X}{\mathcal{X}}

\renewcommand{\P}{\mathcal{P}}
\newcommand{\A}{\mathcal{A}}
\renewcommand{\S}{\mathcal{S}}
\renewcommand{\O}{\mathcal{O}}
\renewcommand{\L}{\mathcal{L}}
\newcommand{\otilde}{\widetilde{O}}
\newcommand{\defeq}{\stackrel{\mathrm{{\scriptscriptstyle def}}}{=}}
\newcommand{\init}{\mathsf{init}}
\renewcommand{\output}{\mathsf{output}}
\newcommand{\update}{\mathsf{update}}

\DeclareMathOperator*{\argmax}{argmax}

\Crefname{algocf}{Algorithm}{Algorithms}
\Crefname{claim}{Claim}{Claims} 
\Crefname{problem}{Problem}{Problems}
\Crefname{fact}{Fact}{Facts}
\Crefname{observation}{Observation}{Observations}

\title{Entropy Regularization and Faster Decremental\\ Matching in General Graphs}
\author{
Jiale Chen \\ Stanford University \\ \texttt{jialec@stanford.edu} 
\and
Aaron Sidford \\ Stanford University \\ \texttt{sidford@stanford.edu}
\and
Ta-Wei Tu \\ Stanford University \\ \texttt{taweitu@stanford.edu}}
\date{}

\begin{document}

\begin{titlepage}
  \maketitle \pagenumbering{roman}

  \begin{abstract}
  We provide an algorithm that maintains, against an adaptive adversary, a $(1-\eps)$-approximate maximum matching in $n$-node $m$-edge general (not necessarily bipartite) undirected graph undergoing edge deletions with high probability with (amortized) $O(\poly(\eps^{-1}, \log n))$ time per update. We also obtain the same update time for maintaining a fractional approximate weighted matching (and hence an approximation to the value of the maximum weight matching) and an integral approximate weighted matching in dense graphs.\footnote{Independently and concurrently, Aditi Dudeja obtained new decremental weighted matching results for general graphs~\cite{Dudeja24b}.}
  Our unweighted result improves upon the prior state-of-the-art which includes a $\poly(\log{n}) \cdot 2^{O(1/\eps^2)}$ update time [Assadi--Bernstein--Dudeja 2022] and an $O(\sqrt{m} \eps^{-2})$ update time [Gupta--Peng 2013], and our weighted result improves upon the $O(\sqrt{m}\eps^{-O(1/\eps)}\log{n})$ update time due to [Gupta--Peng 2013].
  
  To obtain our results, we generalize a recent optimization approach to dynamic algorithms from [Jambulapati--Jin--Sidford--Tian 2022]. We show that repeatedly solving entropy-regularized optimization problems yields a lazy updating scheme for fractional decremental problems with a near-optimal number of updates. To apply this framework we develop optimization methods compatible with it and new dynamic rounding algorithms for the matching polytope.
\end{abstract}

  \setcounter{tocdepth}{3}
  \newpage
  \tableofcontents
  \newpage
\end{titlepage}
	
\newpage
\pagenumbering{arabic}

\section{Introduction} \label{sec:intro}

\emph{Dynamic matching} is a fundamental and incredibly well-studied problem in data structure design and dynamic graph algorithms.
In this problem, there is a sequence of modifications, e.g., edge insertion and deletions, to an undirected graph $G = (V, E)$ with vertices $V$ and edges $E$.
The goal is to maintain a \emph{matching $M \subseteq E$}, i.e., a subset of edges that are pairwise disjoint, of approximately maximum size, i.e., $|M| \geq (1 - \eps) M^*(G)$ where $\eps$ is a specified accuracy parameter and $M^*(G)$ is the size of the maximum matching in the current graph, $G$.
In this paper we consider solving this problem with the goal of obtaining improved (amortized) update time per operation.\footnote{Our new algorithms periodically rebuild a fractional matching when its value degrades. Consequently, their update times are (na\"{i}vely) amortized instead of worst-case.
We leave de-amortizing our results as an open problem.}

The dynamic matching problem is also notoriously challenging to solve efficiently.
In the general setting of maintaining such a $(1-\eps)$-approximate maximum matching in a general undirected graph undergoing edge insertions and deletions, the state-of-the-art includes
an $O\left(\sqrt{m}\eps^{-2}\right)$ update time algorithm of \cite{GuptaP13}
for all $\eps$ and
an $n / (\log^{*}{n})^{\Omega(1)}$ update time algorithm of \cite{AssadiBKL23} for $\eps=\Omega((\log^* n)^{-c})$ where $c>0$ is some small constant.
Subsequent to the initial release of our paper, the work of \cite{AssadiK24}, building on top of \cite{BehnezhadG24}, obtained an $O\left(n^{o(1)} \cdot \mathsf{ORS}(n, \Theta_{\eps}(n))\right)$ update time algorithm for constant $\eps$, where $\mathsf{ORS}(n, r)$ is the maximum possible density of the so-called \emph{ordered Ruzsa-Szemer\'{e}di Graphs}: the current, state-of-the-art bound is $n^{o(1)} \leq \mathsf{ORS}(n, \Theta_{\eps}(n)) \leq n^{1-o(1)}$ for constant $\eps$~\cite{BehnezhadG24}.
Additional improvements are only known in special cases. For example, when there are vertex (instead of edge) updates to a bipartite graph then \cite{BosekLSZ14,ZhengH23} provide $\widetilde{O}(\eps^{-1})$ time\footnote{In this paper, we use $\widetilde{O}(\cdot)$ to hide $\poly(\log(n), \log(\eps^{-1}))$ factors and $\widehat{O}(\cdot)$ to hide subpolynomial, $(n\eps^{-1})^{o(1)}$ factors in the $O(\cdot)$ notation. $\widetilde{\Omega}(\cdot)$, $\widehat{\Omega}(\cdot)$, $\widetilde{\Theta}(\cdot)$, and $\widehat{\Theta}(\cdot)$ are defined analogously.} algorithms.
Additionally, recent breakthrough results maintain better-than-$2$ approximations to the size of the maximum matching (rather than the matching itself) in near-optimal update time~\cite{Behnezhad23,BhattacharyaKSW23dynamic2,AzarmehrBR24} and better approximations in sub-linear update time~\cite{BhattacharyaKS23dynamic1}. 

Towards explaining the challenging nature of dynamic matching, recent work established conditional lower bounds on the problem's complexity.
To maintain exactly the maximum matching, there are conditional lower bounds on the update times of $\Omega(\sqrt m)$~\cite{AbboudW14,HenzingerKNS15,KopelowitzPP16} and $\Omega(n^{1.407})$~\cite{BrandNS19}. To maintain approximate maximum matchings, very recently, \cite{Liu24} ruled out the possibility of a truly sublinear, $\eps^{-O(1)} n^{1-O(1)}$ update time algorithm in the fully dynamic setting, assuming a new dynamic approximate \textsf{OMv} conjecture. Additionally, assuming only the standard \textsf{OMv} conjecture, \cite{Liu24} ruled out the possibility of truly sublinear update time algorithms for the closely related problem of maintaining a $(1+\eps)$-approximate vertex cover.

Corresponding to the difficulty of dynamic matching, recent work studied the \emph{decremental setting} where starting from an initial graph only edge deletions are permitted until the graph is deleted.
Excitingly, \cite{BernsteinGS20} provided an $\widetilde{O}(\eps^{-4})$ time per update algorithm for this problem. \cite{JambulapatiJST22} then obtained update times of $\widetilde{O}(\eps^{-3})$ and $\widehat{O}(\eps^{-2})$ for the problem of maintaining fractional matchings. \cite{BhattacharyaKSW23rounding} then obtained update times of $\widetilde{O}(\eps^{-3})$ and $\widehat{O}(\eps^{-2})$ for maintaining integral matchings through improved dynamic rounding algorithms for fractional matchings.
Recently, \cite{AssadiBD22} obtained an $\widetilde{O}(1) \cdot 2^{O(1/\eps^{2})}$ update time for general (not necessarily bipartite) graphs. 

Strikingly, these results show that decremental matching can be solved near-optimally in \emph{bipartite graphs} for $\eps = 1/\poly(\log n)$.
However, in \emph{general graphs}, for such $\eps$, the state-of-the-art is still $O(\sqrt{m}\eps^{-2})$ time \cite{GuptaP13}.
A central question motivating this work is whether it is possible to narrow this gap and develop improved decremental matching algorithms for general graphs. 
One of our main results is an affirmative answer to this question; we provide $\widetilde{O}(\eps^{-\expnear})$ and $\widehat{O}(\eps^{-\expalmost})$ update time algorithms for maintaining a $(1 - \eps)$-approximate matching in a general graph undergoing edge deletions.
Our algorithm succeeds with high probability (w.h.p.\footnote{We use \emph{w.h.p.} in this paper as an abbreviation of \emph{with high probability}, meaning the success probability can be made at least $1 - n^{-c}$ for any constant $c > 0$ by increasing the runtime by a constant factor.}) against an adaptive adversary that has access to the internal randomness used by our algorithm and can design its future updates based on that.
This is the first decremental dynamic matching algorithm for general graphs that achieves subpolynomial accuracy and subpolynomial update time simultaneously.\footnote{More precisely, when $\eps = \Theta(1/\poly(\log(n)))$, or even $\eps = \Theta(1/n^{o(1)})$, previous algorithms require $\poly(n)$ or even $\exp(n)$ update times, while our algorithms achieve an $n^{o(1)}$ update time.}

We obtain our results through the development of a general framework that also allows us to obtain improved runtimes in the decremental setting of the more challenging dynamic \emph{weighted} matching as well.
In this problem we must maintain a $(1-\eps)$-approximate maximum weight matching given polynomially bounded edge weights $\bw \in \N^E$. 
Prior to our work, the state-of-the-art was an $\widetilde{O}(\sqrt{m}\eps^{-O(1/\eps)})$ time ~\cite{GuptaP13} and an $\widetilde{O}(\eps^{-3})$ time per update algorithm for maintaining fractional matchings in bipartite graphs \cite{BhattacharyaKS23} (which also applied to a broader class of partially-dynamic packing/covering linear programs).
In contrast, we provide algorithms which decrementally maintain $(1-\eps)$-approximate fractional matchings in general weighted graphs in $\widetilde{O}(\eps^{-\expnear})$ and $\widehat{O}(\eps^{-\expalmost})$ time per update and integral $(1-\eps)$-approximate matchings in dense graphs in the same update times.
We also provide an $\widetilde{O}(\eps^{-O(1/\eps)})$ update time integral matching algorithm, which is not limited to dense graphs, using a weight reduction framework of \cite{GuptaP13}.
These are the first near-optimal partially dynamic algorithms for weighted matchings with edge updates.

\subsection{Approach}

\paragraph{Lazy Updates and Congestion Balancing.} 
Our algorithms follow a natural, time-tested \emph{lazy} approach to solving dynamic matching problems~\cite{GuptaP13,BernsteinS15,BernsteinS16,BernsteinGS20,JambulapatiJST22,AssadiBD22}. Broadly, we compute an approximate fractional matching, delete edges from it as needed, and then,  when the updates cause the solution to change, we \emph{rebuild}, computing a new fractional matching. By efficiently computing the fractional matching, computing fractional matchings that limit the number of rebuilds, and efficiently rounding, we obtain our dynamic matching algorithms. 

More specifically, our algorithms follow a template common to \cite{BernsteinGS20, JambulapatiJST22,AssadiBD22}
(see \Cref{sec:tech:lazy} for a more precise description). In this framework, we first, compute a $(1 - \delta)$-approximate fractional matching. Then, when edges are deleted, the corresponding fractional assignment on those edges is removed as well. Once the value of the fractional matching decreases by $(1-\delta)$ multiplicatively, a new  $(1 - \delta)$-approximate fractional matching is computed and the process is repeated; we call each computation of a $(1 - \delta)$-approximate fractional matching a \emph{rebuild}. Any algorithm following this \emph{lazy update framework} clearly maintains a $(1 - \delta)^2$-approximate fractional matching and consequently, by picking, e.g., $\delta = \eps / 2$, this algorithm maintains a $(1-\eps)$-approximate fractional solution. What is perhaps unclear, is how to make this approach efficient. 

The update time of an algorithm following the lazy update framework is governed by:\footnote{Additionally, the algorithm must remove fractional assignments to the deleted edges from the fractional matching, but that is trivial to implement in $O(1)$ per update.}
\begin{enumerate}
    \item \label{item:lazy:rebuild_count} the number of rebuilds, i.e., the number of approximate fractional matchings computed,
    \item \label{item:lazy:rebuild_cost} the cost per rebuilding, i.e., the cost of computing each fractional matching, and
    \item \label{item:lazy:round_cost} the cost of rounding, i.e., the cost of turning these dynamically maintained fractional matchings into dynamically maintained integral matchings. 
\end{enumerate}
Note that computing a $(1 - \eps)$-approximate matching takes $\widetilde{\Omega}(m)$ time in the worst case~\cite{BhattacharyaKS23sublinear}.
Consequently, if each rebuild is computed from scratch (as they are in our algorithms) then the cost of each rebuild is $\Omega(m)$ (\Cref{item:lazy:rebuild_cost}) and to obtain an $O(\poly(\log n, \eps^{-1}))$-update time it must be that the number of rebuilds (\Cref{item:lazy:rebuild_count}) is $O(\poly(\log n, \eps^{-1}))$.
However, it is unclear, just from the approach, whether or why this should be obtainable. 

Nevertheless, in a striking result, \cite{BernsteinGS20} showed that it was possible to follow this framework and obtain $\widetilde{O}(\eps^{-4})$ time per update for bipartite graphs.
The algorithm had $\widetilde{O}(\eps^{-3})$ rebuilds (\Cref{item:lazy:rebuild_count}) at a cost of $\widetilde{O}(m\eps^{-1})$ time per rebuild\footnote{The $\widetilde{O}(m\eps^{-1})$ runtimes stems from running a push-relabel-style flow algorithm to find an approximate flow.
Using recent almost linear time maximum flow algorithms~\cite{ChenKLPGS22,BrandCLPGSS23}, this can be improved to be $\widehat{O}(m)$, leading to an $\widehat{O}(\eps^{-3})$ update time algorithm for maintaining fractional matchings. See \cite[Lemma 5.2]{BernsteinGS20}.}
(\Cref{item:lazy:rebuild_cost}) for a total runtime of $\widetilde{O}(m\eps^{-4})$ and, therefore, an update time of $\widetilde{O}(\eps^{-4})$ for maintaining fractional matchings.
The technique they used to construct the fractional matchings they call \emph{congestion balancing}.
Furthermore, \cite{AssadiBD22} generalized this approach to non-bipartite graphs.
Their algorithm had $\widetilde{O}(2^{O(1/\eps^2)})$ rebuilds (\Cref{item:lazy:rebuild_count}) at a cost of $\widetilde{O}(m\eps^{-1})$ per rebuild (\Cref{item:lazy:rebuild_cost}) for a total runtime of $\widetilde{O}(m \cdot 2^{O(1/\eps^2)})$ and, therefore, an update time of $\widetilde{O}(2^{O(1/\eps^2)})$ for maintaining fractional matchings.
To obtain integral matchings, both papers applied known dynamic rounding algorithms~\cite{Wajc20} to solve \Cref{item:lazy:round_cost}.

\paragraph{Entropy Regularization and Weighted Matching.}
A key question motivating our result is, \emph{how powerful and general is the lazy update framework for decremental problems}?
Recent work of \cite{JambulapatiJST22} opened the door to studying this question.
This work showed that, for bipartite graphs, to bound the number of rebuilds (\Cref{item:lazy:rebuild_count}) it sufficed to set the fractional matchings to be sufficiently accurate solutions to natural, regularized optimization problems.
The current state-of-the-art decremental bipartite matching algorithms follow this framework~\cite{JambulapatiJST22,BhattacharyaKSW23rounding}. 

In this paper, we provide a broad generalization of this result.
We show that for any non-negative, non-degenerate, compact, downward closed, and convex $\X \subseteq \R^d_{\geq 0}$, to maintain approximate maximizers of $\bw^\top \bx$ for $\bx \in \X$ under deletions of coordinates to $\X$ one can apply the lazy update framework and rebuild only $\otilde(\eps^{-2})$ times!
Furthermore, we show that this rebuild count is obtained so long as the algorithm solves certain entropy-regularized versions of the problem, i.e., finding an $\bx \in \X$ approximately maximizing $\bw^\top \bx + \mu \cdot H(\bx)$ for some trade-off parameter $\mu$, where $H(\bx)$ is an appropriately weighted and scaled measure of the entropy of $\bx$. 

This result, when combined with \cite{ChenKLPGS22} and \cite{BhattacharyaKSW23rounding}, recovers the $\widehat{O}(\eps^{-2})$ update time of \cite{BhattacharyaKSW23rounding, JambulapatiJST22} and enables our main results on decremental matching in general graphs.
Letting $\X = \M_G$ be the non-bipartite matching polytope of the input graph $G$, this result implies that to obtain a fractional solution we simply need to repeatedly solve an entropy-regularized matching problem.
Additionally, our entropy regularization result immediately implies the same for maintaining a weighted fractional matching!
This leads to improved algorithms even for bipartite graphs, where the previous best algorithm has an $\widetilde{O}(\eps^{-3})$ update time~\cite{BhattacharyaKSW23dynamic2}.

However, to obtain efficient decremental fractional algorithms for general graphs, we still need to efficiently solve these entropy-regularized problems over the matching polytope (\Cref{item:lazy:rebuild_cost}).
In \cref{sec:decremental-fractional} we show how to solve these entropy-regularized problems to $(1-\delta)$ accuracy in $\widetilde{O}(m \cdot \delta^{-13})$ and $\widehat{O}(m \cdot \delta^{-5})$ time for all $\delta\geq \widetilde\Omega(n^{-1/2})$.
To maintain a $(1-\eps)$-approximate matching, it suffices to set $\delta = 1/\poly(\eps^{-1},\log n)$.
We obtain this by modularizing and generalizing the framework of \cite{AhnG14} for capacitated weighted general $\bb$-matching and provide two different instantiations of the framework.
In one approach we apply the recent convex flow algorithm of \cite{ChenKLPGS22}, leading to the running time of $\widehat{O}(m \cdot \delta^{-5})$.
In the other one we reduce entropy regularization to capacity-constrained weighted general matching, leading to the other running time of $\widetilde{O}(m \cdot \delta^{-13})$.

\paragraph{Dynamic Rounding.}
To turn our dynamic fractional matching results into dynamic integral matching results we develop two unweighted rounding algorithms for general graphs, one randomized and the other deterministic.
For the randomized algorithm, we analyze the standard sampling procedure for rounding in bipartite graphs and argue that, excitingly, it can also be used in general graphs.
On the other hand, our deterministic algorithm is based on a recently dynamized pipage rounding procedure~\cite{BhattacharyaKSW23rounding}.
However, we differ from \cite{BhattacharyaKSW23rounding} by stopping their algorithm earlier and running static algorithms on the result periodically. 
Although the randomized rounding algorithm does not have a strictly better update time than the deterministic one, it has stronger guarantees that we leverage to round weighted matchings. 
We remark that while \cite{AssadiBD22} also required rounding, the fractional matching they maintained has a special property (in particular it is $\poly(\eps)$-restricted~\cite[Definition 1]{BernsteinS16}) that allows them to argue the bipartite rounding algorithm of \cite{Wajc20} directly applies to the fractional matching they found.
In contrast, our rounding algorithms generically work for \emph{any} fractional matching in general graphs, and are the first to achieve such a guarantee.
(See \cref{subsec:related} for a more detailed comparison of rounding algorithms.)

\paragraph{Optimality of Entropy Regularization.} 
Given the utility of entropy regularization and the lazy update framework, we ask, \emph{can we further decrease the number of rebuilds (\Cref{item:lazy:rebuild_count})?}
Interestingly, we show that this is not the case for decremental matching, even on unweighted bipartite graphs.
We show that for any $n \geq 1$ and $\eps=\Omega(n^{-1/2})$, there is an adversarial choice of the initial graph with $n$ vertices and a sequence of deletions such that, regardless of what fractional matching the algorithm maintains, there are at least $\Omega(\log^2(\eps^2n)\cdot \eps^{-2})$ rebuilds.
This shows the optimality of entropy regularization in both $\log n$ and $\eps^{-1}$ factors for $\eps = \Omega(n^{-1/2+\delta})$ for constant $\delta > 0$.

\paragraph{Summary.}
We obtain improved results for decremental matching in general graphs by building new tools to follow the lazy update framework: we prove a general bound on the number of rebuilds when using entropy regularization (\Cref{item:lazy:rebuild_count}), we develop efficient algorithms for solving entropy-regularized optimization problems over the matching polytope (\Cref{item:lazy:rebuild_cost}), and we develop new dynamic rounding algorithms for fractional matchings in general graphs (\Cref{item:lazy:round_cost}).
We think that each tool could be of independent interest.
Additionally, given the generality and optimality (in terms of the number of rebuilds) of our entropy regularization approach to decremental problems, we hope that our results may open the door to new dynamic algorithms in broader settings. 

\paragraph{Paper Organization.}
In the remainder of this introduction we present our results in \cref{subsec:results} and compare and survey previous work in \cref{subsec:related}. We then cover preliminaries in \cref{sec:prelim} and provide a technical overview of our approach in \cref{sec:tech-overview}.
In \cref{sec:congestion-balancing} we show the robustness and optimality of entropy regularization for the decremental linear optimization problem.
We then turn our attention to the special case of decremental matching in non-bipartite graphs, presenting our algorithms for solving entropy-regularized matchings in \cref{sec:decremental-fractional} and our rounding algorithms in \cref{sec:rounding}. 
In \cref{appendix:entropy,appendix:rounding-proofs} we provide additional proofs that are included for completeness. 

\subsection{Results}\label{subsec:results}

In this paper we consider the \emph{decremental matching problem} formally defined as follows.
Note that in this definition, when $W = 1$, the problem is the aforementioned \emph{unweighted} matching problem in which the maintained matching is simply a $(1-\eps)$-approximate maximum (cardinality) matching.

\begin{problem}[Decremental Matching]
  In the \emph{decremental matching} problem, we are given an $n$-vertex $m$-edge graph $G = (V, E)$, integer\footnote{Note that the assumption that $\bw$ is integral is without loss of generality as it can be achieved by scaling.
  For instance, we can first make the minimum weight $1$, and then scale each entry to the nearest value of $\left\lceil (1+O(\eps))^{i}\right\rceil$ since we are only aiming for an approximate solution.} edge weights $\bw \in \N^E$ bounded by $W = \poly(n)$, and an accuracy parameter $\eps \in (0, 1)$.
  The goal is to maintain a $(1-\eps)$-approximate maximum weight matching $M \subseteq E$  at all times under deletions to $E$ until $G$ becomes empty.
  \label{prob:dec-matching}
\end{problem}

We develop a variety of randomized algorithms to solve \Cref{prob:dec-matching}.
Due to the use of randomization, it is important to distinguish between different kinds of update sequences that they can support.
We say a dynamic algorithm is \emph{output-adaptive} (respectively, \emph{fully-adaptive}) if it works for update sequences that are chosen adaptively based on the output (respectively, internal randomness) of the algorithm.
Note that a fully-adaptive algorithm is automatically output-adaptive.
These are in contrast to the \emph{oblivious} algorithms which only work for updated sequences that are fixed in advance.\footnote{It is also common in the dynamic algorithm literature to model the adaptiveness of an algorithm in the \emph{adversarial} setting in which the algorithm is working \emph{against} an \emph{adversary} designing the update sequence either on-the-fly (in which case it is an \emph{adaptive} adversary) or in advance (an \emph{oblivious} adversary). We also remark that the term \emph{adaptive} is widely used in the literature but has mixed meanings and can refer to either output-adaptive or fully-adaptive based on the context. Consequently, in this paper we make an explicit distinction between the two notions.}

We now describe our main results.
Recall that in the unweighted case, the previous state-of-the-art algorithms that solve \Cref{prob:dec-matching} are an $\widetilde{O}(1) \cdot 2^{O(1/\eps^{2})}$ update time algorithm of \cite{AssadiBD22} and an $O(\sqrt{m}\eps^{-2})$ update time algorithm of \cite{GuptaP13}.

\begin{restatable}[Unweighted Decremental Matching]{theorem}{MainUnweighted}
  There are randomized fully-adaptive $\widetilde{O}(\eps^{-\expnear})$ and $\widehat{O}(\eps^{-\expalmost})$ update time algorithms that solve \cref{prob:dec-matching} in the unweighted case w.h.p.
  \label{thm:main-unweighted}
\end{restatable}

For the weighted case, the algorithm of \cite{AssadiBD22} does not apply, and the previous state-of-the-art, even for bipartite graphs, is an $\widetilde{O}(\sqrt{m}\eps^{-O(1/\eps)})$ update time algorithm, also by \cite{GuptaP13}.
Our algorithms have near-optimal update time when either $\eps$ is a constant or the input graph is dense.

\begin{restatable}[Weighted Decremental Matching]{theorem}{MainWeighted}
  There is a randomized fully-adaptive $\widetilde{O}(\eps^{-O(1/\eps)})$ update time algorithm that solves \cref{prob:dec-matching} w.h.p.
  Additionally, when $m = \widetilde{\Theta}(n^2)$, there are randomized output-adaptive $\widetilde{O}(\eps^{-\expnear})$ and $\widehat{O}(\eps^{-\expalmost})$ update time algorithms.
  \label{thm:main-weighted}
\end{restatable}

For dense bipartite graphs we obtain an even better update time.

\begin{restatable}[Weighted Bipartite Decremental Matching]{theorem}{MainWeightedBipartite}
  For bipartite graphs with $m = \widehat{\Theta}(n^2)$ there is an output-adaptive randomized $\widehat{O}(\eps^{-6})$ update time algorithm that solves \cref{prob:dec-matching} w.h.p.
  \label{thm:main-weighted-bipartite}
\end{restatable}

\cref{thm:main-unweighted} to \ref{thm:main-weighted-bipartite}
are all obtained by solving the intermediate problem of decremental fractional matching (\cref{prob:dec-frac-matching}) and then converting the fractional results to integral via rounding algorithms.
Following the entropy regularization approach in the lazy update framework, the fractional matching algorithms we develop achieve an upper bound of $\widetilde{O}(\eps^{-2})$ on the number of rebuilds (\cref{item:lazy:rebuild_count}) in \cref{thm:upper bound on reconstruction rounds approximate multiple phases}.
We investigate whether this bound on the number of rebuilds can be improved.
Interestingly, we show that this is not possible and that entropy regularization is a near-optimal strategy for the lazy update framework in the decremental setting for any $n$ and $\eps= \Omega(n^{-1/2})$.
This lower bound holds even in the simple case of unweighted bipartite matching.

\begin{theorem}[informal, see \cref{thm:lower-bound}]
  For any $n \in \N$, $\eps \geq 2/\sqrt{n}$, and any output-adaptive algorithm implementing the rebuilding subroutine in the lazy update framework, there exists an unweighted bipartite graph of $n$ vertices on each side such that an output-adaptively chosen sequence of edge deletions causes $\Omega(\log^2(\eps^2 n) \cdot \eps^{-2})$ rebuilds.
\end{theorem}

\subsection{Related Work} \label{subsec:related}

Here we give a more extensive summary of previous work related to our results in this paper.

\paragraph{General Matching.}

Due to the existence of odd cycles and blossoms, matching problems in general graphs are often considerably more challenging than in bipartite graphs.
Starting from the blossom algorithm of Edmonds~\cite{Edmonds1965}, with additional ideas and techniques, several works culminated in general matching algorithms that are equally efficient as classic bipartite algorithms, both in the exact~\cite{MicaliV80,GabowT91,DuanPS17} and approximate~\cite{MicaliV80,DuanP14,AhnG14} settings (we make particular use of \cite{AhnG14} in our dynamic algorithms).
Yet, it is still open whether modern optimization-based algorithms for bipartite matching~\cite{Madry13,LeeS14,Madry16,BrandLNPSS0W20,KathuriaLS20,AxiotisMV21,ChenKLPGS22} can lead to runtime improvements for computing matchings in general graphs.

\paragraph{Incremental Matching.}
Another studied partially-dynamic matching problem is the \emph{incremental matching problem}, where instead of edge deletions, there are only edge insertions.
In this setting, near-optimal results are known for obtaining $(1-\eps)$-approximate matchings, achieving $\widetilde{O}(\poly(\eps^{-1}))$ and recently $O(\poly(\eps^{-1}))$ update times in bipartite graphs~\cite{Gupta14,BlikstadK23} and $\eps^{-O(1/\eps)}$ update time in general graphs~\cite{GrandoniLSSS19}.
For bipartite graphs, an $\widetilde{O}(\poly(\eps^{-1}))$ update time is also obtainable by a more general partially dynamic packing/covering LP algorithm \cite{BhattacharyaKSW23dynamic2}.

\paragraph{Additional Matching Results.}
Aside from $(1-\eps)$-approximations, dynamic matching with other approximation ratios has been studied, particularly in the fully dynamic setting.
Notably, for $1/2$-approximate, maximal matching, a line of work culminated into optimal, constant-update-time algorithms~\cite{OnakR10,BaswanaGS11,BhattacharyaHI15,Solomon16,BhattacharyaHN17,BhattacharyaCH17,BhattacharyaK19,BehnezhadDHSS19}.
The work of \cite{BernsteinS15} introduced the notion of edge-degree constrained subgraphs (EDCS) and with this initiated a line of work on non-trivial algorithms to maintain $(2/3-\eps)$ approximate matchings~\cite{BernsteinS16,GrandoniSSU22} and beyond~\cite{BehnezhadK22}.
Improvements in other directions such as derandomization, de-amortization, and frameworks converting unweighted results to the weighted case have also been studied~\cite{BhattacharyaHN16,StubbsW17,CharikarS18,BhattacharyaK19,BernsteinFH19,BernsteinDL21,BhattacharyaK21,Kiss22,RoghaniSW22,BhattacharyaKSW23rounding}.
Dynamic rounding algorithms which reduce dynamic integral matching algorithms to dynamic fractional matching algorithms are also well-studied in various regimes~\cite{ArarCCSW18,Wajc20,BhattacharyaK21,Kiss22,BhattacharyaKSW23rounding}.

\paragraph{Entropy Regularization.}
Even in prior dynamic matching results that do not explicitly use entropy regularization, entropy does play a role implicitly. More precisely,~\cite{Gupta14,BhattacharyaKS23,ZhengH23} all used the multiplicative weight update (MWU) method, which can be viewed as an iterative method for optimization method which uses an entropy regularizer to determine the steps it makes.

\paragraph{Comparison to \cite{AssadiBD22}.}
The $\widetilde{O}(1) \cdot 2^{O(1/\eps^2)}$ update time algorithm of \cite{AssadiBD22} used the congestion balancing approach introduced in \cite{BernsteinGS20} to implement rebuilds in the lazy update framework.
Informally, in congestion balancing, a capacity constraint $\bc \in [0, 1]^{E}$ is maintained, and in each iteration the goal is to find a matching $M_{\bc}$ respecting this capacity constraint of size comparable to the actual maximum matching or, in that case that such a matching does not exist, find a set of edges whose capacity constraint is ``critical'' to $M_{\bc}$ being small.
In the first case, $M_{\bc}$ is used as the output fractional matching until future deletions decrease its value significantly.
In the second case, the capacities of these critical edges are increased to accommodate larger matchings.
The exponential dependence on $\eps^{-1}$ of \cite{AssadiBD22}'s algorithm stems from the difficulty of solving the above subroutine in general graphs.
As we will show later in this paper, there is indeed an $\widetilde{O}(m \cdot \poly(\eps^{-1}))$ algorithm for the capacity-constrained matching problem in general graphs (see \cref{lemma:capacitated-weighted-matching}), but \cite{AssadiBD22} additionally needed a dual certificate (obtainable from, e.g., \cite{DuanP14}, on uncapacitated graphs) of the matching problem to identify critical edges.
Interestingly, obtaining such certificates is still left open by our work.\footnote{More precisely, \cite{AssadiBD22} ran the algorithm of \cite{DuanP14} only on a carefully sampled subgraph of $G$, and thus the dual certificate they obtained and used is different from that of the capacity-constrained matching problem in \cref{prob:capacitated-weighted-matching}. Interestingly, one can show that a dual certificate to \cref{prob:capacitated-weighted-matching} also suffices to identify critical edges, so this leaves extracting the dual from our algorithm the final step toward speeding up \cite{AssadiBD22}.}

\paragraph{Comparison to Previous Randomized Rounding Algorithms.}
Our randomized rounding algorithm adopts the same strategy central to previous algorithms~\cite{ArarCCSW18,Wajc20}.
Their algorithms build upon the subgraph, hereafter referred to as a \emph{sparsifier}, obtained by sampling each edge, either independently or dependently, with probability proportional to the fractional mass assigned to it.
Informally, \cite{ArarCCSW18} showed that the independently sampled sparsifier preserves\footnote{More precisely
the sparsifier contains an integral matching whose size is the same as the fractional one, up to $(1-\delta)$ multiplicatively for any $\delta$ (on which the algorithm's runtime depends).} approximately maximal fractional matchings which are $(1/2-\eps)$-approximate.
This holds in general graphs as well.
\cite{Wajc20}, on the other hand, studied fully-adaptive rounding, in which case updating the sparsifier only partially as in \cite{ArarCCSW18} no longer works.
\cite{Wajc20} therefore designed a \emph{dependent} sampling scheme that has marginal the same as the independent one, which turned out to be more efficiently sample-able from scratch.
Their analysis also showed that the sparsifier preserves $(1-\eps)$-approximate fractional matchings in \emph{bipartite} graphs.
\cite{AssadiBD22} then extended the analysis to show that \emph{$O(\eps)$-restricted} matchings in general graphs are preserved in \cite{Wajc20}'s sparsifier as well.
Finally, our rounding algorithm goes back and considers the independently sampled sparsifier which can now be used in the (output-)adaptive setting due to the dynamic sampler of \cite{BhattacharyaKSW23rounding}.
Adopting a more direct analysis, we show that this sparsifier can in fact preserve \emph{arbitrary} fractional matchings in general graphs.
Given that the sampling distribution of~\cite{Wajc20} is, in essence, an easier-to-sample-from but correlated version of the distribution we consider in this paper, we suspect their analysis can be extended to work for any fractional matching in general graphs, perhaps at the cost of a larger runtime.
After the initial publication of our manuscript, \cite{Dudeja24a} showed that indeed the rounding algorithm of \cite{Wajc20} can be extended to work in general graphs, albeit with a slightly larger update time.

\section{Preliminaries} \label{sec:prelim}

\paragraph{Notation.}
We let $[d] \defeq \{1, 2, \ldots, d\}$ for $d \in \N$, $\R_\infty \defeq \R \cup \{\infty\}$, and $\R_{-\infty} \defeq \R \cup \{-\infty\}$.
We use $\log(\cdot)$ to denote logarithm base $2$ and $\ln(\cdot)$ for natural logarithm.
We let $\llbracket \phi \rrbracket$ be evaluated to $1$ if the expression $\phi$ is true and $0$ otherwise.

Consider a finite set $U$ and $S \subseteq U$.
Let $\boldsymbol{0}^{S}$ and $\boldsymbol{1}^{S}$ be the all-zero and all-one vector in $\R^S$, respectively.\footnote{When $S$ is clear from context, we may drop the superscript and simply use $\boldsymbol{0}$ and $\boldsymbol{1}$.}
Let $\Delta^S$ denote the \emph{simplex} in $\R^S$, i.e., $\Delta^S \defeq \{\bx \in \R_{\geq 0}^{S}: \norm*{\bx}_1 = 1\}$.
For $\bx \in \R^{U}$, let $\bx_S \in \R^S$ be $\bx$ with coordinates restricted to $S$, i.e., $(\bx_S)_i = \bx_i$ for all $i \in S$.
For $\X \subseteq \R^{U}$, let $\X_S$ be $\X$ restricted to coordinates $S$, i.e., $\X_S \defeq \{\bx_S: \bx \in \X\}$, and $\X_{S,+} \defeq \{\bx \in \X_S: \bx_i > 0\;\forall\;i \in S\}$ be the subset of $\X_S$ in $\R_{>0}^S$.
For $\bx \in \R^S$ and $i \in U$, let $\bx \setminus \{i\} \in \R^{S \setminus \{i\}}$ be $\bx$ excluding coordinate $i$, i.e., $(\bx \setminus \{i\})_j = \bx_j$ for all $j \in S \setminus \{i\}$.
For $S^\prime \subseteq S$ and $\bx \in \R_{S^\prime}$, let $\bx^S \in \R^{S}$ be $\bx$ extended to $R^{S}$, i.e., $\bx^S_i = \bx_i$ for $i \in S^\prime$ and $\bx^S_i = 0$ for $i \in S \setminus S^\prime$.

\paragraph{Runtimes.}
In this paper we use the standard word-RAM model where basic arithmetic operations over $O(\log{n})$-bit words can be performed in constant time.
When the input size $n$ is clear from context, we say an $x \in \R$ is \emph{polynomially bounded} if $|x| \in \{0\} \cup [n^{-O(1)}, n^{O(1)}]$. 

\paragraph{Graphs.}
All graphs in this paper are undirected,
simple, and not necessarily bipartite, unless stated otherwise.
For a graph $G = (V, E)$, let $E_v \defeq \{e \in E: v \in e\}$ be the set of edges incident to $v \in V$, and $E[B] \defeq \{e \in E: e \subseteq B\}$ be the set of edges whose endpoints are both in $B \subseteq V$.
For $F \subseteq E$ let $F_v \defeq F \cap E_v$ and $F[B] \defeq F \cap E[B]$.
For $\bx \in \R_{\geq 0}^{E}$, let $\bx(v) \defeq \sum_{e \in E_v}\bx_e$ for $v \in V$ and $\bx(B) \defeq \sum_{e \in E[B]}\bx_e$ for $B \subseteq V$.
Let $\bB_G \in \{-1, 0, 1\}^{E \times V}$ be the \emph{(edge vertex) incidence matrix} of $G$, where there are exactly two non-zero entries per row $e$, one at entry $(e, u)$ with value $-1$ and the other at entry $(e, v)$ with value $1$, for an arbitrary orientation $(u, v)$ of $e = \{u, v\}$.
Let $M^{*}(G)$ be the size of the maximum matching in $G$ and $M^{*}_{\bw}(G)$ for $\bw \in \R_{\geq 0}^{E}$ be the value of the maximum weight matching in $G$ for the weights $\bw$.

\paragraph{Matching Polytope.}
For an undirected graph $G = (V, E)$, the \emph{matching polytope} of $G$ is the convex hull of the indicator vectors of matchings in $G$.
Let
\begin{equation}
\P_{G} \defeq
\left\{
\begin{array}{ll}
\sum_{e \in E_v}x_e \leq 1, & \forall\;v \in V \\
\end{array}
\right\} \cap \R_{\geq 0}^E.
\end{equation}
It is a standard fact that when $G$ is bipartite, $\P_G$ is the matching polytope of $G$.
When $G$ is non-bipartite, we need to further consider odd-set constraints.
Formally, we let 
\[
\O_G \defeq \left\{B \subseteq V: |B| \geq 3\;\text{and $|B|$ is odd}\right\}
\]
be the collections of \emph{odd sets} and define
\begin{equation}
\M_{G} \defeq
\P_G \cap \left\{
\begin{array}{ll}
\bx(B) \leq \left\lfloor\frac{|B|}{2}\right\rfloor, & \forall\;B \in \O_G \\
\end{array}
\right\}.
\label{eq:matching-polytope}
\end{equation}
It is known that $\M_G$ is the matching polytope of $G$~\cite{schrijver2003combinatorial}.
We often consider the relaxation of \eqref{eq:matching-polytope} to only contain small odd sets $\O_{G,\eps} \defeq \{B \in \O_G: |B| \leq 1/\eps\}$ denoted by
\begin{equation}
\M_{G,\eps} \defeq
\P_G \cap \left\{
\begin{array}{ll}
\bx(B) \leq \left\lfloor\frac{|B|}{2}\right\rfloor, & \forall\;B \in \O_{G,\eps} \\
\end{array}
\right\}
\label{eq:matching-polytope-small}
\end{equation}
when dealing with $(1-\eps)$-approximatation algorithms.
The following fact about $M_{G}$ versus $M_{G,\eps}$ is folklore and key to our algorithm development.

\begin{fact}[see, e.g., \cite{AssadiBD22}]
  For $\eps > 0$ and $\bx \in \M_{G,\eps}$ it holds that $\frac{\bx}{1+3\eps} \in \M_G$.
  \label{fact:scaled-down}
\end{fact}

We may refer to an $\bx \in \M_G$ as a \emph{(fractional) matching} in $G$ and $\bx \in \P_G$ as a \emph{relaxed (fractional) matching}\footnote{We remark that an $\bx \in \P_G$  is often referred to as a \emph{fractional matching} even in general graphs in the literature (see, e.g., \cite{ArarCCSW18,BhattacharyaKSW23rounding}). However, we deviate from this convention so that there is no integrality gap between fractional and integral matchings.} in $G$.
When $\bx \in \M_G$ or $\bx \in \P_G$ is clear from context, we may refer to $\bx_e$ as the \emph{mass} the matching $\bx$ puts on edge $e$.

\paragraph{Miscellaneous.}
The \emph{recourse} of a dynamic algorithm is the total number of changes it makes to its output.
When working with $(1-\eps)$-approximation, in the remainder of the paper we may assume without loss of generality that $\eps$ is upper-bounded by an explicit constant.
This only incurs a constant increase in runtimes.

\section{Technical Overview}
\label{sec:tech-overview}

In this section we introduce the problems that we consider in this paper, present the main results and technical tools of each section, and illustrate high-level ideas towards proving them.
At the end of this section, we utilize these results and tools to prove our main theorems stated in \cref{subsec:results}.

\subsection{Lazy Updates for Decremental Linear Optimization} \label{sec:tech:lazy}

In this paper we consider a unifying framework of congestion balancing for solving  \emph{decremental linear optimization problems} formally defined as follows.

\begin{restatable}[Decremental Linear Optimization]{problem}{DecLinearOpt}
	\label{prob:dec_linear}
	In the \emph{decremental linear optimization} problem we are given a positive weight vector $\bw \in \R_{> 0}^d$ and a non-negative, non-degenerate\footnote{That is, for each coordinate $i\in[d]$ there exists an $\bx\in\X$ such that $\bx_i>0$. This is a natural assumption since we can always ignore the degenerate dimensions.}, compact\footnote{Note that this is equivalent to being bounded and closed, which implies $\max_{\bx \in \X}\bw^\top \bx$ is bounded as well.}, downward closed\footnote{That is, for each $x \in \X$ and $y \in \R_{\geq 0}^{d}$ with $y \leq x$ entry-wise, we have $y \in \X$ also.}, convex $\mathcal{X} \subseteq \R^d_{\geq 0}$ where $d > 1$.\footnote{The assumption of $d > 1$ is natural as otherwise the problem degenerates into a 1D optimization and becomes not decremental in essense.}
    Starting from the entire coordinate set $S = [d]$,
    under a sequence of deletions of coordinates from $S$ we must maintain an $\bx \in \X_S$ such that 
	\[
	\bw_S^\top \bx
	\geq (1 - \eps) \max_{\bx^\prime \in \X_S} \bw_S^\top \bx^\prime
	\]
    for a given accuracy $\eps > 0$ until $S$ becomes empty.
\end{restatable}

The framework for solving \cref{prob:dec_linear} that we study in this paper is a generalization of the lazy update scheme
that is widely used for dynamic matching problems as we discussed in \cref{sec:intro}.
Specifically, we consider algorithms that maintain an approximate solution $\bx \in \X_S$ and use it as the solution until its value drops by an $O(\eps)$ fraction at which point we perform a rebuild.
The following \cref{alg:meta-general-ds} is a template for the lazy update approaches, for which we will later specify what approximate solutions $\bx$ will be used in $\texttt{Rebuild()}$ at Line~\ref{line:recomputation}.

\begin{algorithm2e}[!ht]
  \caption{\textsc{LazyUpdate($\X$)}} \label{alg:meta-general-ds}
  
  \SetEndCharOfAlgoLine{}
  \SetKwInput{KwData}{Input}
  \SetKwInput{KwResult}{Output}
  \SetKwInOut{State}{global}
  \SetKwProg{KwProc}{function}{}{}
  \SetKwFunction{Initialize}{Initialize}
  \SetKwFunction{Delete}{Delete}
  \SetKwFunction{Rebuild}{Rebuild}

  \State{weight vector $\bw \in \R_{> 0}^{d}$ and accuracy parameter $\eps \in (0, 1)$.}
  \State{current coordinates $S \subseteq [d]$ and solution $\bx \in \X_S$ with ``rebuild'' value $\nu$.}
  \State{number of rebuilds $t \in \Z_{\geq 0}$.}
  \State{snapshots $x^{(t)}$ and $S^{(t)}$ for analysis.}

  \vspace{0.4em}
  
  \KwProc{\Initialize{$\bw \in \R^{d}_{>0}, \eps \in (0, 1)$}} {
    Save $\bw$ and $\eps$ as global variables.\;
    Initialize $S \gets [d]$ and $t \gets 0$.\;
    $(\bx, \nu) \gets \texttt{Rebuild()}$.\;
  }
  
  \vspace{0.4em}

  \KwProc{\Delete{$i \in [d]$}} {
    Set $\bx \gets \bx \setminus \{i\}$ and $S \gets S \setminus \{i\}$.\;
    \lIf{$\bw_S^\top \bx < \left(1 - \frac{\eps}{2}\right) \nu$} {
      $(\bx, \nu) \gets \texttt{Rebuild()}$.
    }
  }

  \vspace{0.4em}

  \KwProc{\Rebuild{}} {
    \tcp{What $\bx^{(t)}$ is computed below depends on the specific algorithm}
    Set $\bx^{(t)}$ to an element in $\X_S$ with $\bw^\top_S \bx^{(t)} \geq \left(1 - \frac{\eps}{2}\right)\max_{\bx^\prime \in \X_S}{\bw}^\top_S \bx^\prime$. \label{line:recomputation}\;
    Set $\nu^{(t)} \gets \bw^\top_S \bx^{(t)}$ and $S^{(t)} \gets S$.\;
    \textbf{return} $(\bx^{(t)}, \nu^{(t)})$ and set $t \gets t + 1$.\;
  }
\end{algorithm2e}

\begin{observation}
  \cref{alg:meta-general-ds} solves the decremental linear optimization problem (\cref{prob:dec_linear}).
\end{observation}

\begin{proof}
Note that $\bx \in \X_S$ at all times since $\X$ is downward closed. 
The vector $\bx^{(t)}$, when constructed in $\texttt{Rebuild()}$, is an $\left(1-\frac{\eps}{2}\right)$-approximate solution with value $\nu^{(t)}$.
Since the $\texttt{Delete}$ operations only decrease $\max_{\bx^\prime \in \X_S}\bw_S^\top \bx^\prime$, as long as $\bw_S^\top \bx \geq \left(1-\frac{\eps}{2}\right)\nu^{(t)}$, we have
\[
\bw_S^\top \bx \geq \left(1-\frac{\eps}{2}\right)\nu^{(t)} \geq (1-\eps)\max_{\bx^\prime \in \X_S}\bw_S^\top \bx^\prime.
\]
On the other hand, whenever $\bw^\top_S \bx < \left(1-\frac{\eps}{2}\right)\nu^{(t)}$, we call $\texttt{Rebuild()}$.
\cref{alg:meta-general-ds} thus maintains a $(1-\eps)$-approximate solution at all times.
\end{proof}

\subsection{Entropy Regularization}

To solve the above decremental linear optimization problem using the lazy update scheme, we apply a variant of the entropy regularization strategy previously used decremental dynamic matching in unweighted bipartite graphs~\cite{JambulapatiJST22}.
Intuitively, to avoid the adversary from deleting large weight from our solution at once, the idea is to find an $\bx \in \X_S$ with uniformly distributed value on each coordinate.
As such, the approach is to use the entropy-regularized solution as $\bx$ in $\texttt{Rebuild()}$, prioritizing vectors with higher entropy when they have similar weights.

To formally describe our results, consider a fixed positive weight $\bw \in \R_{>0}^{d}$.
For $S \subseteq [d]$, $\mu,\gamma \in \R_{> 0}$, 
we define our \emph{entropy-regularized objective} $f_{S,\gamma}^{\mu}: \X_S \to \R$  for all $\bx\in \X_S$ by
\begin{equation}
f_{S,\gamma}^{\mu}(\bx)\defeq \bw^{\top}_S\bx+\mu \cdot \sum_{i\in S}\bw_{i}\bx_{i}\log \frac{\gamma}{\bw_i\bx_{i}}
\label{eq:entropy-objective}
\end{equation}
and let
\begin{equation}
\bx^{\mu}_{S,\gamma} \defeq \argmax_{\bx \in \X_S}f^{\mu}_{S,\gamma}(\bx)
\label{eq:def-x}
\end{equation}
be the optimal solution to \eqref{eq:entropy-objective}.
The main result we show later in \cref{sec:congestion-balancing} is that solutions to \eqref{eq:entropy-objective} with $\mu = \widetilde{\Theta}(\eps)$ lead to a lazy update scheme with bounded rebuilds.

\begin{restatable}{lemma}{UpperBoundReconstruction}\label{lemma:upper bound on reconstruction rounds}
    For any $\alpha \geq 0$, accuracy parameter $\eps > 0$, and $0 < \mu\leq \frac{\eps}{8\log d}$,
    if the subroutine \emph{\texttt{Rebuild()}} in \cref{alg:meta-general-ds} returns $\bx_{S^{(t)},\gamma}^{\mu}$ as $\bx^{(t)}$ for $\gamma \defeq \alpha$, then \emph{\texttt{Rebuild()}} will be called at most
    $O(\frac{\log d}{\mu\cdot \eps})$ times before $\max_{\bx^\prime\in \X_S}\bw_S^\top \bx^\prime$ drops from at most $\alpha$ to below $\alpha/d$.
\end{restatable}

Analogous to \cite{JambulapatiJST22}, to prove \Cref{lemma:upper bound on reconstruction rounds}, we use $f_{S,\gamma}^{\mu}(\bx_{S,\gamma}^{\mu})$, the optimal value of the entropy-regularized objective on $\X_S$, as a potential function to capture the progress that \cref{alg:meta-general-ds} makes. Applying the optimality conditions for concave optimization to $f_{S,\gamma}^{\mu}(\bx_{S,\gamma}^{\mu})$ allows us to lower bound the decrease from $f_{S,\gamma}^{\mu}(\bx_{S,\gamma}^{\mu})$ to $f_{S^\prime,\gamma}^{\mu}(\bx_{S^\prime,\gamma}^{\mu})$ using the Bregman divergence of the entropy regularizer, which has a close relationship to the weighted value of the deleted coordinates in $S\setminus S^\prime$. The choice of $\mu$ and $\gamma$ is to guarantee that the entropy-regularized solution serves as a valid solution for $\texttt{Rebuild()}$.

In case the exact solution $\bx_{S,\gamma}^{\mu}$ is computationally expensive to compute, we also show that an accurate enough approximation to it admits the same robustness property, by combining the proof of \cref{lemma:upper bound on reconstruction rounds} and the strong concavity of the entropy-regularized objective. 
In the following, we say an $\bx \in \X_{S}$ is a \emph{$(1-\delta)$-approximate solution to $f_{S,\gamma}^{\mu}$} if $f_{S,\gamma}^{\mu}(\bx) \geq (1-\delta) f_{S,\gamma}^{\mu}(\bx_{S,\gamma}^{\mu})$.

\begin{restatable}{lemma}{UpperBoundReconstructionApprox}
  For any $\alpha \geq 0$, accuracy parameter $\eps > 0$, and $0 < \mu\leq \frac{\eps}{128\log d}$,
  if the subroutine \emph{\texttt{Rebuild()}} in \cref{alg:meta-general-ds} returns any $(1-\frac{\mu\eps^2}{512})$-approximate solution to $f_{S^{(t)},\gamma}^{\mu}$ as $\bx^{(t)}$ for $\gamma \defeq \alpha$, then \emph{\texttt{Rebuild()}} will be called at most
  $O(\frac{\log d}{\mu\cdot \eps})$ times
  before $\max_{\bx^\prime\in \X_S}\bw_S^\top \bx^\prime$ drops from at most $\alpha$ to below $\alpha/d$.
  \label{lemma:upper bound on reconstruction rounds approximate}
\end{restatable}

In the general case that $\max_{\bx^\prime \in \X_S}\bw_S^\top \bx^\prime$ drops by more than a factor of $d$, we can simply re-run the algorithm with different values of $\gamma$. \cref{alg:rebuild} below implements this strategy and \cref{thm:upper bound on reconstruction rounds approximate multiple phases} bounds its performance when used as the $\texttt{Rebuild()}$ subroutine in  \cref{alg:meta-general-ds}.

\begin{algorithm2e}[!ht]
  \caption{Implementation of \texttt{Rebuild()} for \cref{thm:upper bound on reconstruction rounds approximate multiple phases}.} \label{alg:rebuild}
  
  \SetEndCharOfAlgoLine{}
  \SetKwInput{KwData}{Input}
  \SetKwInput{KwResult}{Output}
  \SetKwInOut{State}{global}
  \SetKwProg{KwProc}{function}{}{}
  \SetKwFunction{Rebuild}{Rebuild}

  \State{weight vector $\bw \in \R_{> 0}^{d}$, accuracy parameter $\eps \in (0, 1)$, and $\mu \in (0, 1)$.}
  \State{current coordinates $S \subseteq [d]$.}
  \State{number of rebuilds $t \in \Z_{\geq 0}$ and snapshots $x^{(t)}$ and $S^{(t)}$ for analysis.}
  \State{an estimate $\widetilde{\nu}$ for the current phase, initially set to $\alpha$}

  \vspace{0.4em}

  \KwProc{\Rebuild{}} {
    Let $\bx^{(t)} \in \X_S$ be an arbitrary $\left(1-\frac{\mu\eps^2}{512}\right)$-approximate solution to $f_{S,\gamma}^{\mu}$ for $\gamma \defeq \widetilde{\nu}$.\;
    \While{$\bw^\top \bx^{(t)} < \widetilde{\nu}/d$} {\label{line:while-loop}
      $\widetilde{\nu} \gets \frac{1}{1-\eps}\widetilde{\nu}/d$.\label{line:decrease}\;
      Recompute $\bx^{(t)} \in \X_S$ as a $\left(1-\frac{\mu\eps^2}{512}\right)$-approximate solution to $f_{S,\gamma}^{\mu}$ for $\gamma \defeq \widetilde{\nu}$.\;
    }
    Set $\nu^{(t)} \gets \bw^\top_S \bx^{(t)}$ and $S^{(t)} \gets S$.\;
    \textbf{return} $(\bx^{(t)}, \nu^{(t)})$ and set $t \gets t + 1$.\;
  }
\end{algorithm2e}

\begin{theorem}
For parameters $\alpha$, $k$, $\eps > 0$, and $0 < \mu\leq \frac{\eps}{128\log d}$, using \cref{alg:rebuild} as the \emph{\texttt{Rebuild()}} subroutine in \cref{alg:meta-general-ds}, before $\max_{\bx^\prime \in \X_S}\bw^\top_S\bx^\prime$ drops from at most $\alpha$ to below $\alpha/k$ there will be at most $O(\frac{\log k}{\mu \cdot \eps})$ calls to \emph{\texttt{Rebuild()}}.
Moreover, the value $\gamma$ in \emph{\texttt{Rebuild()}} satisfies $\max_{\bx^\prime \in \X_S}\bw^\top_S \bx^\prime \leq \gamma \leq \alpha$, and the while-loop in Line~\ref{line:while-loop} will be run at most $O(\log_d k)$ times in total.
  \label{thm:upper bound on reconstruction rounds approximate multiple phases}
\end{theorem}

\begin{proof}
Let us assume $\eps \leq 1/3$.
Consider dividing \cref{alg:meta-general-ds} into phases, where each phase ends when $\widetilde{\nu}$ is decreased in Line~\ref{line:decrease}.
At the start of each phase, we maintain the invariant that $\widetilde{\nu} \geq \max_{\bx^\prime \in \X_S}\bw_S^\top \bx^\prime$.
The invariant implies that the number of phases is $O(\log_d k)$.

Within each phase, as long as $\bw_S^\top \bx^{(t)} \geq \widetilde{\nu}/d$, we know that $\max_{\bx^\prime \in \X_S}\bw_S^\top \bx^\prime \geq \widetilde{\nu} / d$.
\cref{lemma:upper bound on reconstruction rounds approximate} therefore shows that the vector $\bx$ we maintain throughout this phase is a valid approximation.
On the other hand, when $\bw^\top \bx^{(t)} < \widetilde{\nu} / d$, we know that $\max_{\bx^\prime \in \X_S}\bw_S^\top \bx^\prime$ must fall below $\frac{\widetilde{\nu} / d}{1-\eps} \leq \frac{3\widetilde{\nu}}{4}$.
Hence, the new $\widetilde{\nu}$ remains an upper bound on it and the invariant holds.
\cref{lemma:upper bound on reconstruction rounds approximate} also shows that within each phase the subroutine $\texttt{Rebuild()}$ will be called at most $O\left(\frac{\log d}{\mu \cdot \eps}\right)$ times.
As such throughout the $O(\log_d k)$ phases the number of rebuilds is at most $O\left(\frac{\log k}{\mu \cdot \eps}\right)$.
\end{proof}

To complement our algorithmic results, we investigate the limit of the lazy update framework. It turns out that entropy regularization is a nearly optimal strategy for the lazy update framework in its dependence on both $\log n$ and $\eps^{-1}$ for the decremental matching problem for any given $n$ and $\eps= \Omega(n^{-1/2})$, even in unweighted bipartite graphs.

\begin{restatable}{theorem}{LowerBound}
  For any $n \in \N$, $2/\sqrt{n}\leq \eps\leq 1/4$, and implementation of \emph{\texttt{Rebuild()}} in \cref{alg:meta-general-ds}, there exists a bipartite graph $G$ with $n$ vertices on each side and an output-adaptively chosen sequence of deletions that when $\X \defeq \P_G$ \cref{alg:meta-general-ds}  calls \emph{\texttt{Rebuild()}} $\Omega(\log^2(\eps^2 n) \cdot \eps^{-2})$ times before $G$ is empty.
  \label{thm:lower-bound}
\end{restatable}

\begin{figure}[ht]
\centering
\begin{tikzpicture}[x=0.75pt,y=0.75pt,yscale=-1,xscale=1]

\draw  [fill={rgb, 255:red, 0; green, 0; blue, 0 }  ,fill opacity=0.5 ] (130,130.5) -- (211,211.5) -- (130,211.5) -- cycle ;
\draw   (211.08,211.42) -- (130,130.5) -- (211,130.42) -- cycle ;
\draw  [fill={rgb, 255:red, 80; green, 227; blue, 194 }  ,fill opacity=0 ][dash pattern={on 3.75pt off 3pt on 7.5pt off 1.5pt}] (121.09,119.37) .. controls (121.65,118.8) and (122.57,118.8) .. (123.14,119.36) -- (221.04,216.34) .. controls (221.61,216.91) and (221.61,217.82) .. (221.05,218.39) -- (217.98,221.49) .. controls (217.42,222.06) and (216.5,222.06) .. (215.93,221.5) -- (118.03,124.52) .. controls (117.46,123.96) and (117.46,123.04) .. (118.02,122.47) -- cycle ;

\draw   (411.08,210.42) -- (330,129.5) -- (411,129.42) -- cycle ;
\draw  [fill={rgb, 255:red, 0; green, 0; blue, 0 }  ,fill opacity=0.5 ] (330.08,129.42) -- (411.08,210.42) -- (330.08,210.42) -- cycle ;
\draw  [fill={rgb, 255:red, 255; green, 255; blue, 255 }  ,fill opacity=1 ][dash pattern={on 3.75pt off 3pt on 7.5pt off 1.5pt}] (321.09,118.37) .. controls (321.65,117.8) and (322.57,117.8) .. (323.14,118.36) -- (421.04,215.34) .. controls (421.61,215.91) and (421.61,216.82) .. (421.05,217.39) -- (417.98,220.49) .. controls (417.42,221.06) and (416.5,221.06) .. (415.93,220.5) -- (318.03,123.52) .. controls (317.46,122.96) and (317.46,122.04) .. (318.02,121.47) -- cycle ;
\draw   (330.08,210.42) -- (330.08,129.42) -- (411,129.42) ;
\draw   (411,129.42) -- (410.77,210.64) -- (330.08,210.42) ;
\draw    (230,180) -- (308,180) ;
\draw [shift={(310,180)}, rotate = 180.72] [color={rgb, 255:red, 0; green, 0; blue, 0 }  ][line width=0.75]    (10.93,-3.29) .. controls (6.95,-1.4) and (3.31,-0.3) .. (0,0) .. controls (3.31,0.3) and (6.95,1.4) .. (10.93,3.29)   ;

\draw (270,165) node   [align=left] {\begin{minipage}[lt]{45.56pt}\setlength\topsep{0pt}
Deletions
\end{minipage}};
\end{tikzpicture}
\caption{Illustration of the adversarial strategy for proving the lower bound.}
\label{fig:lower bound}
\end{figure}
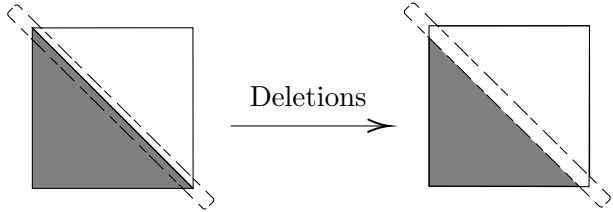

To give an intuition of the construction, \cref{fig:lower bound} shows the choice of the initial graph and the output-adaptive deletion strategy. We use the square to represent the adjacency matrix of the chosen bipartite graph, where the gray area represents edges between row vertices and column vertices. In this graph, the only maximum matching is the one on the diagonal of the square. We generalize this observation to any $(1-O(\eps))$-approximate fractional maximum matching and show that a constant fraction of the mass will be concentrated around the diagonal. Thus edge deletion near the diagonal suffices to cause rebuilds of the algorithm.
We show that within each phase, which is defined by deleting everything in the dashed rectangle, an output-adaptive deletion sequence can cause $\Omega(\log(\eps^2n)/\eps)$ calls to \texttt{Rebuild()}.
Additionally, there are $\Omega(\log(\eps^2n)/\eps)$ such phases in total, which together establish the lower bound.

\subsection{Dynamic Fractional Matching}

We obtain our algorithms for \cref{prob:dec-matching} by first solving the following intermediate fractional version of the problem.

\begin{problem}[Decremental Fractional Matching]
  In the \emph{decremental fractional matching} problem, we are given the same set of inputs as \cref{prob:dec-matching} does, and the goal is to maintain a $(1-\eps)$ approximate maximum weight \emph{fractional} matching, i.e., an $\bx \in \M_G$ such that $\sum_{e \in E}\bw_e\bx_e \geq (1-\eps)M_{\bw}^{*}(G)$, at all times under deletions to $E$ until $G$ becomes empty.
  \label{prob:dec-frac-matching}
\end{problem}

Our dynamic algorithms use the framework established in the preceding sections, and in particular they apply \cref{thm:upper bound on reconstruction rounds approximate multiple phases} with $\X = \M_G$.
This suffices to solve \cref{prob:dec-frac-matching} provided we can solve the following entropy-regularized problem efficiently.

\begin{restatable}[Entropy-Regularized Matching]{problem}{EntropyMatching}
  Given an $n$-vertex $m$-edge graph $G = (V, E)$ with edge weights $\bw \in \N^{E}$ bounded by $W = \poly(n)$, a trade-off parameter $\mu \leq 1$, a $\gamma \in \R_{\geq 0}$ such that $M_{\bw}^{*}(G) \leq \gamma \leq n^3 W$, and an accuracy parameter $\eps \in (0, 1)$, the \emph{entropy-regularized matching problem} is to compute a $(1-\eps)$-approximate solution $\bx \in \M_G$ to
  \begin{equation}
    \max_{\bx \in \M_G}\left\{\bw^\top \bx + \mu \cdot \sum_{e \in E}\bw_e\bx_e\log\frac{\gamma}{\bw_e\bx_e}\right\}.
    \label{eq:entropy-matching}
  \end{equation}
  \label{problem:entropy-matching}
\end{restatable}

We develop two algorithms for solving the entropy-regularized matching problem as specified in the following theorem.

\begin{restatable}{theorem}{EntropySolver}
  For any $\eps \geq \widetilde{\Omega}(n^{-1/2})$, there are randomized $\widetilde{O}\left(m\eps^{-6} + n\eps^{-13}\right)$ and $\widehat{O}\left(m\eps^{-5}\right)$ time algorithms that solve \cref{problem:entropy-matching} w.h.p.
  \label{thm:entropy-solver}
\end{restatable}

This when combined with \cref{thm:upper bound on reconstruction rounds approximate multiple phases} immediately implies algorithms for solving the decremental fractional matching problem.

\begin{restatable}{theorem}{DecrementalFractional}
  There are fully-adaptive randomized algorithms that, for $\eps \geq \widetilde{\Omega}(n^{-1/6})$, solve \cref{prob:dec-frac-matching} w.h.p. with amortized update times $\widetilde{O}(\eps^{-\expnear})$ and $\widehat{O}(\eps^{-\expalmost})$.
  Additionally, if $G$ is bipartite, then there is such an algorithm that works for any $\eps \geq 1/\poly(n)$ with amortized update time $\widehat{O}(\eps^{-2})$.
  The recourse of the algorithms is $\widetilde{O}(m\eps^{-2})$.
  \label{thm:decremental-fractional}
\end{restatable}

\begin{proof}
  Letting $\X \defeq \M_G$ for the decremental matching problem, \cref{thm:upper bound on reconstruction rounds approximate multiple phases} shows that by using \cref{problem:entropy-matching} with $\mu \defeq \frac{\eps}{128\log{m}}$ and accuracy parameter $\eps^\prime \defeq \frac{\mu\eps^2}{512} = \widetilde{\Theta}(\eps^3)$ inside \cref{alg:rebuild} as the subroutine \texttt{Rebuild()}, there will be at most $\widetilde{O}(\eps^{-2})$ calls to $\texttt{Rebuild()}$ throughout the algorithm, as the weight of any non-empty matching is at most $nW$ and at least one.
  For general graphs, we have $\eps^\prime \geq \widetilde{\Omega}(\eps^{-1/2})$, and the update times of $\widetilde{O}(\eps^{-\expnear})$ and $\widehat{O}(\eps^{-\expalmost})$ can be obtained by running \cref{thm:entropy-solver} as the subroutine $\texttt{Rebuild()}$ and amortizing over the $m$ updates.
  For bipartite graphs, we can use \cite[Theorem 10.16]{ChenKLPGS22} to solve \cref{problem:entropy-matching} to high accuracy in $\widehat{O}(m)$ time, resulting in the amortized update time of $\widehat{O}(\eps^{-2})$.
  The recourse is $\widetilde{O}(m\eps^{-2})$ since there are at most $\widetilde{O}(\eps^{-2})$ different matchings.
\end{proof}

Our algorithms for \cref{thm:entropy-solver} build upon the MWU-based algorithm for weighted non-bipartite $\bb$-matching by \cite{AhnG14}.
Informally, \cite{AhnG14} showed that the weighted $\bb$-matching problem in non-bipartite graphs reduces to solving a sequence of the same problem in bipartite graphs, possibly with different weights.
We observe that the analysis of \cite{AhnG14} seamlessly extends to general concave objective optimization over the non-bipartite matching polytope.
We then leverage the recent almost-linear time convex flow algorithm of \cite{ChenKLPGS22} for our almost-linear time algorithm for entropy-regularized matching.
Alternatively, by approximating the concave weight with piecewise linear functions and splitting each edge into multiple copies, we reduce \cref{problem:entropy-matching} to a capacity-constrained maximum weight matching problem, which is then solved by similarly applying the generalized approach of \cite{AhnG14}.
The runtime of this algorithm does not have the subpolynomial factor incurred by the use of \cite{ChenKLPGS22} but suffers from a larger dependence on $\eps^{-1}$.

\subsection{Dynamic Rounding of Fractional Matchings} \label{subsec:overview:rounding}

The results in previous sections show that we can solve the decremental fractional matching problem adaptively.
To turn the fractional matching into an integral one, we further design dynamic rounding algorithms for general graphs.
To present our algorithms in a unified way, we consider the following weighted definition of dynamic rounding algorithms.

\begin{restatable}[Dynamic Rounding Algorithm]{definition}{DefRounding}
  A \emph{dynamic rounding algorithm}, for a given $n$-vertex graph $G = (V, E)$, edge weights $\bw \in \N^{E}$ bounded by $W = \poly(n)$, and accuracy parameter $\eps > 0$, initializes with an $\bx \in \M_G$ and must maintain an integral matching $M \subseteq \supp(\bx)$ with $\bw(M) \geq (1-\eps)\bw^\top \bx$ under entry updates to $\bx$ that guarantee $\bx \in \M_G$ after each operation.
  \label{def:rounding}
\end{restatable}

We prove the following deterministic rounding algorithm which has near-optimal overhead in the unweighted case.

\begin{restatable}{theorem}{Rounding}
  There is a deterministic dynamic rounding algorithm for general graphs with amortized update time $\widetilde{O}(W\eps^{-4})$.
  \label{thm:rounding}
\end{restatable}

Our algorithm for \cref{thm:rounding} builds on top of the pipage-rounding algorithm recently dynamized by \cite{BhattacharyaKSW23rounding} for bipartite graphs.
Their algorithm circumvents the inherent barrier of the periodic-recomputation approaches by directly rounding to integral matchings without creating an intermediate sparsifier.
Although this approach does not generalize to non-bipartite graphs due to odd-set constraints, we observe that terminating their algorithm early in fact creates a good sparsifier for general graphs.
The main intuition is that the first few rounds of their algorithm only have small additive effects on the value $\bx_e$, and perturbing each edge slightly indeed does not have a huge impact on odd-set constraints.
While \cref{def:rounding} is weighted, we remark that the algorithm of \cref{thm:rounding} is essentially unweighted in its design, hence the linear dependence on $W$.
We thus apply it in the unweighted case or through reductions to the case when $W$ is small.
See \cref{subsec:everything} for more details.

We further consider rounding with better dependence on $W$.
For this we directly adopt the standard sampling approach of creating a matching sparsifier.
Similar approaches were studied before in, e.g., \cite{ArarCCSW18,Wajc20,BhattacharyaKSW23rounding}, for rounding matchings in bipartite graphs or certain structured matchings in general graphs. 
Our analysis of the sparsifier, however, differs from the previous ones in that we directly analyze the violation of each odd-set constraint while previous work used various proxies (e.g., kernels or $\eps$-restrictness) when arguing the integrality gap.
In particular, by standard Chernoff bounds, we show that the sparsifier maintains (i) the degree of the vertices, (ii) the total edge mass in odd sets, and (iii) the unweighted matching size. Coupled with a dynamic set sampler from, e.g., \cite{BhattacharyaKSW23rounding},
we obtain from the sampling approach an unweighted rounding algorithm for general graphs.
Though the sampling scheme itself does not lead to a runtime improvement over \cref{thm:rounding}, we further show that in the decremental setting, surprisingly, properties (i) and (ii) suffice to round the entropy-regularized fractional matching maintained by \cref{thm:decremental-fractional} in weighted graphs with large $W$.
The resulting rounding algorithm in \cref{thm:weighted-rounding} below has a near-optimal overhead in dense graphs.

\begin{restatable}{theorem}{WeightedRounding}
  There are randomized output-adaptive algorithms that solve \cref{prob:dec-matching} w.h.p. with amortized update times $\widetilde{O}(\eps^{-\expnear} + (n^2/m) \cdot \eps^{-6})$ and $\widehat{O}(\eps^{-\expalmost} + (n^2/m) \cdot \eps^{-6})$. Additionally, if $G$ is bipartite then there is such an algorithm with amortized update time $\widehat{O}((n^2/m) \cdot \eps^{-6})$.
  \label{thm:weighted-rounding}
\end{restatable}

\subsection{Putting Everything Together} \label{subsec:everything}

We conclude this overview by using the previously stated results to prove our main theorems.

\MainUnweighted*

\begin{proof}
  For $\eps < n^{-1/6}$, the update times of $\widetilde{O}(\eps^{-\expnear})$ and $\widehat{O}(\eps^{-\expalmost})$ can be obtained by re-running the static algorithm of \cite{DuanP14} after each update.
  As a result we assume $\eps \geq n^{-1/6}$ in the remainder of the proof.
  By \cref{thm:decremental-fractional} with accuracy parameter $\eps/2$, we can maintain a $\left(1-\frac{\eps}{2}\right)$-approximate fractional matching in amortized update times $\widetilde{O}(\eps^{-\expnear})$ and $\widehat{O}(\eps^{-\expalmost})$.
  We then apply \cref{thm:rounding} with accuracy parameter $\eps/2$ to round the fractional matching we maintain to a $(1-\eps)$-approximate integral matching.
  Since the recourse of \cref{thm:decremental-fractional} is $\widetilde{O}(m\eps^{-2})$, there will be $\widetilde{O}(m\eps^{-2})$ updates to \cref{thm:rounding} in total, incurring an additional $\widetilde{O}(\eps^{-6})$ amortized time per update that is subsumed by the update time of \cref{thm:decremental-fractional}.
  Since \cref{thm:rounding} is deterministic, our final algorithm works against a fully-dynamic adversary like \cref{thm:decremental-fractional} does.
\end{proof}

\MainWeighted*

We make use of the following weight reduction framework from \cite{GuptaP13} which allows us to assume that the maximum weight is bounded by $\eps^{-O(1/\eps)}$.

\begin{proposition}[\cite{GuptaP13}]
  Given a fully-dynamic/incremental/decremental algorithm for $(1-\eps)$-approximate maximum weighted matching on $n$-vertex $m$-edge graphs of maximum weight $W$ with worst-case/amortized update time $T(n, m, \eps, W)$, there is a fully-dynamic/incremental/decremental algorithm for the same task with worst-case/amortized update time $\widetilde{O}\left(T\left(n, m, \Theta(\eps), \eps^{-O(1/\eps)}\right)\right)$.
  \label{prop:weight-reduce-eps}
\end{proposition}

\begin{proof}[Proof of \cref{thm:main-weighted}]
  As in the proof of \cref{thm:main-unweighted} we assume $\eps \geq n^{-1/6}$.
  For the first algorithm, we apply the weight reduction framework of \cite{GuptaP13} in \cref{prop:weight-reduce-eps} to make $W \leq \eps^{-O(1/\eps)}$.
  Again, running \cref{thm:decremental-fractional} with accuracy parameter $\eps/2$ we maintain a $\left(1-\frac{\eps}{2}\right)$-approximate fractional matching in amortized update time $\widetilde{O}(\eps^{-\expnear})$.
  The rounding algorithm \cref{thm:rounding} now has amortized update time $\widetilde{O}(\eps^{-O(1/\eps)})$, incurring an additional $\widetilde{O}(\eps^{-O(1/\eps)})$ amortized time per update which subsumes the update time of the fractional matching.
  The algorithms for dense graphs follow from \cref{thm:weighted-rounding}.
\end{proof}

Finally, the bipartite result \cref{thm:main-weighted-bipartite} also follows from \cref{thm:weighted-rounding}.

\MainWeightedBipartite*

\section{Entropy Regularization for Decremental Linear Optimization}\label{sec:congestion-balancing}

In this section we analyze our entropy regularization strategy for the lazy update framework that solves the decremental linear optimization problem.
In~\cref{subsec:robustness,subsec:robustness-approximate} we show the robustness of the entropy regularization strategy, and in \cref{subsec:optimality} we prove its optimality.
We consider a fixed instance of the decremental linear optimization problem, including the $d$-dimensional convex set $\X$ and the input weight $\bw \in \R_{> 0}^{d}$.

\subsection{Notation and General Setup}

Before showing the robustness of our framework, we first set up the notation and various optimization constructs that will be used throughout the section.
For $S \subseteq [d]$ and $\gamma \geq 0$, we consider the entropy regularizer $r_{S,\gamma}: \X_S \to \R$ defined by
\begin{equation}
r_{S,\gamma}(\bx)\defeq \sum_{i\in S}\bw_{i}\bx_{i}\log\frac{\gamma}{\bw_i\bx_i},
\label{eq:entropy}
\end{equation}
which is a weighted and scaled version of the original entropy function $\entropy_S(\by) \defeq \sum_{i\in S}\by_i\log\frac{1}{\by_i}$ for $\by\in\Delta^S$ that is usually applied on the simplex.
For intuition, observe that for an $\bx \in \X_S$ with $\sum_{i \in S}\bw_i \bx_i = \gamma$,  we have $r_{S,\gamma}(\bx) = \gamma \cdot \entropy_S(\Pi(\bx))$, where $\Pi(\bx) \in \Delta^S$ with $(\Pi(\bx))_i \defeq \bw_i\bx_i / \gamma$.
Below we give upper and lower bounds on the value of  $r_{S,\gamma}(\bx)$ which generalize known properties of entropy on the simplex.

\begin{lemma}
  For $\bx \in \X_S$ with $\bw_S^\top \bx = \nu$, we have $\nu\log(\gamma/\nu) \leq r_{S,\gamma}(\bx) \leq \nu\log(d\gamma/\nu)$.
  \label{lemma:entropy-property}
\end{lemma}

\begin{proof}
  For the upper bound, consider the relaxation of the problem
  \begin{equation}
  \max_{\bx^\prime \in \R^S: \norm*{\bx^\prime}_1=\nu/\gamma}\gamma\cdot \sum_{i\in S}\bx^\prime_{i}\log \frac{1}{\bx^\prime_{i}}.
  \label{eq:relaxation}
  \end{equation}
  For any fixed $\bx\in\X_S$, there is a feasible point $\by^\prime$ of \eqref{eq:relaxation} with $\by^\prime_i\defeq\bw_i\bx_i/\gamma$, and $\gamma\cdot \sum_{i\in S}\by^\prime\log\frac{1}{\by^\prime}=r_{S,\gamma}(\bx)$.
  Thus the optimal value of \eqref{eq:relaxation} is an upper bound on $r_{S,\gamma}(\bx)$.
  Let $\bg_{\bx^\prime}\in\R^{S}$ with $\left(\bg_{\bx^\prime}\right)_{i}=-(1+\log \bx^\prime_{i})$ be the gradient of the objective of \eqref{eq:relaxation} at $\bx^\prime$.
  Note that the optimality conditions of the problem are that $\bg_{\bx^\prime}\perp\ker({\boldsymbol{1}^S}^{\top})$ or equivalently that $\bg_{\bx^\prime}=\alpha\cdot \boldsymbol{1}^S$ for some $\alpha\in\R$.
  Thus, it holds that $\bx^\prime=\beta\cdot \boldsymbol{1}^S$ for some $\beta\in\R$.
  Combining with $\norm*{\bx^\prime}_1=\nu/\gamma$, we have $\bx^\prime_{i}=\nu/(d\gamma)$, implying that the maximizing value of \eqref{eq:relaxation} is $\nu\log(d\gamma/\nu)$.
  
  For the lower bound, since $\bw^\top_S \bx=\nu$, we have $\max_{i \in S}\bw_i\bx_i\leq \nu$, and thus
  \[ r_{S,\gamma}(\bx) \geq \sum_{i \in S}\bw_i\bx_i\log{(\gamma/\nu)} = \nu\log(\gamma/\nu). \]
\end{proof}

Using $r_{S,\gamma}$ we may rewrite the entropy-regularized objective defined in \eqref{eq:entropy-objective} as $f_{S,\gamma}^{\mu}(\bx) \defeq \bw_S^\top\bx + \mu \cdot r_{S,\gamma}(\bx)$.
Note that $\mu>0$ in the definition.
Let $Z_{S,\gamma}^{\mu} \defeq \max_{\bx \in \X_S}f_{S,\gamma}^{\mu}(\bx)$ be the maximum entropy-regularized objective value and $\nu_{S}^{*} \defeq \max_{\bx \in \X_S}\bw^\top_S \bx$ be the maximum of the actual linear objective.
Below are properties of $f_{S,\gamma}^{\mu}$ that we will use in this section.

\begin{lemma}\label{lemma:entropy-regularized objective property}
The entropy-regularized objective function $f_{S,\gamma}^\mu$ admits the following properties:
\begin{enumerate}[(i)]
    \item\label{item:unique maximizer}$f_{S,\gamma}^\mu$ has a unique maximizer $\bx_{S,\gamma}^\mu$ on $\X_S$.
    \item\label{item:differentiability at maximizer}$\bx_{S,\gamma}^{\mu}$ has positive coordinates, i.e., $\bx_{S,\gamma}^{\mu}\in \X_{S,+}$, and thus $f_{S,\gamma}^\mu$ and $r_{S,\gamma}$ are differentiable at $\bx_{S,\gamma}^{\mu}$.
    \item\label{item:quadratic upper bound}
    $f_{S,\gamma}^{\mu}(\bx)\leq f_{S,\gamma}^{\mu}\left(\bx_{S,\gamma}^{\mu}\right)-\frac{\mu}{2\nu_{S}^{*}}\norm*{\bx-\bx_{S,\gamma}^{\mu}}^2_{\bw,S}$
    for all $\bx \in \X_S$, where $\norm{\bx}_{\bw,S} \defeq \sum_{i \in S}\bw_i |\bx_i|$ is the weighted $\ell_1$-norm.\footnote{$\norm{\cdot}_{\bw,S}$ is indeed a norm since $\bw_i > 0$ for all $i \in [d]$.}
\end{enumerate}
\end{lemma}

We provide a fairly standard proof of the properties from first principle.

\begin{proof}
    By the compactness of $\X_S$, $f_{S,\gamma}^{\mu}$ has a maximizer on $\X_S$.
    We prove the following claims.
    \begin{claim}
      Any maximizer of $f_{S,\gamma}^{\mu}$ on $\X_S$ must be in $\X_{S,+}$.
      \label{claim:maximizer}
    \end{claim}
    \begin{proof}
      Consider the first-order partial derivative of $f_{S,\gamma}^\mu$ with respect to an $\bx \in \X_{S,+}$ and coordinate $i\in S$, which by calculation is
      \[ \partial_i f_{S,\gamma}^{\mu}(\bx) \defeq \frac{\partial f_{S,\gamma}^\mu(\bx)}{\partial \bx_i} =(1-\mu)\bw_i+\mu\bw_i\log\left(\frac{\gamma}{\bw_i \bx_i}\right). \]
      Since $\X$ contains no degenerate dimension and is convex, there is a point $\by \in \X_{S,+}$.
      Consider any point $\bx\in\X_S\setminus\X_{S,+}$, and let $\bx_{\alpha} \defeq \bx + \alpha(\by - \bx)$ for $\alpha \in [0, 1]$.
      Note that $\bx_{\alpha} \in \X_{S,+}$ for $\alpha > 0$.
      By the mean value theorem, for any $\alpha > 0$ there is a $0 < \beta < \alpha$ such that
      \[
      f_{S,\gamma}^\mu(\bx_{\alpha})-f_{S,\gamma}^\mu(\bx)
      =\nabla f_{S,\gamma}^\mu(\bx_{\beta})^\top (\bx_{\alpha}-\bx)=\sum_{i\in S}\partial_if_{S,\gamma}^{\mu}(\bx_{\beta}) \cdot ((\bx_{\alpha})_i-\bx_i).
      \]
      For $i\in S$ with $\bx_i>0$, 
      $|\partial_i f_{S,\gamma}^{\mu}(\bx_{\beta})|$ is bounded since $(\bx_{\beta})_i$ is between $\bx_i$ and $\by_i$.
      For $i\in S$ with $\bx_i=0$, however, when $(\bx_{\beta})_i$ approaches $0$, $\partial_i f_{S,\gamma}^{\mu}(\bx_{\beta})$ goes to infinity.
      This shows that we can pick an $\alpha$ close enough to $0$ so that $f_{S,\gamma}^{\mu}(\bx_{\alpha}) > f_{S,\gamma}^{\mu}(\bx)$, proving that $\bx \in \X_{S} \setminus \X_{S,+}$ is not a maximizer.
    \end{proof}

    \begin{claim}
      For any $\bx, \bz \in \X_{S,+}$ it holds that $\bz^\top \nabla^2 f_{S,\gamma}^{\mu}(\bx)\bz \leq -\frac{\mu}{\nu_S^{*}}\norm*{\bz}_{\bw,S}^{2}$.
      \label{claim:hessian}
    \end{claim}

    Note that $f_{S,\gamma}^{\mu}$ is twice-differentiable on $\X_{S,+}$ and thus \cref{claim:hessian} is well-defined.

    \begin{proof}
      The second-order partial derivatives of $f_{S,\gamma}^{\mu}$ satisfy
      \[\frac{\partial^2 f_{S,\gamma}^{\mu}(\bx)}{\partial \bx_i^2}=-\frac{\mu \bw_i}{\bx_i}\quad\text{and}\quad\frac{\partial^2 f_{S,\gamma}^{\mu}(\bx)}{\partial \bx_{i}\partial \bx_{j}}=0\;\text{for}\;j \neq i.\]
      Using Cauchy-Schwarz inequality, for any $\bx,\bz\in \X_{S,+}$ we have
      \[
        \bz^\top \nabla^2f_{S,\gamma}^{\mu}(\bx)\bz=\sum_{i \in S}-\frac{\mu \bw_i}{\bx_i}\bz_i^2\leq -\mu\frac{\norm{\bz}_{\bw,S}^2}{\norm{\bx}_{\bw,S}}\leq -\frac{\mu}{\nu_{S}^{*}}\norm{\bz}_{\bw,S}^2,
      \]
      as claimed.
    \end{proof}

    Now consider a maximizer $\by$ of $f_{S,\gamma}^{\mu}$ and an $\bx \in \X_{S,+}$.
    Since $\by \in \X_{S,+}$, $f_{S,\gamma}^{\mu}$ is differentiable at $\by$.
    Letting $\by_{\alpha} \defeq \by + \alpha (\bx - \by)$ for $\alpha \in [0, 1]$, we have
    \begin{align}
      f_{S,\gamma}^{\mu}(\bx)
      &= f_{S,\gamma}^{\mu}(\by) + \left(\nabla f_{S,\gamma}^{\mu}(\by)\right)^\top(\bx - \by) + \int_{0}^{1}\int_{0}^{t}(\bx - \by)^\top \nabla^2 f_{S,\gamma}^{\mu}(\by_{\alpha}) (\bx - \by)\mathrm{d}\alpha\mathrm{d}t \\
      &\stackrel{(i)}{\leq} f_{S,\gamma}^{\mu}(\by) - \frac{\mu}{\nu_S^{*}}\norm*{\bx - \by}_{\bw,S}^2 \cdot  \int_{0}^{1}\int_{0}^{t}\mathrm{d}\alpha\mathrm{d}t = f_{S,\gamma}^{\mu}(\by) - \frac{\mu}{2\nu_S^{*}}\norm*{\bx - \by}_{\bw,S}^2,\label{eq:f-interior}
    \end{align}
    where (i) uses the optimality conditions of $f_{S,\gamma}^{\mu}$ at $\by$ and \cref{claim:hessian}.
    On the other hand, for $\bx \in \X_S \setminus \X_{S,+}$, there is a sequence $\{\bx_n\} \subseteq \X_{S,+}$ approaching $\bx$.
    By continuity of $f_{S,\gamma}^{\mu}$ we have
    \begin{equation}
      f_{S,\gamma}^{\mu}(\bx) = \lim_{n \to \infty}f_{S,\gamma}^{\mu}(\bx_n) \leq \lim_{n \to \infty}f_{S,\gamma}^{\mu}(\by) - \frac{\mu}{2\nu_S^{*}}\norm*{\bx_n - \by}_{\bw,S}^2 = f_{S,\gamma}^{\mu}(\by) - \frac{\mu}{2\nu_S^{*}}\norm*{\bx - \by}_{\bw,S}^2.
      \label{eq:f-boundary}
    \end{equation}
    This implies that there is a unique maximizer $\bx_{S,\gamma}^{\mu}$ of $f_{S,\gamma}^{\mu}$ on $\X_S$.
    The rest of the lemma follows from \cref{claim:maximizer} and \cref{eq:f-interior,eq:f-boundary}.
\end{proof}

Finally, our analysis of the framework uses the Bregman divergence of $r_{S,\gamma}$ as a proxy to bound the decrease of a certain potential.

\begin{definition}[Bregman divergence]
  For differentiable $r: \X \to \R$ and $\bx, \by \in \X$ for a domain $\X$, the \emph{Bregman divergence} of $r$ \emph{from $\by$ to $\bx$} is 
  \[ V_{\by}^{r}(\bx) \defeq r(\bx) - \left(r(\by) + \left(\nabla r(\by)\right)^\top \left(\bx-\by\right)\right).\]
\end{definition}

Overloading notation, let $V_{\by}^{S,\gamma}(\bx)$ for $S \subseteq [d], \bx\in \X_S$ and $\by \in \X_{S,+}$ be the Bregman divergence induced by $r_{S,\gamma}$, which is non-positive since $r_{S,\gamma}$, as a generalization of the entropy function, is concave.
We have by direct calculation from the definition that

\begin{align}
V_{\by}^{S,\gamma}(\bx) &\defeq V_{\by}^{r_{S,\gamma}}(\bx)
= r_{S,\gamma}(\bx)-\left(r_{S,\gamma}(\by)+\left(\nabla r_{S,\gamma}(\by)\right)^\top\left(\bx-\by\right)\right)\\
& =\left(\sum_{i\in S}\bw_{i}\bx_{i}\log \frac{\gamma}{\bw_i\bx_{i}}\right)-\left[\left(\sum_{i\in S}\bw_{i}\by_{i}\log \frac{\gamma}{\bw_i\by_{i}}\right)-\left(\sum_{i\in S}\bw_{i}\left(1+\log \frac{\bw_i\by_{i}}{\gamma}\right)(\bx_{i}-\by_{i})\right)\right]\\
 & =\sum_{i\in S}\bw_{i}\bx_{i}\log\frac{\by_i}{\bx_i}+\sum_{i\in S}\bw_{i}(\bx_{i}-\by_{i}).\label{eq:divergence}
\end{align}

Note that we need $\by\in\X_{S,+}$ because the gradient does not exist on $\X_S\setminus \X_{S,+}$.
Indeed, in the remainder of the section we will only use $V_{\by}^{S,\gamma}(\bx)$ for $\by=\bx^{\mu}_{S,\gamma}$, which lies in $\X_{S,+}$ by \cref{lemma:entropy-regularized objective property}\labelcref{item:differentiability at maximizer}.
The following lemma bounds the entropy-regularized objective by the Bregman divergence from $\bx_{S,\gamma}^{\mu}$.

\begin{lemma}\label{lemma:Bregman divergence property}
    For any $\bx\in\X_S$, we have $f_{S,\gamma}^\mu(\bx)\leq f_{S,\gamma}^\mu(\bx_{S,\gamma}^\mu)+\mu V_{\bx_{S,\gamma}^\mu}^{S,\gamma}(\bx)$.
\end{lemma}

\begin{proof}
For clarity let us write $\bx^{*} \defeq \bx_{S,\gamma}^{\mu}$.
Optimality conditions for concave optimization applied to $f_{S,\gamma}^\mu(\bx^{*})$ imply that

\begin{equation}
(\bw_S + \mu\nabla r_{S,\gamma}(\bx^{*})^\top (\bx-\bx{*})=\nabla f_{S,\gamma}^{\mu}(\bx^{*})^\top (\bx-\bx^{*})\leq 0.
\label{eq:opt-condition}
\end{equation}

As such, we have
\begin{align*}
f_{S,\gamma}^{\mu}(\bx^{*})-f_{S,\gamma}^{\mu}(\bx)
&\stackrel{(i)}{=}\bw^\top_S (\bx^{*}-\bx)+\mu(r_{S,\gamma}(\bx^{*})-r_{S,\gamma}(\bx))\\
&\stackrel{(ii)}{=}(\bw_S+\mu\nabla r_{S,\gamma}(\bx^{*}))^\top (\bx^{*} - \bx)-\mu V^{S,\gamma}_{\bx^{*}}(\bx)\stackrel{(iii)}{\geq} -\mu V^{S,\gamma}_{\bx^{*}}(\bx),
\end{align*}
where (i) is by definition of $f_{S,\gamma}^{\mu}$, (ii) is by definition of the Bregman divergence, and (iii) follows from \eqref{eq:opt-condition}.
\end{proof}

The following \cref{tab:notation} summarizes the notation introduced here for future reference.

\begin{table}[ht]
\centering
\small
\def\arraystretch{1.5}
\begin{tabular}{l|l|l}
 Notation & Definition & Description \\
\hline
$r_{S,\gamma}(\bx)$ & $\sum_{i\in S}\bw_{i}\bx_{i}\log\frac{\gamma}{\bw_i\bx_i}$ & entropy regularizer\\
\hline
$V_{\by}^{S,\gamma}(\bx)$ & $\sum_{i\in S}\bw_{i}\bx_{i}\log(\by_{i}/\bx_{i})+\sum_{i\in S}\bw_{i}(\bx_{i}-\by_{i})$ & Bregman divergence of entropy regularizer\\
\hline
$f_{S,\gamma}^{\mu}(\bx)$ & $\bw^{\top}_S\bx+\mu \cdot r_{S,\gamma}(\bx)$ & entropy-regularized objective\\
\hline
$\bx^{\mu}_{S,\gamma}$ & $\argmax_{\bx \in \X_S}f^{\mu}_{S,\gamma}(\bx)$ & optimal entropy-regularized solution \\
\hline
$\nu_{S}^{*}$ & $\max_{\bx \in \X_S}\bw^\top_S \bx$ & optimal linear objective value\\
\hline
$Z_{S,\gamma}^{\mu}$ & $\max_{\bx \in \X_S}f_{S,\gamma}^{\mu}(\bx)$ & optimal entropy-regularized objective value\\
\hline
$\norm{\bx}_{\bw,S}$ & $\sum_{i \in S}\bw_i|\bx_i|$ & norm of $\bx \in \X_S$ induced by $\bw$ \\
\end{tabular}
\caption{Frequently used notation}
\label{tab:notation}
\end{table}

\subsection{Robustness of Entropy Regularization}\label{subsec:robustness}

In this section we show that if the subroutine $\texttt{Rebuild()}$ in \cref{alg:meta-general-ds} returns the exact maximizer of the entropy-regularized objective, then $\texttt{Rebuild()}$ will be called at most $\otilde (\eps^{-2})$ times.
The first step is to derive a proper choice of $\gamma$ and $\mu$ so that $\bx_{S,\gamma}^{\mu}$ becomes an approximate maximum weight solution that can be used in \cref{alg:meta-general-ds}.
We call a tuple containing an accuracy parameter $\eps > 0$, coordinate subset $S \subseteq [d]$, and estimate
\begin{equation}
  \nu_{S}^{*} \leq \gamma \leq d \cdot \nu_{S}^{*}
  \label{eq:gamma}
\end{equation}
a \emph{valid iterate} which will repeatedly appear as input to lemmas in the remainder of this section.

\begin{lemma}\label{lemma:absolute entropy bound}
    For any valid iterate $(\eps, S, \gamma)$ and $0 < \mu \leq \frac{\eps}{8\log{d}}$, it holds that $0\leq \mu \cdot r_{S,\gamma}(\bx) \leq \frac{\eps}{2} \cdot \nu_{S}^{*}$ for all $\bx \in \X_S$.
\end{lemma}

\begin{proof}
    The lower bound follows from \cref{lemma:entropy-property} and the assumption of $\gamma$ in \eqref{eq:gamma}.
	For the upper bound, let $\bx \in \X_S$ be arbitrary and define $\nu \defeq \bw^\top_S \bx$.
    \cref{lemma:entropy-property} implies that
    \begin{align*}
        \mu \cdot r_{S,\gamma}(\bx)
        &\leq \mu\nu\log\frac{d\gamma}{\nu}
        =\mu\nu\log\frac{d\gamma}{\nu_S^{*}}+\mu\nu_{S}^*\cdot \frac{\nu}{\nu_S^{*}}\log\frac{\nu_S^{*}}{\nu} \\
        &\stackrel{(i)}{\leq} \mu \nu \cdot (2\log d)  + \mu\nu_{S}^* \stackrel{(ii)}{\leq} \frac{\varepsilon}{2}\cdot \nu_{S}^*,
    \end{align*}
    where (i) comes from $x\log(1/x) \leq 1$ for all $x \in \R_{\geq 0}$ and the assumption of $\gamma$ in \eqref{eq:gamma} and (ii) is by the assumption of $\mu$ and $\nu\leq \nu_{S}^{*}$.
\end{proof}

\begin{lemma}\label{lemma:valid approximation}
    For any valid iterate $(\eps, S, \gamma)$ and $0 < \mu \leq \frac{\eps}{8\log{d}}$, it holds that
    \[\bw^\top_S \bx_{S,\gamma}^{\mu}\geq \left(1-\frac{\eps}{2}\right)\cdot \nu^*_{S}.\]
\end{lemma}

\begin{proof}
    For any $\bx\in \X_S$ such that $\bw^\top_S \bx<(1-\eps/2)\cdot \nu^*_{S}$, we have
    \[
        f_{S,\gamma}^{\mu}\left(\bx_{S,\gamma}^{\mu}\right)-f_{S,\gamma}^{\mu}(\bx)\stackrel{(i)}{\geq} \nu_{S}^{*}-\left(\bw^\top_S \bx+\frac{\eps}{2} \cdot \nu_{S}^{*}\right)\stackrel{(ii)}{>} 0,
    \]
    where (i) comes from \Cref{lemma:absolute entropy bound} (note that the non-negativity of the entropy-regularizer implies $f_{S,\gamma}^{\mu}(\bx_{S,\gamma}^{\mu}) \geq \nu_S^{*}$) and (ii) comes from the assumption that $\bw^\top_S \bx < (1-\eps/2)\cdot \nu_{S}^{*}$.
    The lemma follows since $x_{S,\gamma}^{\mu}$  is a maximizer of $f_{S,\gamma}^{\mu}$.
\end{proof}

\Cref{lemma:valid approximation} shows that for $\gamma$ being a $d$-approximate upper bound of $\nu_{S}^*$, if we choose $\mu$ to be roughly proportional to $\eps$, then the corresponding value of $\bw^\top_S \bx_{S,\gamma}^{\mu}$ is a $(1-\eps/2)$-approximation to $\nu_{S}^*$.
Thus we can use the entropy-regularized solution for the $\texttt{Rebuid()}$ procedure in \Cref{alg:meta-general-ds}.
Moreover, as we will show below, a decrease in the linear objective value of $\bx_{S,\gamma}^{\mu}$ implies a decrease in the entropy-regularized objective globally.
This enables the use of the latter quantity as a potential to bound the number of calls to $\texttt{Rebuild()}$ \cref{alg:meta-general-ds} needs.

\begin{lemma}
  For any valid iterate $(\eps, S, \gamma)$, $0 < \mu \leq \frac{\eps}{8\log{d}}$, and $S^\prime \subseteq S$, if
  \[ {\bw}^\top_{S^\prime} \left(\bx_{S,\gamma}^{\mu}\right)_{S^\prime} < \left(1-\frac{\eps}{2}\right)\cdot \bw^\top_S \bx_{S,\gamma}^{\mu}, \]
  then $Z_{S^\prime,\gamma}^{\mu} \leq \left(1-\frac{\mu\eps}{3}\right)Z_{S,\gamma}^{\mu}$.
  \label{lemma:entropy-decreases}
\end{lemma}

\begin{proof}
For clarity let $\bx_{S} \defeq \bx_{S,\gamma}^{\mu}$, and $\bx_{S^\prime} \defeq \left(\bx_{S^\prime,\gamma}^{\mu}\right)^S$ be $\bx_{S^\prime,\gamma}^{\mu}$ extended to have coordinates $S$.
\cref{lemma:Bregman divergence property} shows that 
\[f_{S,\gamma}^{\mu}(\bx_{S^\prime})\leq f_{S,\gamma}^{\mu}(\bx_{S})+\mu V^{S,\gamma}_{\bx_{S}}(\bx_{S^\prime}).\]
Since $S^\prime \subseteq S$, we have
\begin{equation}
Z_{S^\prime,\gamma}^{\mu} = f_{S^\prime,\gamma}^{\mu}\left(\bx_{S^\prime,\gamma}^{\mu}\right) = f_{S,\gamma}^{\mu}(\bx_{S^\prime})\leq f_{S,\gamma}^{\mu}(\bx_{S})+\mu V_{\bx_{S}}^{S,\gamma}(\bx_{{S^\prime}}) = Z_{S,\gamma}^{\mu} + \mu V_{\bx_S}^{S,\gamma}(\bx_{S^\prime}).
\label{eq:f-bound}
\end{equation}
Further, by above letting $D \defeq S \setminus S^\prime$ and therefore $\bx_{{S^\prime},i} = 0$ for $i \in D$, we have from \eqref{eq:divergence} and the non-positiveness of the Bregman divergence for concave functions that
\begin{align*}
V_{\bx_{S}}^{S,\gamma}(\bx_{{S^\prime}}) & =\sum_{i\in S}\bw_{i}(\bx_{{S^\prime}})_i\log\frac{(\bx_{S})_i}{(\bx_{{S^\prime}})_i}+\sum_{i\in S}\bw_{i}\left((\bx_{{S^\prime}})_i-(\bx_{S})_i\right)\\
 & \leq \sum_{i\in D}\bw_{i}(\bx_{{S^\prime}})_i\log\frac{(\bx_S)_i}{(\bx_{S^\prime})_i}+\sum_{i\in D}\bw_{i}\left((\bx_{{S^\prime}})_i-(\bx_{S})_i\right)\\
 & =-\sum_{i\in D}\bw_i (\bx_{S})_i\leq-\frac{\eps}{2}\cdot \bw^{\top}_S\bx_{S},
\end{align*}
which when combined with \eqref{eq:f-bound} and \cref{lemma:absolute entropy bound} shows that
\begin{align*}
Z_{S^\prime,\gamma}^{\mu}\leq Z_{S,\gamma}^{\mu} - \frac{\mu\eps}{2} \cdot \bw^{\top}_S\bx_{S}\leq \left(1-\frac{\mu \eps}{(1+\eps/2) \cdot 2}\right)\cdot Z_{S,\gamma}^{\mu}\leq \left(1-\frac{\mu \eps}{3}\right)\cdot Z_{S,\gamma}^{\mu}.
\label{eq:f-bound}
\end{align*}
\end{proof}

We now conclude the number of rebuilds before the actual optimal value $\nu_{S}^*$ drops to a certain number, if we use an entropy-regularized solution $\bx_{S,\gamma}^{\mu}$ for $\texttt{Rebuild()}$.

\UpperBoundReconstruction*

\begin{proof}
Note that while $\nu_{S}^{*}$ is between $\alpha/d$ and $\alpha$, $\gamma \defeq \alpha$ is a $d$-approximate upper bound of $\nu_{S}^{*}$ and thus the preceding lemmas apply.
By \cref{lemma:valid approximation}, $\bx^{(t)}$ is a $\left(1-\frac{\eps}{2}\right)$-approximate solution to the linear objective ${\bw}^\top_{S^{(t)}} \bx$ before $\nu_{S}^{*}$ drops below $\alpha$.
Fix a round $t$ and consider $S^{(t+1)}$ which is obtained by deleting some coordinates from $S^{(t)}$ so that
\[{\bw}^{\top}_{S^{(t+1)}}\bx^{(t)}_{S^{(t+1)}}<\left(1-\frac{\eps}{2}\right){\bw}^{\top}_{S^{(t)}}\bx^{(t)}.\]
In other words, the quality of the current solution $\bx^{(t)}$ decreases by a multiplicative factor of $1-\frac{\eps}{2}$ when we go from $S^{(t)}$ to $S^{(t+1)}$.
\cref{lemma:entropy-decreases} then implies
\[
  Z_{S^{(t+1)},\gamma}^{\mu} \leq \left(1-\frac{\mu\eps}{3}\right)Z_{S^{(t)},\gamma}^{\mu}
\]
for each round $t$.
This allows us to bound the number of calls to $\texttt{Rebuild()}$ as follows.
Initially when $\max_{x \in \X_S}\bw^\top_S \bx \leq \alpha$, the optimal entropy-regularized objective is no more than $(1+\eps/2)\cdot \alpha$ by \cref{lemma:absolute entropy bound}.
Likewise, at the end before $\nu_{S}^{*}$ drops below $\alpha/d$, the objective is at least $\alpha/d$.
Thus there will be at most
\[
  \log_{1 - \frac{\mu\eps}{3}}\left((1 + \eps/2)d\right) = O\left(\frac{\log d}{\mu\eps}\right)
\]
calls to $\texttt{Rebuild()}$.
\end{proof}

\subsection{Sufficiency of Approximate Solutions} \label{subsec:robustness-approximate}

From \cref{subsec:robustness} we have seen that the maximizer of the entropy-regularized objective solves the decremental linear optimization problem.
However, exact maximizers are not always easy to obtain, and therefore in this section, we show that any accurate enough approximate maximizer of $f_{S,\gamma}^\mu$ suffices for the lazy update framework to work efficiently.
This is by the following lemma which states that such a solution is also close in $\norm{\cdot}_{\bw,S}$-distance to the actual maximizer $\bx_{S,\gamma}^{\mu}$.

\begin{lemma}\label{lemma:approximate solution}
For any valid iterate $(\eps, S, \gamma)$, $0 < \mu \leq \frac{\eps}{8\log{d}}$, and $\bx \in \X_S$ with $f_{S,\gamma}^{\mu}(\bx) \geq \left(1 - \frac{\mu\eps^2}{2}\right)Z_{S,\gamma}^{\mu}$, it holds that $\norm*{\bx - \bx_{S,\gamma}^{\mu}}_{\bw,S} \leq \eps\cdot Z_{S,\gamma}^{\mu}$.
\end{lemma}
\begin{proof}
    \cref{lemma:entropy-regularized objective property}\labelcref{item:quadratic upper bound} states that
    \[ f_{S,\gamma}^{\mu}(\bx)\leq f_{S,\gamma}^{\mu}\left(\bx_{S,\gamma}^{\mu}\right)-\frac{\mu}{2\nu_{S}^{*}}\norm*{\bx-\bx_{S,\gamma}^{\mu}}^2_{\bw,S}\]
    for all $\bx \in \X_S$.
    Thus for every $\bx\in \X_S$ such that $f_{S,\gamma}^{\mu}(\bx)\geq \left(1-\frac{\mu\eps^2}{2}\right)\cdot Z_{S,\gamma}^{\mu}$, we have
    \[\norm*{\bx-\bx_{S,\gamma}^{\mu}}_{\bw,S}\leq \sqrt{\frac{2\nu_{S}^{*}}{\mu}\left(f_{S,\gamma}^{\mu}\left(\bx_{S,\gamma}^{\mu}\right)-f_{S,\gamma}^{\mu}(\bx)\right)}\leq \sqrt{\varepsilon^2\cdot \nu_{S}^{*}\cdot Z_{S,\gamma}^{\mu}}\leq \varepsilon\cdot Z_{S,\gamma}^{\mu},\]
    where the last inequality follows from \Cref{lemma:absolute entropy bound} and therefore $\nu_{S}^{*} \leq Z_{S,\gamma}^{\mu}$.
\end{proof}

The closeness of an approximate maximizer $\bx$ to the actual one $\bx_{S,\gamma}^{\mu}$ allows us to bound the decrease in the objective value of $\bx_{S,\gamma}^{\mu}$ that is hidden to us.
This establishes the number of rebuilds needed if we only have an accurate enough approximation to $f_{S,\gamma}^{\mu}$.

\UpperBoundReconstructionApprox*

\begin{proof}
  Setting $\eps^\prime \defeq \eps/4$ and $\eps^{\prime\prime} \defeq \eps/16$, we have $\mu \leq \frac{\eps^{\prime\prime}}{8\log{d}} \leq \frac{\eps^{\prime}}{8\log{d}}$ and $\bx^{(t)}$ being a $\left(1-\frac{\mu\eps^{\prime\prime}}{2}\right)$-approximate solution to $f_{S^{(t)},\gamma}^{\mu}$.
  By \cref{lemma:valid approximation} with accuracy parameter $\eps^\prime$ we know that $\bx_{S^{(t)},\gamma}^{\mu}$ is a $\left(1 - \frac{\eps}{8}\right)$-approximate solution to the linear objective ${\bw}^\top_{S^{(t)}} \bx$ before $\nu_{S}^*$ drops below $\alpha/d$.
  By triangle inequality of the norm $\norm{\cdot}_{\bw,S^{(t)}}$ and \cref{lemma:approximate solution} with accuracy parameter $\eps^{\prime\prime}$, the value of ${\bw}^\top_{S^{(t)}} \bx^{(t)}$ satisfies
  \begin{align}
    {\bw}^\top_{S^{(t)}} \bx^{(t)} = \norm*{\bx^{(t)}}_{\bw,S^{(t)}} &\geq \norm*{\bx_{S^{(t)},\gamma}^{\mu}}_{\bw,S^{(t)}} - \norm*{\bx_{S^{(t)},\gamma}^{\mu} - \bx^{(t)}}_{\bw,S^{(t)}} \\
    &\geq \left(1 - \frac{\eps}{8}\right)\nu_{S^{(t)}}^{*} - \frac{\eps}{16}Z_{S^{(t)},\gamma}^{\mu} \geq \left(1 - \frac{\eps}{4}\right)\nu_{S^{(t)}}^{*}, \label{eq:value}
  \end{align}
  where the last inequality uses that $Z_{S^{(t)},\gamma}^{\mu} \leq \left(1 + \frac{\eps}{8}\right)\nu_{S^{(t)}}^{*} \leq 2\nu_{S^{(t)}}^{*}$ by \cref{lemma:absolute entropy bound} with accuracy parameter $\eps^\prime$.
  This shows that $\bx^{(t)}$ is indeed a $\left(1 - \frac{\eps}{2}\right)$-approximate solution to the linear objective, as required by \cref{alg:meta-general-ds}.
  Fix a round $t$ and consider $S^{(t+1)}$  which is obtained by deleting some coordinates from $S^{(t)}$ so that
  \[
    {\bw}^\top_{S^{(t+1)}} \bx^{(t)}_{S^{(t+1)}} < \left(1 - \frac{\eps}{2}\right){\bw}^\top_{S^{(t)}} \bx^{(t)}.
  \]
  This implies
  \begin{align*}
    {\bw}^\top_{S^{(t+1)}} \left(\bx^{\mu}_{S^{(t)},\gamma}\right)_{S^{(t+1)}} &= \norm*{\left(\bx^{\mu}_{S^{(t)},\gamma}\right)_{S^{(t+1)}}}_{\bw,S^{(t+1)}}  \\
    &\leq \norm*{\left(\bx^{(t)}\right)_{S^{(t+1)}}}_{\bw,S^{(t+1)}} + \norm*{\left(\bx^{\mu}_{S^{(t)},\gamma} - \bx^{(t)}\right)_{S^{(t+1)}}}_{\bw,S^{(t+1)}} \\
    &< \left(1-\frac{\eps}{2}\right){\bw}^\top_{S^{(t)}} \bx^{(t)} + \norm*{\bx^{\mu}_{S^{(t)},\gamma} - \bx^{(t)}}_{\bw,S^{(t)}} \\
    &\stackrel{(i)}{\leq} \left(\frac{1-\eps/2}{1-\eps/4}\right){\bw}^\top_{S^{(t)}} \bx_{S^{(t)},\gamma}^{\mu} + \frac{\eps}{16}\cdot Z_{S^{(t)},\gamma}^{\mu} \\
    &\stackrel{(ii)}{\leq} \left(\frac{1-\eps/2}{1-\eps/4}\right){\bw}^\top_{S^{(t)}} \bx_{S^{(t)},\gamma}^{\mu} + \frac{\eps}{16(1-\eps/8)}{\bw}^\top_{S^{(t)}} \bx_{S^{(t)},\gamma}^{\mu} \\
    &\stackrel{(iii)}{\leq} \left(1-\frac{\eps}{8}\right){\bw}^\top_{S^{(t)}} \bx^{\mu}_{S^{(t)},\gamma},
  \end{align*}
  where (i) is by \eqref{eq:value}, (ii) is by \cref{lemma:absolute entropy bound} with accuracy parameter $\eps^\prime$, and (iii) uses $\frac{1-\eps/2}{1-\eps/4} \leq 1-\eps/4$ for $\eps \in (0, 1)$.
  The theorem then follows from \cref{lemma:entropy-decreases} with accuracy parameter $\eps^\prime$ and the same reasoning that proves \cref{lemma:upper bound on reconstruction rounds}.
\end{proof}

\subsection{Near-Optimality of Entropy Regularization} \label{sec:lowerbound}\label{subsec:optimality}

We have shown earlier in this section that \cref{alg:meta-general-ds} with the entropy regularization strategy \cref{alg:rebuild} solves the decremental linear optimization problem with at most $O(\log^2 d/\varepsilon^2)$ calls to $\texttt{Rebuild()}$ before the optimal linear objective drops from $d\cdot \alpha$ to $\alpha$. Complementing this result, in this section we show that this bound is near optimal for a certain range of $\eps$.

More specifically, we focus on the special case of unweighted bipartite matching, i.e., when $\X = \P_G = \M_G$ for some bipartite graph $G$, and consider any algorithm that implements the $\texttt{Rebuild()}$ subroutine. In this decremental unweighted bipartite matching problem, \cref{thm:upper bound on reconstruction rounds approximate multiple phases} gives an upper bound of $O(\log^2 n/\varepsilon^2)$ on the number of calls to \texttt{Rebuild()}, where $n$ is the number of vertices.
The following theorem establishes an $\Omega(\log^2(\eps^2 n)/\eps^2)$ lower bound in the regime of $\eps \geq \Omega(1/\sqrt{n})$ against an output-adaptive adversary.

\LowerBound*

We first introduce a graph structure $G_k$ that will repeatedly appear during the deletion process.

\begin{definition}
  For $k \in \N$, $G_k$ is a bipartite graph with $k$ vertices on each side, and the edge set $E(G_k)$ of $G_k$ is $\{\{i^{\ell},j^r\} \mid 1\leq j\leq i\leq k\}$.
  \label{definition: G_k}
\end{definition}

It is straightforward to check that $G_k$ has a unique maximum matching $M=\{\{i^\ell,i^r\} \mid 1\leq i\leq k\}$. The following lemma shows a generalization of the observation, that any large enough fractional matching has a mass concentration on edges with small differences in their endpoint labels.

\begin{lemma} \label{lemma:concentration of fractional matching}
    For $k \in \N$ and $\eta,\delta>0$, any fractional matching $\bx$ of $G_k$ with matching size $\norm*{\bx}_1 \geq \left(1-\eta\right)k$ satisfies that
    \[\sum_{0\leq i-j\leq \delta k}\bx_{\{i,j\}}\geq  \left(1-(1+\delta^{-1})\eta\right)\cdot k.\]
\end{lemma}
\begin{proof}
By \cref{definition: G_k}, for any edge $\{i^\ell,j^r\}$ in $G_k$, we have $i-j\geq 0$.
Also, we can upper-bound the weighted sum of difference by
\[
\sum_{\{i^\ell,j^r\}\in E(G_k)}(i-j)\bx_{\{i^\ell,j^r\}}=\sum_{i\in[k]}i\cdot \bx(i^\ell)-\sum_{j\in[k]}j\cdot\bx(j^r)\leq \sum_{j\in[k]}j\cdot(1-\bx(j^r))\leq k(k-\norm*{\bx}_1)\leq \eta k^2.
\]
By Markov's inequality, $\sum\limits_{i-j>\delta k}\bx_{\{i^\ell,j^r\}}\leq  \eta\delta^{-1}k$, and thus $\sum\limits_{0\leq i-j\leq \delta k}\bx_{\{i^\ell,j^r\}}\geq  \left(1-(1+\delta^{-1})\eta\right)\cdot k$.
\end{proof}

Now we are ready to construct an output-adaptive adversary to achieve the previously claimed lower bound.

\begin{proof}[Proof of \cref{thm:lower-bound}]
Let $\A$ be an instance of \cref{alg:meta-general-ds} with the given implementation of $\texttt{Rebuild()}$.
The adversarial input graph to the algorithm is $G_n$, and the adversary works in phases.
At the beginning of the $t$-th phase, let $k_t$ be the largest number such that $G_{k_t}$ is a subgraph of the current graph. The adversary will guarantee that $k_t = (1-4\eps)^t\cdot n$.
In the $t$-th phase, the adversary will cause $\Omega(\log(\varepsilon^2 k_t)/\varepsilon)$ rebuilds in this phase. Before $k_t$ reaches $4/\varepsilon^2$, there will be at least $\log_{1-4\varepsilon}(4/(\varepsilon^2n))=\Omega(\log(\varepsilon^2 n)/\varepsilon)$ phases, achieving a 
\[\sum_{t=0}^{\Omega(\log (\varepsilon^2n)/\varepsilon)}\Omega\left(\log \left(\varepsilon^2 k_t\right)/\varepsilon\right)=\sum_{t=0}^{\Omega(\log (\varepsilon^2n)/\varepsilon)}\Omega\left(\log \left(\varepsilon^2 (1-4\varepsilon)^t n\right)/\varepsilon\right)=\Omega(\log^2(\varepsilon^2n)/\varepsilon^2)\]
lower bound on the total number of rebuilds.
In the remainder of the proof, we focus on a single phase $t$,
and show how the adversary can cause $\Omega(\log(\varepsilon^2 k_t)/\varepsilon)$ rebuilds in this phase.

\begin{algorithm2e}[!ht]
  \caption{The adversarially chosen sequence of deletions} \label{alg:adversary}
  
  \SetEndCharOfAlgoLine{}
  
  Let $G \gets G_n$ be the initial graph and feed it to $\A$.\;
  Let $t \gets 0$.\;

  \While(\tcp*[f]{phase $t$}){$k_t \defeq (1-4\eps)^t$ satisfies $k_t > 4/\eps^2$} {
    Identify a subgraph $G_{k_t} \subseteq G$, delete $G \setminus G_{k_t}$, and re-label vertices so that $E(G) = \{\{i^{\ell}, j^r\}: 1 \leq j \leq i \leq k_t\}$.  \tcp*{preprocessing}
    Let $\M \defeq \{M_p: 0 \leq p \leq \eps k_t / 2\}$, where $M_p \defeq \{\{i^\ell, j^r\}: i = j + p\}$.\;
    \While(\tcp*[f]{regular deletion}){$|\M| > 1/\eps$} {
      Let $\bx$ be the current matching output by $\A$.\;
      Choose $r \defeq 2\eps|\M|$ matchings $M_{p_1}, \ldots, M_{p_r}$ such that $\norm*{\bx_{M_{p_1}}}_1 + \cdots + \norm*{\bx_{M_{p_r}}}_1 \geq \eps/2 \cdot k_t$.\tcp*{guaranteed by \cref{lemma:concentration of fractional matching}}
      Delete $M_{p_1} \cup \cdots \cup M_{p_r}$ and set
      $\M=\M \setminus \{M_{p_j}: 1 \leq j \leq r\}$.\tcp*{cause a rebuild}
    }
    $t \gets t + 1$.\;
  }
\end{algorithm2e}

\paragraph{Preprocessing.} At the beginning of phase $t$, by definition of $k_t$, $G_{k_t}$ is a subgraph of the current graph $G$. The adversary will first delete edges outside this subgraph, then relabel vertices in this subgraph on each side from $1$ to $k_t$ in a way that the current set of edges is $\{\{i^\ell,j^r\}\mid 1 \leq j \leq i \leq k_t\}$. The adversary will then delete edges $\left\{\{i^\ell,j^r\} \mid j+\eps k_t/2 <i<j+4\eps k_t\right\}$. After that, regular deletion starts.

\paragraph{Regular Deletion.} A regular deletion starts after a preprocessing finishes. During the regular deletion of phase $t$, the adversary will only delete edges in the subgraph $G_{k_t}^\eps \defeq \{\{i^\ell,j^r\} \mid j\leq i\leq j+\eps k_t/2\}$ of $G_{k_t}$, and the deletion continues until the maximum matching size $M^*(G)$ becomes less than $(1-\eps/2)k_t$. After that, the current phase ends. Since no edges in the subgraph $\{\{i^\ell,j^r\}\mid i\geq j+4\eps k_t\}$ are deleted, $G_{(1-4\eps)k_t}=G_{k_{t+1}}$ is a subgraph of the current graph at the beginning of phase $t+1$. It remains to show that $\Omega(\log(\eps^2k_t)/\eps)$ rebuilds could be caused during the deletion described above in $G_{k_t}^\eps$.

Note that $G_{k_t}^\varepsilon$ is the union of $\Theta(\varepsilon k_t)$ matchings where the $p$-th matching is $M_p \defeq \{\{i^\ell,j^r\} \mid i=j+p\}$, each of size at least $(1-\eps/2)k_t$.
Therefore, as long as one of the $M_p$'s remains intact, we have $(1-\eps/2)k_t \leq M^*(G) \leq k_t$.
This implies that the fractional matching $\bx$ that $\A$ maintains must have size at least
\[ \norm*{\bx}_1 \geq \left(1-\frac{\eps}{2}\right)^2 k_t - \frac{\eps}{2} \cdot k_t \geq \left(1-\frac{3\eps}{2}\right)k_t\]
throughout this phase, since by definition of \cref{alg:meta-general-ds} $\bx$ was a $(1-\eps/2)$-approximate matching since the last rebuild, and after $\frac{\eps}{2} \cdot k_t$ units of mass get deleted, a rebuild must be caused.

By the above argument it suffices to delete $\eps/2 \cdot k_t$ units of mass to cause a rebuild from $\A$.
Applying \Cref{lemma:concentration of fractional matching} with $\eta=3\eps/2$ and $\delta=4\eps$ on the fractional matching $\bx$, we know that at least $k_t/4$ units of mass are on $G_{k_t}^{\varepsilon}$, i.e, $\norm*{\bx_{E(G_{k_t}^\eps)}}_1 \geq k_t/4$.
Thus via an averaging argument the adversary can delete $2\eps$ fraction of the matchings $M_p$'s to cause a rebuild.
After the rebuild, we repeat the same argument again, choosing $2\eps$ fraction of the remaining $M_p$'s to cause another rebuild.
Note that after the first batch of deletions, the current graph is no longer $G_{k_t}$.
Nevertheless, we can still apply \cref{lemma:concentration of fractional matching} by interpreting the future $\bx$'s as a fractional matching on $G_{k_t}$ by assigning a mass of zero on deleted edges.
This shows that the mass of $\bx$ is still concentrated on the remaining $M_p$'s.
Before the number of intact matchings in $M_p$'s reaches $1/(2\varepsilon)$, there will be at least $\log_{1-2\varepsilon}(4/(\varepsilon^2k_t))=\Omega(\log(\varepsilon^2 k_t)/\varepsilon)$ reconstructions, if $\eps k_t/2\geq 1/(2\eps)$ or equivalently $k_t\leq 4/\eps^2$.
This completes the proof.
\end{proof}

\section{Decremental Algorithms for Fractional Matching}\label{sec:decremental-fractional}

In \cref{sec:congestion-balancing} we showed that 
the decremental linear optimization problem reduces to computing an approximate solution to the entropy-regularized problem with moderate accuracy. In particular, in the special case of decremental bipartite matching, i.e., when $\X$ is the matching polytope $\M_G$, it suffices to solve the entropy-regularized matching problem recalled below.

\EntropyMatching*

Consequently, in the remainder of the section we focus on solving \cref{problem:entropy-matching}.
We prove the following theorem.

\EntropySolver*

Our approaches for \cref{thm:entropy-solver} follow from a generalization of the MWU-style algorithm of \cite{AhnG14} that solves the uncapacitated and capacitated versions of weighted matching.
Given a graph $G = (V, E)$ and consider a downward closed convex set $\P \subseteq \R_{\geq 0}^E$ of interest.
By modularizing and generalizing the framework of \cite{AhnG14}, we derive the following \cref{lemma:main-general} on concave function optimization over the matching polytope.
Informally speaking, \cref{lemma:main-general} gives a reduction from approximately maximizing a concave function $\P$ \emph{together with} odd-set constraints to the same task just in $\P$ by increasing the dependence on $\eps^{-1}$.
We then, in \cref{subsec:almost-linear-solver,subsec:near-linear-reduction} respectively, instantiate the framework of \cref{lemma:main-general} in two different ways with the entropy-regularized function to obtain the two runtimes stated in \cref{thm:entropy-solver}.

Let $\kappa_{\mathcal{P}}$ be the minimum number such that $\P \subseteq \kappa_{\P} \cdot \M_G$, which for both $\P_G$ and $\P_G^{\bc}$ that we will consider later in this section are $O(1)$.
A function $f: \R_{\geq 0}^{E} \to \R_{-\infty}$ is \emph{coordinate-separable concave/linear} if $f$ is of the form $f(\bx) = \sum_{e \in E}f_e(\bx_e)$ where each $f_e: \R_{\geq 0} \to \R_{-\infty}$ is concave/linear.
Specifically, if $\ell: \R_{\geq 0}^{E} \to \R$ is coordinate-separable linear, we use $\boldsymbol{\ell} \in \R^{E}$ to denote the linear coefficients, i.e., $\ell(\bx) \defeq \boldsymbol{\ell}^\top \bx$.

\begin{definition}
For $\beta \geq 1$ and $0 \leq \zeta \leq 1$, an algorithm $\mathcal{A}$ is a \emph{$(\beta, T_{\mathcal{A}}, \zeta)$-oracle} for $(\P, f)$, where $\P \subseteq \R_{\geq 0}^{E}$ is convex and downward closed and $f$ is concave, if given any concave function $g$ of the form $g = f - \ell$ for some coordinate-separable linear $\ell$ such that each $\boldsymbol{\ell}_e$ is polynomially bounded, with an additional guarantee that
\begin{equation}
\max_{\bx \in \P}g(\bx) \geq \frac{1}{\poly(n)},
\label{eq:assumption-max-g}
\end{equation}
it outputs an $\bx_{g} \in \beta \cdot \mathcal{P}$ with $g(\bx_{g}) \geq \zeta \cdot \max_{\bx \in \mathcal{P}}g(\bx)$ in $T_{\mathcal{A}}(n, m) = \Omega(m)$ time.
\label{def:oracle}
\end{definition}

We defer the proof and discussion on modularization and explicit dependence on $\eps^{-1}$ of the following \cref{lemma:main-general} to \cref{appendix:entropy} for completeness.
It is worth noting that the two stated runtimes below are obtained by leveraging recently developed algorithms in certain components (specifically the computation of Gomory-Hu trees) of \cite{AhnG14}.

\begin{lemma}[adapted from \cite{AhnG14}]
  For any $\eps \geq \widetilde{\Omega}(n^{-1/2})$, downward closed and convex $\P \subseteq \R_{\geq 0}^{E}$, and concave function $f: \R_{\geq 0}^E \to \R_{-\infty}$ such that $\max_{\bx \in \P \cap \M_G}f(\bx) \geq 1$, given an $(\beta, T_{\mathcal{A}}, \zeta)$-oracle $\mathcal{A}$ for $\left(\mathcal{P}, f\right)$ where $\kappa_{\P} \cdot \beta < n$, there are randomized $\widetilde{O}\left(\left(T_{\A}\left(n, m\right) + m\eps^{-1} + n\eps^{-9}\right) \cdot \kappa_{\P} \beta\eps^{-4}\right)$ and $\widehat{O}\left(\left(T_{\A}\left(n, m\right) + m\eps^{-1}\right) \cdot \kappa_{\P} \beta\eps^{-4}\right)$ time algorithms that w.h.p. compute an $\bx_f \in \left(\beta \cdot \mathcal{P}\right) \cap \mathcal{M}_G$ such that $f(\bx_f) \geq \zeta (1-\eps)\max_{\bx \in \mathcal{P} \cap \mathcal{M}_G}f(\bx)$.
  \label{lemma:main-general}
\end{lemma}

For instance, if we set $\P$ to be $\P_G$,
then $\kappa_{\mathcal{P}} = O(1)$ since for any $\bx \in \mathcal{P}_G$ we have
\[
  \sum_{e \in E[B]}x_e \leq \frac{\sum_{v \in B}\sum_{e \in E_v}x_e}{2} \leq \frac{|B|}{2} \leq 2 \cdot \left\lfloor\frac{|B|}{2}\right\rfloor
\]
for all odd sets $B \in \O_G$.
As such, optimizing concave functions (specifically the entropy-regularized objective) over the general matching polytope reduces to optimizing over $\P_G$, which may be substantially simpler to solve.
In \cref{subsec:almost-linear-solver,subsec:near-linear-reduction} we will develop two incomparable algorithms that both follow \cref{lemma:main-general}.
These ultimately assemble into the two runtimes of \cref{thm:entropy-solver} in \cref{subsec:everything-entropy}.

\subsection{Almost-Linear Time Oracle via Convex Flow Algorithms} \label{subsec:almost-linear-solver}

Using the convex flow algorithm of \cite{ChenKLPGS22}, we can indeed optimize any ``efficiently-computable'' concave function over $\P_G$.
We then use that result to obtain an algorithm optimizing concave functions over the general matching polytope in \cref{lemma:almost-linear-solver}, with the help of \cref{lemma:main-general}.
The algorithm of \cite{ChenKLPGS22} requires as input self-concordant barriers on the domain $\{(x, y): y > h_e(x)\}$, where $h_e$ is the convex edge weight, that satisfy assumptions detailed in \cite[Assumption 10.2]{ChenKLPGS22} which ensures the numbers encountered during the algorithm are quasi-polynomially bounded.

\begin{lemma}[{\cite[Theorem 10.13]{ChenKLPGS22}} derandomized by {\cite{BrandCLPGSS23}}]
  Given $m$-edge graph $G = (V, E)$, demands $\bd \in \R^V$, convex $h_e: \R \to \R_{\infty}$, and $\nu$-self-concordant barriers $\psi_{e}(x, y)$ on the domain $\{(x, y): y > h_e(x)\}$ that satisfy \cite[Assumption 10.2]{ChenKLPGS22},
  there is a deterministic $m^{1+o(1)}$ time algorithm that computes a flow $\Bf \in \R^{E}$ with $\bB_G^\top \Bf = \bd$ such that
  \[
    h(\Bf) \leq \min_{\bB_G^\top \Bf^{*} = \bd}h(\Bf^{*}) + \exp(-\log^C m)
  \]
  for any fixed constant $C > 0$, where $h(\Bf) \defeq \sum_{e \in E}h_e(\Bf_e)$.
  \label{lemma:convex-flow}
\end{lemma}

The following is our main lemma in the section which turns the flow algorithm above into an algorithm for optimizing concave functions over $\M_G$.
For simplicity of the construction of barriers, we allow the functions $f_e$ to be decomposed into $\widetilde{O}(1)$ portions, each with its own barrier.

\begin{lemma}
  Suppose we are given a graph $G=(V,E)$, $\eps \geq \widetilde{\Omega}(n^{-1/2})$, and a coordinate-separable concave function $f: \R^{E}_{\geq 0} \to \R_{-\infty}$ such that $\max_{\bx \in \M_G}f(\bx) \geq 1$, where the function on each edge is given as $f_e(x) \defeq f_e^{(1)}(x) + \cdots + f_e^{(k_e)}(x)$ for some $k_e = \widetilde{O}(1)$, each equipped with a $\nu$-self-concordant barrier $\psi_{e}^{(i)}(x, y)$ on the domain $\{(x, y):  y > -f_e^{(i)}(x)\}$ that satisfy \cite[Assumption 10.2]{ChenKLPGS22}.
  Then, there is a randomized $\widehat{O}(m\eps^{-5})$ time algorithm that w.h.p. computes an $\bx \in \M_G$ such that
  \[
    f(\bx) \geq (1-\eps)\max_{\bx^{*} \in \M_G}f(\bx^{*}).
  \]
  \label{lemma:almost-linear-solver}
\end{lemma}

\begin{proof}
  By \cref{lemma:main-general}, to prove \cref{lemma:almost-linear-solver} it suffices to provide a $(1, m^{1+o(1)}, 1-\eps/2)$-oracle for $(\P_G, f)$.
  We achieve this by reducing the optimization problem over the relaxed matching polytope $\P_G$ to a minimum cost circulation\footnote{A circulation is a flow $\Bf$ such that routes the demand $\boldsymbol{0}$, i.e., $\bB^\top_G \Bf = \boldsymbol{0}$.} problem as follows.
  Consider the following directed graph $\widetilde{G} = (\widetilde{V}, \widetilde{E})$ with $\widetilde{V} \defeq \{v_{\text{in}}: v \in V\} \cup \{v_{\text{out}}: v \in V\}$ and
  \[ \widetilde{E} \defeq \{(v_{\text{in}}, v_{\text{out}}): v \in V\} \cup \{(u_{\text{out}}, v_{\text{in}}): \{u, v\} \in E\} \cup \{(v_{\text{out}}, u_{\text{in}}): \{u, v\} \in E\}. \]
  In other words, each vertex $v$ is split into two copies $v_{\text{in}}$ and $v_{\text{out}}$, and each edge $\{u, v\}$ is directed from $u_{\text{out}}$ to $v_{\text{in}}$ and from $v_{\text{out}}$ to $u_{\text{in}}$.
  For each $e = \{u, v\}$ in $E$, let $e^{(1)}$ and $e^{(2)}$ denote $(u_{\text{out}}, v_{\text{in}}) \in \widetilde{E}$ and $(v_{\text{out}}, u_{\text{in}}) \in \widetilde{E}$ respectively.
  Let $e^{(v)} \in \widetilde{E}$ for $v \in V$ be $(v_{\text{in}}, v_{\text{out}})$. 
  
  Let $g \defeq f - \ell$ be the coordinate-separable concave function that the oracle needs to optimize which we write as $g(\bx) \defeq \sum_{e \in E}g_e(\bx_e)$.
  We translate the concave edge weights $g_e$'s on $E$ to convex $h_{e}$'s on $\widetilde{E}$ as follows.
  For each $e \in E$, let $h_{e^{(1)}}, h_{e^{(2)}}: \R \to \R_{\infty}$ be defined as
  \[ h_{e^{(1)}}(x) = h_{e^{(2)}}(x) \defeq \frac{-g_e(x)}{2} + \phi(x) = \frac{-f_e(x)}{2} + \frac{\ell_e(x)}{2} + \phi(x), \]
  where
  \[
    \phi(x) =
    \begin{cases}
       0, & \text{if $0 \leq x \leq 1$} \\
       \infty, & \text{otherwise}
    \end{cases}
  \]
  is a convex regularizer which enforces each edge having flow at most one.
  For each $v \in V$, let $h_{e^{(v)}}: \R \to \R_{\infty}$ be defined as
  \[ h_{e^{(v)}}(x) \defeq \phi(x). \]

  Observe that we can translate between circulations in $\widetilde{G}$ and relaxed fractional matchings in $G$ as follows.
  For a circulation $\Bf \in \R^{\widetilde{E}}$ with $h(\Bf) < \infty$, by definition of the edge weights we may assume $0 \leq \Bf_{e} \leq 1$ for every $e \in \widetilde{E}$.
  Setting $\bx_e \defeq \frac{\Bf_{e^{(1)}} + \Bf_{e^{(2)}}}{2}$ for each edge $e \in E$, we see that by concavity of $g_e$ that $g_e(\bx_e) \geq -h_{e^{(1)}}(\Bf_{e^{(1)}}) - h_{e^{(2)}}(\Bf_{e^{(2)}})$.
  Conversely, for any relaxed fractional matching $\bx \in \P_G$, let $\Bf_{e^{(1)}} = \Bf_{e^{(2)}} \defeq \bx_e$ and $\Bf_{e^{(v)}} \defeq \bx(v)$.
  It can be easily checked that $\Bf$ is a circulation and has weight $h(\Bf) = -f(\bx)$.
  This shows that the optimal values of these two problems are the same (up to negation).
  Consequently, if we get a $\delta$-additive-approximate minimizing circulation $\Bf$, then the corresponding relaxed fractional matching $\bx$ would satisfy
  \[
    g(\bx) \geq -h(\Bf) \geq -\min_{\bB_G^\top\Bf^{*} = \boldsymbol{0}}h(\Bf^{*}) - \delta = \max_{\bx^{*} \in \P_G}g(\bx^{*}) - \delta.
  \]
  We will then choose $\delta$ to be sufficiently small to make $\bx$ a $(1-\eps/2)$-approximate maximizer of $g$ in $\P_G$.
  
  To apply the convex flow algorithm of \cref{lemma:convex-flow} to minimize $h(\Bf)$, we need to provide self-concordant barriers for the edge weights.
  Consider edges $e^{(1)}, e^{(2)} \in \widetilde{E}$ for some $e \in E$.
  The edge weights $h_{e^{(1)}}$ and $h_{e^{(2)}}$ can be divided into $k_e + 2$ parts which we handle by splitting $e^{(1)}$ and $e^{(2)}$ further into paths of length $k_e + 2$ (recall that $k_e$ the function $f_e$ is given to \cref{lemma:almost-linear-solver} as $k_e$ portions, each with its own barrier).
  Among the first $k_e$ parts, the $i$-th of which has weight $\frac{-f_e^{(i)}(x)}{2}$ which we use $\widetilde{\psi}(x, y) \defeq \psi_{e}^{(i)}(x, 2y)$ as the barrier, where recall that $\psi_{e}^{(i)}$ is the given barrier to the $i$-th portion of edge $e$; the second last part has weight $\frac{\ell_e(x)}{2}$ which we use $\widetilde{\psi}(x, y) \defeq -\log\left(y-\frac{\ell_e(x)}{2}\right)$ as the barrier;
  and the last part has weight $\phi(x)$ which we use $\widetilde{\psi}(x, y) \defeq x^{-\alpha} + (1-x)^{-\alpha}$ for $\alpha \defeq \frac{1}{1000\log mU}$ as the barrier.
  The barrier $-\log\left(y-\frac{\ell_e(x)}{2}\right)$ is the same barrier that \cite[Theorem 10.16]{ChenKLPGS22} used for linear functions (note that $\boldsymbol{\ell}_e$ is polynomially bounded); the barrier $x^{-\alpha} + (1-x)^{-\alpha}$ is the same barrier that \cite{ChenKLPGS22} used in their min-cost flow algorithm to enforce capacity constraints.
  Both of the barriers were shown to satisfy the assumption in \cite{ChenKLPGS22}.
  The barrier for $\phi(x)$ is also used for edges $e^{(v)} \in \widetilde{E}$.

  This gives us an exact reduction from maximizing concave weights over $\P_G$ to finding a circulation minimizing convex weights in a graph with $\widetilde{O}(m)$ edges.
  Applying \cref{lemma:convex-flow} with the constant $C$ chosen in such a way that $\exp(-\log^C m) \leq \eps/2 \cdot \max_{\bx \in \P_G}g(\bx) = \Theta(1/\poly(n))$
  thus results in a $(1, m^{1+o(1)}, 1 - \eps/2)$-oracle for $(\P_G, f)$.
  The theorem then follows from \cref{lemma:main-general} with accuracy parameter $\eps/2$.
\end{proof}

\subsection{Near-Linear Time Reduction to Linear Optimization} \label{subsec:near-linear-reduction}

To remove the $n^{o(1)}$ factors incurred by the almost-linear time flow algorithm in \cref{lemma:almost-linear-solver}, we can instead reduce concave function maximization over the matching polytope directly to a capacity-constrained weighted matching problem at the cost of a larger $\eps^{-1}$ factor in the running time.
For $\bc \in [0, 1]^{E}$, let $\Gamma_{\bc} \defeq \{0 \leq \bx_e \leq \bc_e,\;\forall\;e \in E\} \subseteq \R_{\geq 0}^{E}$ be the capacity-constrained polytope.
Let $\M_G^{\bc} \defeq \M_G \cap \Gamma_{\bc}$ and $\P_G^{\bc} \defeq \P_G \cap \Gamma_{\bc}$.

\begin{problem}
  In the \emph{capacity-constrained weighted matching} problem, we are given a graph $G = (V, E)$, an accuracy $\eps > 0$, edge weights $\bw \in \R_{\geq 0}^{E}$, capacities $\bc \in [0, 1]^{E}$, all polynomially bounded.
  The goal is find an $\bx \in \M_G^{\bc}$ such that
  \[ \bw^\top \bx \geq (1-\eps)\max_{\bx^{*} \in \M_G^{\bc}}\bw^\top \bx^{*}.
  \]
  \label{prob:capacitated-weighted-matching}
\end{problem}

To solve \cref{prob:capacitated-weighted-matching}, we use the following constant-approximate algorithm to the ``relaxed'' capacitated $\bb$-matching problem as an oracle and apply \cref{lemma:main-general}.
Their algorithm works for multigraphs with integral demands and capacities.

\begin{lemma}[{\cite[Theorem 13]{AhnG14}}]
  Given an $m$-edge multigraph $G=(V, E)$, edge weights $\bw \in \R_{\geq 0}^{E}$, demands $\bb \in \Z_{\geq 0}^{V}$, capacities $\bc \in \Z_{\geq 0}^E$, all polynomially bounded entrywise, there is a deterministic $\widetilde{O}(m)$ time algorithm that obtains a $1/8$-approximate maximizer to the following ``relaxed'' capacitated $\bb$-matching problem:
  \begin{equation}
    \begin{array}{llll}
    \mathrm{maximize} &  \displaystyle{\bw^\top\bx} \\
    \\
    \mathrm{subject\;to} & \bx(v) \leq \bb_v, & \forall\;v \in V, \\
    & 0 \leq \bx_e \leq \bc_e, & \forall\;e \in E.\\
    \end{array}
    \label{eq:capacitated-b-matching}
  \end{equation}
  \label{lemma:capacitated-b-matching}
\end{lemma}

The criterion of $\bb$ and $\bc$ being integral in \cref{lemma:capacitated-b-matching} can be relaxed via scaling.

\begin{corollary}
  For polynomially bounded $\bb \in \R_{\geq 0}^{V}$ and $\bc \in \R_{\geq 0}^{E}$
  there is an $\widetilde{O}(m)$ time algorithm that obtains a $1/16$-approximate maximizer to \eqref{eq:capacitated-b-matching}.
  \label{cor:capacitated-b-matching}
\end{corollary}

\begin{proof}
  Since $\bb$ and $\bc$ are polynomially bounded, we can scale them to integers by replacing each $\bb_v$ and $\bc_e$ with $\lfloor \bb_v / \bb_{\min} \rfloor$ and $\lfloor \bc_e / \bc_{\min} \rfloor$, respectively.
  Observe that $\lfloor \bb_v / \bb_{\min}\rfloor \geq \bb_v / (2\bb_{\min})$ and $\lfloor \bc_e / \bc_{\min}\rfloor \geq \bc_e / (2\bc_{\min})$, and thus we only lose a factor of $2$ in the approximation ratio from \cref{lemma:capacitated-b-matching}.
\end{proof}

Using \cref{cor:capacitated-b-matching} as an oracle, \cref{lemma:main-general} now implies the following algorithm for \cref{prob:capacitated-weighted-matching}.

\begin{lemma}
  There is a randomized $\widetilde{O}(m\eps^{-5} + n\eps^{-13})$ time algorithm for $\eps \geq \widetilde{\Omega}(n^{-1/2})$ that solves \cref{prob:capacitated-weighted-matching} w.h.p.
  \label{lemma:capacitated-weighted-matching}
\end{lemma}

\begin{proof}
  Consider the function $f_{\bw}(\bx) \defeq \bw^\top \bx$.
  By scaling $\bw$ we may assume $\max_{e \in E}\bw_e\bc_e \geq 1$ and therefore $\max_{\bx \in \P_{G}^{\bc} \cap \M_G}f_{\bw}(\bx) \geq 1$.
  By running \cref{cor:capacitated-b-matching} and returning the vector $\bx$ it outputs by $16$, we get a $(16, \widetilde{O}(m), 1)$-oracle for $(\P_{G}^{\bc}, f_{\bw})$.
  Since $\M_{G}^{\bc} = \P_{G}^{\bc} \cap \M_G$, the lemma follows from \cref{lemma:main-general}.
\end{proof}

Note that an immediate corollary of the above \cref{lemma:capacitated-weighted-matching} is that we can solve the subproblem in \cite{AssadiBD22} of finding an approximate maximum matching obeying capacity and odd-set constraints in $\widetilde{O}\left(m \cdot \poly(\eps^{-1})\right)$ time.
However, we remark again this does not suffice to make their framework run in $\widetilde{O}\left(m \cdot \poly(\eps^{-1})\right)$ completely, as a dual variable to \cref{prob:capacitated-weighted-matching} is still required to identify the set of critical edges along which the capacity is increased.

We now present the reduction from maximizing convex objective to \cref{prob:capacitated-weighted-matching} by approximating the objective with piecewise linear functions, thereby effectively splitting each edge into $\widetilde{O}(\eps^{-1})$ copies of different capacities and weights.
Similar approaches were used before, e.g., in \cite{MaiPRV17}.

\begin{lemma}
  Given an $n$-vertex $m$-edge graph $G = (V, E)$ and a coordinate-separable concave function $f: \R_{\geq 0}^{E} \to \R$ satisfying
  \begin{enumerate}[(1)]
    \item\label{item:poly-bounded} $f_e(x)$ is polynomially bounded for $x \geq 1/\poly(n)$,
    \item each $f_e$ can be evaluated in $O(1)$ time, and
    \item\label{item:maximizer-bounded} $z_e^{*} \defeq \argmax_{x \in [0, 1]}f_e(x)$ is given and satisfies $z_e^{*} \geq 1/\poly(n)$,
  \end{enumerate}
  for any $\eps \geq \widetilde{\Omega}(n^{-1/2})$ there is a randomized $\widetilde{O}(m\eps^{-6} + n\eps^{-13})$ time algorithm that computes an $\bx_f \in \M_G$ such that w.h.p.,
  \[
    f(\bx_f) \geq (1-\eps) \max_{\bx \in \M_G}f(\bx)\,.
  \]
  \label{lemma:near-linear-reduction}
\end{lemma}

\begin{proof}
  Let us consider a fixed $\eps^\prime = O(\eps)$ that we will set later.
  For each edge $e \in E$, let $r_e^{(0)} \defeq z_e^{*}$ and $r_e^{(i)} \defeq r_e^{(i-1)} / (1 + \eps^\prime)$ for $i \in \{1, \ldots, k\}$, where $k \defeq \left\lceil\log_{1+\eps^\prime}\frac{n^2}{\eps^\prime}\right\rceil = \widetilde{O}(\eps^{-1})$.
  Let $r_e^{(k+1)} \defeq 0$.
  Splitting each edge $e$ into $k + 1$ copies $e^{(0)}, e^{(1)}, \ldots, e^{(k)}$, we get a graph $G^\prime = (V, E^\prime)$ with $m^\prime = \widetilde{O}(m \eps^{-1})$ edges.
  Define $\bc^\prime \in \mathbb{R}^{E^\prime}$ and $\bw^\prime \in \mathbb{R}^{E^\prime}$ with $\bc^\prime_{e^{(i)}} \defeq r^{(i)}_e - r^{(i+1)}_e$ and $\bw^\prime_{e^{(i)}} \defeq \frac{f_e(r^{(i)}_e) - f_e(r^{(i+1)}_e)}{\bc^\prime_{e^{(i)}}}$ for each $e \in E$ and $i \in \{0, \ldots, k\}$.
  Note that $\bc_{e^{(i)}}$ and $\bw_{e^{(i)}}$ are both polynomially bounded by \labelcref{item:poly-bounded,item:maximizer-bounded}.
  Recall
  that $\M_{G^\prime}^{\bc^\prime}$ is the capacity-constrained matching polytope of $G^\prime$.
  We show the following two claims.
  
  \begin{claim}
  For any $\bx^\prime \in \M^{\bc^\prime}_{G^\prime}$, the vector $\bx \in \R_{\geq 0}^E$ given by $\bx_e \defeq \bx^\prime_{e^{(0)}} + \cdots + \bx^\prime_{e^{(k)}}$ satisfies $\bx \in \M_G$ and $f(\bx) \geq {\bw^{\prime}}^\top \bx^\prime$.
  \label{claim:w-to-g}
  \end{claim}
  \begin{proof}
  That $\bx \in \M_G$ is immediate from $\bx^\prime \in \M^{\bc^\prime}_{G^\prime}$.
  Let
  \[
    \widetilde{f_e}(x) \defeq \int_{0}^{x}\widetilde{\bw}_{e}(y)\;\mathrm{d}y, \quad\text{where}\quad\widetilde{\bw}_e(y) \defeq \bw^\prime_{e^{(i)}}\;\text{for $r^{(i+1)} \leq y < r^{(i)}$}
  \]
  be the piecewise linearized version of $f_e$.
  By concavity of $f_e$, it holds that $\widetilde{f_e}(x) \leq f_e(x)$ for all $x \in [0, z^{*}_e]$.
  We have
  \begin{align*}
    f(\bx) \geq \sum_{e \in E}\widetilde{f_e}(\bx_e) = \sum_{e \in E}\sum_{0 \leq i \leq k}\bw^\prime_{e^{(i)}} \cdot \min\left\{\bc^\prime_{e^{(i)}}, \max\left\{0, \bx_e - r^{(i+1)}_e\right\}\right\} \geq {\bw^\prime}^\top \bx^\prime,
  \end{align*}
  where the last inequality uses the fact that $\bw^\prime_{e^{(k)}} \geq \bw^\prime_{e^{(k-1)}} \geq \cdots \geq \bw^\prime_{e^{(0)}}$ by concavity of $f_e$ so it is always better to saturate $e^{(i)}$ before putting mass on $e^{(i-1)}$.
  \end{proof}
  
  \begin{claim}
  For any $\bx \in \M_G$, there exists an $\bx^\prime \in \M_{G^\prime}^{\bc^\prime}$ with ${\bw^\prime}^\top \bx^\prime \geq (1-2\eps^\prime)f(\bx)$.
  \label{claim:g-to-w}
  \end{claim}
  \begin{proof}
  Let $\widetilde{\bx} \in \R_{\geq 0}^{E^\prime}$ by defined as
  \[
    \widetilde{\bx}_{e^{(i)}} \defeq
      \begin{cases}
        \min\left\{\bc^\prime_{e^{(i)}}, \max\left\{0, \bx_e - r_{e}^{(i+1)}\right\}\right\}, & \text{if $i < k$,} \\
        \bc^\prime_{e^{(k)}}, & \text{if $i = k$.}
      \end{cases}
  \]
  Observe that $\widetilde{\bx}_{e^{(k)}} \leq \eps^\prime/n^2$ for every $e \in E$ by definition.
  As such we have $\widetilde{\bx}_{e^{(0)}} + \cdots + \widetilde{\bx}_{e^{(k)}} \leq \bx_e + \eps^\prime/n^2$, which implies $\widetilde{\bx} \in (1+\eps^\prime)\M_{G^\prime}^{\bc^\prime}$ given $\bx \in \M_G$.
  Letting $0 \leq t_e \leq k$ be the smallest integer such that $r^{(t_e)}_e \leq \max\{\bx_e, \bc^\prime_{e^{(k)}}\}$, we also have
  \begin{align*}
    {\bw^\prime}^\top \widetilde{\bx}
    \geq \sum_{e \in E}\sum_{i=t_e}^{k}\bw^\prime_{e^{(i)}}\bc^\prime_{e^{(i)}} \stackrel{(i)}{=} \sum_{e \in E \setminus F}f_e(r^{(t_e)}_e) &\stackrel{(ii)}{\geq} \frac{1}{1+\eps^\prime} \sum_{e \in E}f_e(\bx_e) \geq (1-\eps^\prime)f(\bx),
  \end{align*}
  where $(i)$ is by definition of $\bw^\prime$ and $\bc^\prime$, and $(ii)$ uses the fact that $r^{(t_e)}_e \geq \bx_e / (1+\eps^\prime)$ by definition and concavity of $f_e$.
  Therefore, the vector $\bx^\prime \defeq \widetilde{\bx}/(1+\eps^\prime)$ satisfies $\bx^\prime \in \M_G^{\bc^\prime}$ and ${\bw^\prime}^\top \bx^\prime \geq (1-2\eps^\prime)f(\bx)$.
  \end{proof}

  Going back to the proof of \cref{lemma:near-linear-reduction}, with $\eps^\prime \defeq \eps/4$, we run \cref{lemma:capacitated-weighted-matching} on the split graph $G^\prime$ with edge weights $\bw^\prime$ and capacities $\bc^\prime$ to accuracy $1-\eps^\prime$ in time $\widetilde{O}\left(m^\prime \eps^{-5} + n\eps^{-13}\right) = \widetilde{O}\left(m\eps^{-6} + n\eps^{-13}\right)$, obtaining an $\bx^\prime \in \M_{G^\prime}^{\bc^\prime}$.
  Letting $\bx \in \M_G$ be derived from $\bx^\prime$ as in \cref{claim:w-to-g}, it follows that
  \begin{align*}
    f(\bx) \geq {\bw^\prime}^\top \bx^\prime \geq (1-\eps^\prime)\max_{\bx^{\prime\prime} \in \M_{G^\prime}^{\bc^\prime}} {\bw^\prime}^\top \bx^{\prime\prime}
    &\geq (1-\eps^\prime)(1-2\eps^\prime)\max_{\bx^{*} \in \M_G} f(\bx^{*})
    &\geq (1-\eps)\max_{\bx^{*} \in \M_G} f(\bx^{*}),
  \end{align*}
  where we used \cref{claim:g-to-w}.
  This concludes the proof.
\end{proof}

\subsection{Putting Everything Together} \label{subsec:everything-entropy}

We can now prove \cref{thm:entropy-solver} by combining \cref{lemma:almost-linear-solver,lemma:near-linear-reduction}.

\begin{proof}[Proof of \cref{thm:entropy-solver}]
  Let $f_e(x) \defeq \bw_e x + \mu \cdot \bw_e x \log\frac{\gamma}{\bw_e x}$ so that the entropy-regularized matching objective is a coordinate-separable concave function $f(\bx) \defeq \sum_{e \in E}f_e(\bx_e)$.
  
  The runtime of $\widetilde{O}(m\eps^{-6} + n\eps^{-13})$ follows from \cref{lemma:near-linear-reduction}, since the function $f$ satisfies
  (1) $\bw_e x \leq f_e(x) \leq \bw_e + \mu \bw_e \log(\gamma/\bw_e)$ is polynomially bounded for $x \geq 1/\poly(n)$,
  (2) each $f_e$ can be evaluated in $O(1)$ time, and
  (3) $z_e^{*} = 1$ since $\mu \leq 1$.
  
  For the runtime of $\widehat{O}(m\eps^{-5})$, we use \cref{lemma:almost-linear-solver}, in which each $f_e(x)$ is decomposed into $f_e^{(1)}(x) + f_e^{(2)}(x)$, where $f_e^{(1)}(x) \defeq (\bw_e + \mu \bw_e \log\gamma)) x$ is linear and $f_e^{(2)}(x) \defeq -\mu \cdot \bw_e x \log \bw_e x$.
  We use the barriers $\psi_e^{(1)}(x, y) \defeq -\log(y + (\bw_e + \mu \bw_e\log\gamma) x)$ for $\{(x, y): y \geq -f_e^{(1)}(x)\}$ and $\psi_e^{(2)}(x, y) \defeq -\log(\bw_e x) - \log\left(y/\mu - \bw_e x\log (\bw_e x)\right)$ for $\{(x, y): y \geq -f_e^{(2)}(x)\}$.
  The barrier $\psi_e^{(1)}$ is the same as the one used in \cite[Theorem 10.16]{ChenKLPGS22} for linear functions, and $\psi^{(2)}$ is the same barrier that \cite[Theorem 10.16]{ChenKLPGS22} used for the entropy term.
  Both of the barriers were shown to satisfy \cite[Assumption 10.2]{ChenKLPGS22} as long as the coefficients are polynomially bounded (note that we apply an affine substitution for $\psi_{e}^{(2)}$ which preserves self-concordance~\cite[Proposition 3.1.1]{nesterov2003introductory}).
\end{proof}

\section{Dynamic Rounding Algorithms}\label{sec:rounding}

In \cref{sec:congestion-balancing,sec:decremental-fractional} we have shown how to solve the decremental fractional matching problem with $\widetilde{O}(\poly(\eps^{-1}))$ amortized update time and $\widetilde{O}(m \cdot \poly(\eps^{-1}))$ recourse.
Here we further show how to obtain an integral matching from the fractional one with dynamic rounding algorithms whose definition is recalled below.

\DefRounding*

Note that although the previous sections of our paper focus on the decremental setting, the rounding algorithm as defined in \cref{def:rounding} and given below in \cref{thm:rounding} is fully dynamic.
In other words, it works under arbitrary updates to the fractional matching $\bx$, irrespective of how the underlying fractional algorithm maintains $\bx$.
Our rounding algorithms are obtained by maintaining a sparse subgraph in which the maximum weight matching is approximately preserved.

\begin{definition}
  For a fractional matching $\bx \in \M_G$, a subgraph $H \subseteq \supp(\bx)$ is an \emph{$s$-sparse $\eps$-sparsifier}
  for $s \defeq s(n, m, \eps)$ of $G$ if $|H| \leq s \cdot \norm{\bx}_1$ and $M^{*}_{\bw}(H) \geq (1-\eps)\bw^\top \bx$. 
  We call a fractional matching $\bx^{(H)} \in \M_H$ a \emph{certificate} of $H$ if $\bw^\top \bx^{(H)} \geq (1-\eps)\bw^\top \bx$.
\end{definition}

Following standard techniques of periodic recomputation, if we can maintain an $\widetilde{O}(\poly(\eps^{-1}))$-sparse
$O(\eps)$-sparsifier under updates to $\bx$, then this gives the desired fully-dynamic rounding algorithm using the below static algorithm of \cite{DuanP14}.

\begin{proposition}[\cite{DuanP14}]
  There is an $\widetilde{O}(m\eps^{-1})$ time algorithm that given an $m$-edge graph $G=(V,E)$ weighted by $\bw \in \R_{\geq 0}^{E}$ computes a matching $M \subseteq E$ such that
  \[ \bw(M) \geq (1-\eps) \cdot \max_{\mathrm{matching}\;M^\prime \subseteq E}\bw(M^\prime).\]
\end{proposition}

The notion of sparsifier maintenance is formalized as follows.

\begin{definition}
  An algorithm $\S$ is an \emph{$(s, T_{\init}, T_{\update}, T_{\output})$-sparsifier-maintainer} if given a fractional matching $\bx$ and parameter $\eps > 0$, it initializes in $T_{\init}(n, m, \eps)$ time, processes each entry update to $\bx$ in $T_{\update}(n, m, \eps)$ amortized time, and outputs an $s$-sparse $\eps$-sparsifier $H$ of the current $\bx$ in $T_{\output}(n, m, \eps) \cdot |H|$ time.
\end{definition}

\begin{lemma}[rounding via sparsification]
  Given an $(s, T_{\init}, T_{\update}, T_{\output})$-sparsifier-maintainer $\S$, there is a dynamic rounding algorithm that initializes in $\widetilde{O}\left(m\cdot\eps^{-1} + T_{\init}(n, m, \eps/4)\right)$ time and maintains an integral matching $M \subseteq \supp(\bx)$ with $\bw(M) \geq \left(1-\eps\right)\bw^\top \bx$ in amortized
  \[ \widetilde{O}\left(T_{\update}(n, m, \eps/4) + s(n, m, \eps/4) \cdot W \cdot \left(\frac{T_\output(n, m, \eps/4)}{\eps} + \frac{1}{\eps^2}\right)\right) \]
  time per update to $\bx$.
  The dynamic rounding algorithm has the same properties as $\S$ does in terms of being deterministic/randomized and being fully/output-adaptive.
  \label{lemma:rounding-via-sparsifier}
\end{lemma}

\begin{proof}
  Let $M$ be an initial $\left(1-\frac{\eps}{2}\right)$-approximate maximum weight matching over $\supp(\bx)$ obtained by the static algorithm of \cite{DuanP14}.
  We initialize $\S$ and feed every update to $\bx$ to it to maintain an $\eps/4$-sparsifier $H$ of $\bx$ in $T_{\update}(n, m, \eps/4)$ time per update.
  If $\bx_e$ is set to zero for some $e \in M$, we remove $e$ from $M$.
  Every time $\bw(M) < (1-\eps)\bw^\top \bx$, we query $\S$ to get an $s(n,m,\eps/4)$-sparse $\eps/4$-sparsifier $H$ and re-compute $M$ as a $\left(1-\frac{\eps}{4}\right)$-approximate matching of $H$, again using \cite{DuanP14}.
  This step takes time
  \[ \widetilde{O}\left(T_{\output}(n, m, \eps/4) \cdot s(n, m, \eps/4) \cdot \norm{\bx}_1 + s(n, m, \eps/4) \cdot \norm{\bx}_1 \cdot \eps^{-1}\right). \]
  By definition of an $\eps/4$-sparsifier, we have $\bw(M) \geq \left(1-\frac{\eps}{4}\right)\left(1-\frac{\eps}{4}\right)\bw^\top \bx \geq \left(1-\frac{\eps}{2}\right)\bw^\top \bx$.
  The re-computation happens at most once every $\frac{\eps}{2W} \cdot \bw^\top \bx \geq \frac{\eps}{2W} \cdot \norm*{\bx}_1$ updates, as we need that many updates to either decrease $\bw(M)$ by $\frac{\eps}{2} \cdot \bw^\top \bx$ or increase the value of the maximum weight matching size by $\frac{\eps}{2} \cdot \bw^\top \bx$, so the time of re-computation amortizes to $\widetilde{O}\left(s(n, m, \eps/4) \cdot W \cdot \left(\frac{T_\output(n, m, \eps/4)}{\eps} + \frac{1}{\eps^2}\right)\right)$ time per update.
  Combined with the update time of $\S$, the lemma follows.
\end{proof}

In \cref{subsec:deterministic-sparsifier} we design
a determinstic $(\widetilde{O}(\eps^{-2}), \widetilde{O}(m), \widetilde{O}(W\eps^{-1}), O(1))$-sparsifier-maintainer in \cref{lemma:det-sparsifier}.
This together with \cref{lemma:rounding-via-sparsifier} proves the following theorem.

\Rounding*

In \cref{subsec:weighted rounding} we further obtain a rounding algorithm specifically for \cref{thm:decremental-fractional}, proving the following.

\begin{theorem}[informal, see \cref{thm:weighted-rounding}]
  We can round the entropy-regularized matching maintained in \cref{thm:decremental-fractional} in $\widetilde{O}((n^2/m) \cdot \eps^{-6})$ additional amortized time per update.
  \label{thm:weighted-rounding-informal}
\end{theorem}

We emphasize that while \cref{thm:weighted-rounding-informal} has a near-optimal dependence on $W$, it is not a generic dynamic rounding algorithm as defined in \cref{def:rounding}.
See \cref{subsec:overview:rounding} for a more detailed exposition on this.
We leave obtaining a generic, fully-dynamic weighted rounding algorithm that has polylogarithmic or even better dependence on $W$ as an important open question.
No such algorithms are known even for bipartite graphs.

\subsection{Deterministic Sparsifier}\label{subsec:deterministic-sparsifier}

To obtain a deterministic sparsifier, we employ the ``bit-by-bit'' rounding approach of \cite{BhattacharyaKSW23rounding} that iteratively sparsifies the support of $\bx$ while maintaining the degree value of each vertex.
For bipartite graphs, \cite{BhattacharyaKSW23rounding} showed that using this approach, we can directly round to an integral matching without resorting to periodic re-computation, e.g., as in \cref{lemma:rounding-via-sparsifier}.
However, it is unclear how to extend this approach to general graphs due to odd-set constraints.
Our key observation here is that, nevertheless, if we stop the bit-by-bit rounding process earlier, we still get a sparsifier on which we can then perform periodic re-computation.

\begin{lemma}
  For any $\eps > 0$,
  given a fractional matching $\bx \in \M_G$ and a vector $\widetilde{\bx} \in \R_{\geq 0}^{E}$ such that (1) $\supp(\widetilde{\bx}) \subseteq \supp(\bx)$, (2) $\bw^\top \widetilde{\bx} \geq (1-\eps/2)\bw^\top \bx$, (3) $\widetilde{\bx}(v) \leq \bx(v) + \eps/4$ for all $v \in V$, and (4) $\left|\bx_e - \widetilde{\bx}_e\right| \leq \frac{\eps^2}{144}$ 
  for all $e \in E(G)$, we can conclude that $H \defeq \supp(\widetilde{\bx})$ is an $\eps$-sparsifier of $\bx$ and it has a certificate $\bx^{(H)}$ satisfying $\bx^{(H)}_e = \Theta(\widetilde{\bx}_e)$ for all $e \in E$.
  \label{lemma:sparsifier-condition}
\end{lemma}

\begin{proof}
  Let $\eps^\prime \defeq \eps/4$ and let $\bx^{\prime} \defeq \frac{\widetilde{\bx}}{1+\eps^\prime}$ whose support is $\supp(\bx^{\prime}) = \supp(\widetilde{\bx}) = H \subseteq \supp(\bx)$.
  We have $\bx^{\prime}(v) = \frac{\widetilde{\bx}(v)}{1 +\eps^\prime} \leq 1$ for each $v \in V$.
  For each odd set $B \in \O_G$ of size $|B| \leq 3/\eps^\prime$, it follows from $|\bx_e - \widetilde{\bx}_e| \leq \frac{{\eps^\prime}^2}{9}$ that $\left|\bx(B) - \widetilde{\bx}(B)\right| \leq \frac{\eps^\prime}{3} |B|$.
  As such, we have $\bx^\prime(B) = \frac{\widetilde{\bx}(B)}{1+\eps^\prime} \leq \frac{\bx(B)+\eps^\prime|B|/3}{1+\eps^\prime} \leq \lfloor |B|/2\rfloor$.
  Therefore by \cref{fact:scaled-down} we have $\bx^{(H)} \defeq \frac{\bx^{\prime}}{1+\eps^\prime}$ satisfies $\bx^{(H)} \in \M_{G}$ and this fractional matching has weight $\bw^\top \bx^{(H)} \geq (1-\eps/2)(1-\eps^\prime)^2\bw^\top \bx \geq (1-\eps)\bw^\top \bx$.
  This shows that $H$ is indeed an $\eps$-sparsifier of $\bx$.
  That $\bx^{(H)}_e = \Theta(\widetilde{\bx}_e)$ is obvious from definition of $\bx^{(H)}$.
\end{proof}

The pipage-rounding algorithm of \cite{BhattacharyaKSW23rounding} is based on the following $\texttt{degree-split}$ algorithm for dividing a graph into a collection of cycles and paths.

\begin{proposition}[{\cite[Proposition 2.4]{BhattacharyaKSW23rounding}}]
  There exists an algorithm \emph{\texttt{degree-split}}, which on multigraph $G = (V, E)$ with maximum edge multiplicity at most two (i.e., no edge has more than two copies) computes in $O(|E|)$ time two (simple) edge-sets $E^{(1)}$ and $E^{(2)}$ of two disjoint sub-graphs of $G$, such that $E^{(1)}$, $E^{(2)}$ and the degrees $d_G(v)$ and $d^{(i)}(v)$ of $v$ in G and $H^{(i)} \defeq (V, E^{(i)})$ satisfy:
  \begin{enumerate}[(1)]
    \item\label{item:edge-set-size} $\left|E^{(1)}\right| = \left\lceil\frac{|E|}{2}\right\rceil$ and $\left|E^{(2)}\right| = \left\lfloor\frac{|E|}{2}\right\rfloor$ and
    \item\label{item:degree} $d^{(i)}(v) \in \left[\frac{d_G(v)}{2} - 1, \frac{d_G(v)}{2} + 1\right]$ for all $v \in V$.
  \end{enumerate}
\end{proposition}

The rounding algorithm of \cite{BhattacharyaKSW23rounding} works in a bit-by-bit fashion, and thus for $\bx \in [0, 1]^{E}$ we encode each entry $\bx_e$ as a binary string $\bx_e = \sum_{i}(\bx_e)_i \cdot 2^{-i}$.
The following observation suggests that it is without loss of optimality to focus only on the first $\widetilde{O}(\eps^{-1})$ bits of this encoding.

\begin{observation}[similar to {\cite[Observation 2.2]{BhattacharyaKSW23rounding}}]
  For an $\bx \in \M_G$ with $\bw^\top \bx \geq 1$ and accuracy parameter $\eps > 0$, the matching $\bx^\prime$ obtained from $\bx$ by zeroing out edges $e$ with $\bx_e < \frac{\eps}{2n^2W}$ and setting $(\bx^\prime_e)_i = 0$ for all $e \in E$ and $i > L$ for $L \defeq 1 + \left\lceil \log\frac{2n^2W}{\eps}\right\rceil$ satisfies $\bw^\top \bx^\prime \geq (1-\eps)\bw^\top \bx$.
  \label{obs:few-bits}
\end{observation}

\begin{proof}
  Direct calculation shows
  \[
    \bw^\top \bx^\prime \geq \bw^\top \bx - W \cdot \left(\binom{n}{2} \cdot \frac{\eps}{2n^2 W} + \sum_{e \in E}\sum_{i > L}2^{-i}\right) \geq \bw^\top \bx - \eps \geq (1-\eps)\bw^\top \bx.
  \]
\end{proof}

\begin{algorithm2e}[!ht]
  \caption{\textsc{DetSparsifier}} \label{alg:det-rounding}
  
  \SetEndCharOfAlgoLine{}
  \SetKwInput{KwData}{Input}
  \SetKwInput{KwResult}{Output}
  \SetKwInOut{State}{global}
  \SetKwProg{KwProc}{function}{}{}
  \SetKwFunction{Initialize}{Initialize}
  \SetKwFunction{Update}{Update}
  \SetKwFunction{Rebuild}{Rebuild}
  \SetKwFunction{Output}{Output}

  \State{edge weights $\bw \in \N^E$ and accuracy parameter $\eps \in (0, 1)$.}
  \State{fractional matching $\bx \in \M_G$.}
  \State{maximum and minimum level $L_{\max}, L_{\min} \in \N$.}
  \State{edge sets $E_i \subseteq \supp_{i}(\bx)$ and $F_i \subseteq E_i$ for $i \in \{L_{\min} + 1, \ldots, L_{\max}\}$.}
  \State{counters $c_i \in \N$ for $i \in \{L_{\min} + 1, \ldots, L_{\max}\}$.}
  
  \vspace{0.4em}

  \KwProc{\Initialize{$G=(V, E, \bw), \bx \in \M_G\;\text{such that}\;(\bx_e)_i = 0\;\text{for all $i > L$}, \eps \in (0, 1)$}} {
    Save $\bw$, $\eps$, and $\bx$ as global variables.\;
    Set $L_{\min} \gets \left\lceil \log\frac{2000L}{\eps^2} \right\rceil$ and $L_{\max} \gets L$.\;
    Call $\texttt{Rebuild}(L_{\max})$.\;
  }

  \vspace{0.4em}

  \KwProc{\Rebuild{$i$}} {
    \lIf{$i \leq L_{\min}$} {
      \textbf{return}.
    }

    $E_i \gets \supp_i(\bx)$ and $c_i \gets 0$.\;
    Let $(E^{(1)}, E^{(2)}) \gets \texttt{degree-split}(G[E_i \cup F_i])$.\;
    Set $F_{i-1}$ to be the $E^{(b)}$ for $b \in \{1, 2\}$ with larger weight, i.e., the one that satisfies $\bw(E^{(b)}) \geq \bw(E^{(3-b)})$.\;

    Call $\texttt{Rebuild(}i-1\texttt{)}$.\;
  }

  \vspace{0.4em}

  \KwProc{\Update{$e, \nu$ such that $(\nu)_i = 0$ for all $i > L_{\max}$}} {
    $\bx_e \gets \nu$.\;
    \For{$i = L_{\max}, L_{\max} - 1, \ldots, L_{\min} + 1$} {
      \lIf{$e \in E_i$} {
        remove $e$ from $E_i$.
      }
      \lIf{$e \in F_{i-1}$} {
        remove $e$ from $F_{i-1}$.
      }
      \lElse {
        remove one edge adjacent to each endpoint of $e$ from $F_{i-1}$ (if there is one).
      }
      $c_i \gets c_{i} + 1$.\;
      \lIf{$c_i > 2^{i-2} \cdot \frac{\eps \cdot \bw^\top \bx}{L_{\max} \cdot W}$} {
        call $\texttt{Rebuild}(i)$ and \textbf{return}.
      }
    }
  }

  \KwProc{\Output{}} {
    \textbf{return} $H \defeq \supp_{0}(\bx) \cup \cdots \cup \supp_{L_{\min}}(\bx) \cup F_{L_{\min}}$ as the sparsifier.\;
  }
\end{algorithm2e}

Our algorithm for deterministically maintaining a dynamic sparsifier is \cref{alg:det-rounding}, which is almost identical to \cite[Algorithm 2]{BhattacharyaKSW23rounding} except that we only run it until level $L_{\min}$ instead of $0$ (modulo rather straightforward changes needed for generalization to the weighted case).
Stopping the algorithm early disallows us from having the same ``straight-to-integral'' property as \cite[Algorithm 2]{BhattacharyaKSW23rounding} does but enables rounding in general graphs.

For simplicity of analysis, we depart from \cite{BhattacharyaKSW23rounding} in the implementation of $\texttt{Rebuild(}i\texttt{)}$ in that we are using a tail-recursion style implementation while they use a loop from $j = i$ to $L_{\min} + 1$.
These two implementations are equivalent, and we choose the current one to emphasize that a call to $\texttt{Rebuild(}i\texttt{)}$ causes a $\texttt{Rebuild(}i-1\texttt{)}$.

\paragraph{Notation from \cite{BhattacharyaKSW23rounding}.}
For $\bx \in \M_G$, let $\supp_{i}(\bx)$ be the set of edges $e$ whose $\bx_e$ has its $i$-th bit set to one, i.e., $\supp_i(\bx) \defeq \{e \in E: (\bx_e)_i = 1\}$.
Let $F_i(v)$ and $S_i(v)$ be the number of edges in $F_i$ and $\supp_i(\bx)$ incident to $v$ respectively.
Let
$F_i(e) \defeq \llbracket e \in F_i \rrbracket$ and $E_i(e) \defeq \llbracket e \in E_i \rrbracket$.
Let $\bx^{(i)} \in \R_{\geq 0}^{E}$ be given by
\begin{equation}
  \bx_e^{(i)} \defeq F_i(e) \cdot 2^{-i} + \sum_{j=0}^{L_{\min}}S_j(e) \cdot 2^{-j} + \sum_{j=L_{\min}+1}^{i}E_j(e) \cdot 2^{-j}.
\end{equation}
for $L_{\min} \leq i \leq L_{\max}$.

\begin{lemma}[similar to {\cite[Lemma 3.8]{BhattacharyaKSW23rounding}}]
  If $\bx \in \M_G$ holds for each update, then $\bx^{(i)}(v) \leq \bx(v) + \eps/2$ for all $L_{\min} \leq i \leq L_{\max}$ holds after each update.
  \label{lemma:degree}
\end{lemma}

\begin{proof}
  We adopt a similar proof strategy except that the $\texttt{degree-split}$ subroutine can only give us a weaker guarantee in non-bipartite graphs.
  By backward induction from $i = L_{\max}$ to $L_{\min} - 1$ we show that
  \begin{equation}
    F_i(v) \leq \left(2^i \cdot \sum_{j=i+1}^{L_{\max}}S_j(v) \cdot 2^{-j}\right) + L_{\max} - i.
    \label{eq:induction}
  \end{equation}
  The base case when $i = L_{\max}$ is trivial as the left-hand side is $0$ and the right-hand side is non-negative.
  Now consider an $i < L_{\max}$ and suppose first that a $\texttt{Rebuild}(i+1)$ happened right before the current moment.
  Then $E_{i+1}$ is set to $S_{i+1}$ and we have
  \begin{align*}
    F_i(v)
    \stackrel{(i)}{\leq} \frac{S_{i+1}(v) + F_{i+1}(v)}{2} + 1
    &\stackrel{(ii)}{\leq} \frac{1}{2}\left(S_{i+1}(v) + \left(2^{i+1} \cdot \sum_{j=i+2}^{L_{\max}}S_j(v) \cdot 2^{-j}\right)\right)  + L_{\max} - (i+1)+1\\
    &\leq \left(2^i \cdot \sum_{j=i+1}^{L_{\max}}S_j(v) \cdot 2^{-j}\right) + L_{\max} - i,
  \end{align*}
  where (i) is by Property~\labelcref{item:degree} and (ii) is by the inductive hypothesis.
  On the other hand, if the $\texttt{Update}(e, \nu)$ operation did not cause a $\texttt{Rebuild}(i+1)$, then the right-hand side might decrease by at most one as an edge incident to $v$ is removed.
  If $e \in F_{i}$, then the left-hand side was also decreased by one so the inequality holds.
  If $e \not\in F_{i}$, then we still decrease the left-hand side by one via removing an edge in $F_i$ adjacent to $v$ unless it is empty, for which the left-hand side was already $0$ so the inequality holds again.
  Observe that by definition
  \[
  \bx(v) = \sum_{j=0}^{L_{\max}}S_j(v) \cdot 2^{-j},
  \]
  and therefore
  \begin{align*}
    \bx^{(i)}(v) - \bx(v) &= F_{i}(v) \cdot 2^{-i} - \left(\sum_{j=L_{\min} + 1}^{L_{\max}}S_j(v)\cdot 2^{-j} - \sum_{j=L_{\min} + 1}^{i}E_j(v) \cdot 2^{-j}\right) \\
    &\stackrel{(i)}{\leq} L_{\max} \cdot 2^{-i} - \left(\sum_{j=L_{\min} + 1}^{i}S_j(v)\cdot 2^{-j} - \sum_{j=L_{\min} + 1}^{i}E_j(v) \cdot 2^{-j}\right) \\
    &\leq \eps/2,
  \end{align*}
  where (i) uses \eqref{eq:induction} and that $S_j(v) \geq E_j(v)$ for all $j$ since $E_j \subseteq \supp_j(\bx)$
  and (ii) is by our choice of $L_{\min}$ and $L_{\max}$.
\end{proof}

Let $\widetilde{\bx} \defeq \bx^{(L_{\min})}$.
In particular \cref{lemma:degree} shows that $\widetilde{\bx}(v) \leq \bx(v) + \eps/2$ for all $v \in V$.
Also observe that $\widetilde{\bx}$ is a vector supported on $H$, the sparsifier outputted by \cref{alg:det-rounding}.
The following helper lemmas analogous to ones in \cite{BhattacharyaKSW23rounding} can be shown, and to avoid repeating their arguments we defer the formal proofs to \cref{appendix:rounding-proofs}.

\begin{restatable}[analogous to {\cite[Lemma 3.10]{BhattacharyaKSW23rounding}}]{lemma}{Support}
  $\supp(\bx^{(i)}) \subseteq \supp(\bx)$ holds for each $L_{\min} \leq i \leq L_{\max}$ after each operation.
  \label{lemma:support}
\end{restatable}

\begin{restatable}[analogous to {\cite[Lemma 3.11]{BhattacharyaKSW23rounding}}]{lemma}{SizeOfMatching}
  $\bw^\top \widetilde{\bx} \geq (1-\eps)\bw^\top \bx$ holds after each operation.\footnote{Unlike \cite{BhattacharyaKSW23rounding}, we  assume the input matching $\bx$ is already preprocessed by \Cref{obs:few-bits}, and thus the bound stated here is $(1-\eps)$ instead of $(1-2\eps)$.}
  \label{lemma:size}
\end{restatable}

\begin{restatable}[analogous to {\cite[Lemma 3.13]{BhattacharyaKSW23rounding}}]{lemma}{UpdateTime}
  The amortized time per \emph{$\texttt{Update}(e, \nu)$} of \cref{alg:det-rounding} is $O(W \cdot \eps^{-1} \cdot L_{\max}^2)$. %
  \label{lemma:update-time}
\end{restatable}

Using these helper lemmas we show that \cref{alg:det-rounding} maintains a deterministic sparsifier.

\begin{lemma}
  There is a deterministic $(\widetilde{O}(\eps^{-2}), \widetilde{O}(m), \widetilde{O}(W \eps^{-1}), O(1))$-sparsifier-maintainer.
  \label{lemma:det-sparsifier}
\end{lemma}

\begin{proof}
  By preprocessing the input fractional matching $\bx$ to a $\bx^\prime$ with large-enough entries and low bit-complexity as in \cref{obs:few-bits} with accuracy parameter $\eps/2$, we may run \cref{alg:det-rounding} on $\bx^\prime$ with accuracy parameter $\eps / 2$ and $L \defeq 1 + \left\lceil\log\frac{4n^2W}{\eps}\right\rceil$.
  The initialization, update, and output time of \cref{alg:det-rounding} follow from definition and \cref{lemma:update-time}.
  Observe that since the rounding algorithm is only run until level $L_{\min}$, we have $|\bx^\prime_e - \widetilde{\bx}_e| \leq 2^{-L_{\min} + 1} \leq \frac{(\eps/2)^2}{144}$ for all $e \in E$.
  By \cref{lemma:support} we also have $\supp(\widetilde{\bx}) \subseteq \supp(\bx^\prime) \subseteq \supp(\bx)$.
  Therefore, together with \cref{lemma:degree,lemma:size}, \cref{lemma:sparsifier-condition} applies and the $H = \supp\left(\widetilde{\bx}\right)$ returned by \cref{alg:det-rounding} is an $\eps/2$-sparsifier of $\bx^\prime$ with certificate $\bx^{(H)}$, hence an $\eps$-sparsifier of $\bx$.
  Finally, as $\bx^{(H)}_e = \Theta(\widetilde{\bx}_e) \geq \Omega(2^{-L_{\min}}) = \widetilde{\Omega}(\eps^2)$, $H$ is $\widetilde{O}(\eps^{-2})$-sparse.
\end{proof}

\subsection{Randomized Degree Sparsifier}\label{subsec:randomized-sparsifier}

Toward getting the near-optimal weighted rounding algorithm for decremental dense graphs, in this section we consider the following notion of degree sparsifiers.
We remark that Property \labelcref{item:sparsifier-matching-size} below in fact is not used later in \cref{subsec:weighted rounding}.
Nevertheless, we choose to include it to make degree-sparsifier a stronger notion than the sparsifier in previous sections on unweighted graphs.\footnote{One can fairly directly extend Property \labelcref{item:sparsifier-matching-size} to be weighted, i.e., $\bw^\top \bx^{(H)} \geq (1-\eps)\bw^\top \bx$, at the cost of an $O(W)$ blow-up in the sparsity.}
This, as we will demonstrate later, also shows that the simple and standard sampling approach previously studied for bipartite graphs can also be used for rounding in general graphs.

\begin{definition}[Degree Sparsifiers]
  For a fractional matching $\bx \in \M_G$, a subgraph $H \subseteq \supp(\bx)$ is an \emph{$s$-sparse $\eps$-degree-sparsifier} for $s \defeq s(n, m, \eps)$ of $G$ if $|H| \leq s \cdot \norm{\bx}_1$ and there exists a fractional matching $\bx^{(H)} \in \M_G$ supported on $H$, which is similarly called a \emph{certificate} of $H$, such that
  \begin{enumerate}[(1)]
    \item\label{item:sparsifier-matching-size} $\norm*{\bx^{(H)}}_1 \geq (1-\eps)\norm*{\bx}_1$.
    \item\label{item:sparsifier-vertex} $\bx^{(H)}(v) \geq \bx(v) - \eps$ for all $v \in V$.
    \item\label{item:sparsifier-oddset} $\bx^{(H)}(B) \geq \bx(B) - \frac{\eps}{3} \cdot |B| \geq \bx(B) - \eps \cdot \left\lfloor |B|/2\right\rfloor$ for all odd sets $B \in \O_G$.
  \end{enumerate}
  \label{def:sparsifier}
\end{definition}

Given $\bx \in \M_G$, our randomized degree sparsifier $H$ is obtained from a simple sampling scheme that includes each edge $e$ into $H$ independently with probability
\begin{equation}
  \bp_e \defeq \min\left(1, \frac{\bx_e}{\gamma}\right) \quad \text{where}\;\gamma \defeq \frac{\eps^2}{320 \cdot d \cdot \ln{n}},
  \label{eq:p_e}
\end{equation}
where $d \defeq \max\{c, 1\}$ for $c > 0$ a given constant to \cref{lemma:sampling} below.

\begin{lemma}[Degree Sparsification]
  For any fractional matching $\bx \in \mathcal{M}_G$ with $\norm*{\bx}_1 \geq 1$ and constant $c > 0$, the random subgraph $H$ obtained by sampling each edge $e$ independently with probability $\bp_e$ defined in \eqref{eq:p_e}, is an $\widetilde{O}(\eps^{-2})$-sparse $\eps$-degree-sparsifier with probability at least $1 - n^{-c}$.
  \label{lemma:sampling}
\end{lemma}

To prove \cref{lemma:sampling}, we will use the following version of the Chernoff bound to show concentration on all the constraints.

\begin{proposition}[Chernoff Bound]
  Given $m$ independent random variables $X_1, \ldots, X_m$ in $[0, b]$, for $X = \sum_{i = 1}^{m}X_i$ and any $\delta > 0$ it holds that
  \[
    \Pr\left(\left|X - \mathbb{E}[X]\right| > \delta \cdot \mathbb{E}[X]\right) \leq 2\exp\left(\frac{-\delta^2 \cdot \mathbb{E}[X]}{(2 + \delta)b}\right)
    \,.
  \]
  \label{prpposition:chernoff}
\end{proposition}

Let $\bx^\prime \in \R^E$ be a random vector conditioned and supported on $H$ defined by
\[
\bx^\prime_e \defeq
  \begin{cases}
    \bx_e, & \text{if $\bx_e \geq \gamma$,}  \\
    \gamma, & \text{if $\bx_e < \gamma$ and $e \in H$,} \\
    0, & \text{if $\bx_e < \gamma$ and $e \in E \setminus H$.} \\
  \end{cases}
\]
Equivalently, for $\bx_e < \gamma$, $\bx_e^\prime$ is a random variable that takes value $\gamma$ with probability $\bp_e = \bx_e/\gamma$ and $0$ with probability $1-\bp_e$.
Therefore, $\bx^\prime$ is an unbiased estimator of $\bx$.
Let $F \defeq \{e \in E: \bx_e < \gamma\}$.
We prove the following claims about $\bx^\prime$.

\begin{claim}
  With probability at least $1 - 2 \cdot n^{-20d}$, we have $\left|\bx^\prime(v) - \bx(v)\right| \leq \frac{\eps}{2}$.
  \label{claim:chernoff-vertex-tight}
\end{claim}

\begin{proof}
  By definition, $\left|\bx^\prime(v) - \bx(v)\right| = \left|\norm*{\bx^\prime_{F_v}}_1 - \norm*{\bx_{F_v}}_1\right|$, and thus it suffices to bound the right-hand side.
  Using the Chernoff bound with $\delta \defeq \frac{\eps}{2\norm*{\bx_{F_v}}_1}$, we have
  \[
    \Pr\left(\left|\norm*{\bx^\prime_{F_v}}_1 - \norm*{\bx_{F_v}}_1\right| > \frac{\eps}{2}\right) \leq 2\exp\left(\frac{-\left(\frac{\eps}{2\norm*{\bx_{F_v}}_1}\right)^2 \cdot \norm*{\bx_{F_v}}_1}{\left(2 + \frac{\eps}{2\norm*{\bx_{F_v}}_1}\right)\frac{\eps^{2}}{320 \cdot d \cdot \ln{n}}}\right) \leq 2\exp\left(-20d\ln{n}\right) \leq 2 \cdot n^{-20d}
  \]
  using that $\bx_{F_v} \leq 1$.
\end{proof}

\begin{claim}
  With probability at least $1 - 2 \cdot n^{-2d|B|}$, we have $\left|\bx^\prime(B) - \bx(B)\right| \leq \eps \cdot \frac{|B|}{12}$.
  \label{claim:chernoff-oddset-tight}
\end{claim}

\begin{proof}
  Again $\left|\bx^\prime(B) - \bx(B)\right| = \left|\norm*{\bx^\prime_{F[B]}} - \norm*{\bx_{F[B]}}\right|$ and therefore we will bound the right-hand side.
  We set $\delta \defeq \frac{|B|\eps}{12\norm*{\bx_{F[B]}}_1}$ and apply the Chernoff bound, giving
  \begin{align*}
    \Pr\left(\left|\norm*{\bx^\prime_{F[B]}}_1 - \norm*{\bx_{F[B]}}_1\right| > \frac{|B|}{12} \cdot \eps\right)
    &\leq 2\exp\left(\frac{-\left(\frac{|B|\eps}{12\norm*{\bx_{F[B]}}_1}\right)^2 \cdot \norm*{\bx_{F[B]}}_1}{\left(2 + \frac{|B|\eps}{12\norm*{\bx_{F[B]}}_1}\right)\frac{\eps^{2}}{320 \cdot d \cdot \ln{n}}}\right) \\
    &\leq 2\exp\left(\frac{-26d
    \ln{n}|B|^2}{24\norm*{\bx_{F[B]}}_1 + |B|\eps}\right).
  \end{align*}
  Since $\norm*{\bx_{F[B]}}_1 \leq \frac{|B|}{2}$, we get
  \[
    \Pr\left(\left|\norm*{\bx^\prime_{F[B]}}_1 - \norm*{\bx_{F[B]}}_1\right| > \frac{|B|}{12} \cdot \eps\right) \leq 2n^{-2d|B|}.
  \]
\end{proof}

\begin{claim}
  With probability at least $1 - 2 \cdot n^{-32d}$, we have $\left(1-\frac{\eps}{2}\right)\norm*{\bx}_1 \leq \norm*{\bx^\prime}_1 \leq \left(1 + \frac{\eps}{2}\right)\norm*{\bx}_1$.
  \label{claim:chernoff-sum}
\end{claim}

\begin{proof}
  An application of the Chernoff bound shows that
  \begin{align*}
    \Pr\left(\left|\norm*{\bx^\prime}_1 - \norm*{\bx}_1\right| > \frac{\eps}{2} \norm*{\bx}_1 \right)
    &= \Pr\left(\left|\norm*{\bx^\prime_F}_1 - \norm*{\bx_F}_1 \right| > \frac{\eps}{2} \norm*{\bx}_1 \right) \\
    &\leq 2\exp\left(\frac{-\left(\frac{\eps \cdot \norm*{\bx}_1}{2 \cdot \norm*{\bx_F}_1}\right)^2\cdot \norm*{\bx_F}_1}{\left(2 + \frac{\eps \cdot \norm*{\bx}_1}{2 \cdot \norm*{\bx_F}_1}\right)\frac{\eps^2}{320 \cdot d \cdot \ln{n}}}\right) \\
    &\leq 2\exp\left(-32 \norm*{\bx}_1 \ln{n}\right) \leq 2n^{-32}
  \end{align*}
  using that $\norm*{\bx_F}_1 \leq \norm*{\bx}_1$ and $\norm*{\bx}_1 \geq 1$.
\end{proof}

\begin{proof}[Proof of \cref{lemma:sampling}]
  Take $\bx^{(H)} \defeq \frac{\bx^\prime}{1 + \eps/2}$ as the fractional matching over $H$.
  By union bounds over \cref{claim:chernoff-vertex-tight,claim:chernoff-oddset-tight,claim:chernoff-sum} (note that there are at most $n^k$ odd sets of size $k$), with probability $1 - n^{-d} \geq 1 - n^{-c}$ 
  we have $\bx^{(H)}(v) = \frac{\bx^\prime(v)}{1+\eps/2} \leq 1$ for all $v \in V$, $\bx^{(H)}(B) = \frac{\bx^\prime(B)}{1+\eps/2} \leq \lfloor |B|/2 \rfloor$ for all $B \in \O_G$,
  and $(1-\eps) \cdot \norm{\bx}_1 \leq \norm*{\bx^{(H)}}_1 \leq (1+\eps)\cdot \norm{\bx}_1$.
  This implies $\bx^{(H)} \in \M_G$ and shows Property \labelcref{item:sparsifier-matching-size}.
  \cref{claim:chernoff-vertex-tight} also implies that $\bx^{(H)}(v) = \frac{\bx^\prime(v)}{1+\eps/2} \geq \bx(v) - \eps$ which establishes Property \labelcref{item:sparsifier-vertex}.
  Similarly, \cref{claim:chernoff-oddset-tight} implies that $\bx^{(H)}(B) = \frac{\bx^\prime(B)}{1 + \eps/2} \geq \bx(B) - \frac{\eps}{3} \cdot |B|$
  which establishes Property \labelcref{item:sparsifier-oddset}.
  For the sparsity of $H$, note that $\bx^{(H)}_e \geq \Omega(\eps^2)$ and by $\norm*{\bx^{(H)}}_1 \leq O(\norm*{\bx}_1)$ we have $|H| \leq \widetilde{O}(\eps^{-2}) \cdot \norm{\bx}_1$.
  Finally, $H \subseteq \supp(\bx)$ is apparent by definition.
  This concludes the proof.
\end{proof}

Combining the above sampling scheme with the dynamic set sampler of \cite{BhattacharyaKSW23rounding}, we get the following algorithm for maintaining a degree sparsifier.

\begin{lemma}[{\cite[Theorem 5.2]{BhattacharyaKSW23rounding}}]
  There is an output-adaptive data structure that, given $\bp \in [0,1]^{d}$ with the guarantee that $\min_{i \in [d]: \bp_i \neq 0}\bp_i \geq 1/\poly(n)$ for each update, initializes in $O(d)$ time and supports (1) updating an entry $\bp_i$ in $O(1)$ time and (2) sampling a set $T \subseteq [d]$ in $O(1 + |T|)$ time such that each $i \in [d]$ is included independently with probability $\bp_i$.
  \label{lemma:dynamic-set-sampler} 
\end{lemma}

\begin{lemma}
  There is a randomized output-adaptive $(\widetilde{O}(\eps^{-2}), O(m), O(1), O(1))$-degree-sparsifier-maintainer for any $\eps \geq 1/\poly(n)$ that succeeds w.h.p.
  \label{lemma:rand-sparsifier}
\end{lemma}

\begin{proof}
  For a given $\eps > 0$, let $\eps^\prime \defeq \eps/2$.
  We can preprocess the matching $\bx$ to be $\bx^\prime$ by zeroing out edges with $\bx_e < \frac{\eps^\prime}{2n^2}$.
  By \cref{obs:few-bits}, we have $\norm*{\bx^\prime}_1 \geq (1-\eps^\prime)\norm*{\bx}$.
  Moreover, $|\bx^\prime(v) - \bx(v)| \leq \eps^\prime$ and $|\bx^\prime(B) - \bx(B)| \leq \eps^\prime$ hold for all $v \in V$ and $B \in \O_G$.
  Now we can simply run the dynamic sampler of \cref{lemma:dynamic-set-sampler}, with $\bp_e$ being the probability defined in \eqref{eq:p_e} with accuracy parameter $\eps^\prime$ and constant $c > 0$ chosen to make \cref{lemma:rand-sparsifier} succeed w.h.p.
  Note that as $\bx^\prime_e \geq \frac{\eps^\prime}{2n^2} \geq 1/\poly(n)$ for all $e$ with non-zero $\bx^\prime_e$, $\bp_e$'s are polynomially bounded and thus the runtime of \cref{lemma:dynamic-set-sampler} holds.
  Every time there is an update to $\bx_e$, we interpret that as an update to $\bx^\prime_e$ (by again zeroing out the coordinate if $\bx_e < \frac{\eps^\prime}{2n^2}$), and hence to $\bp_e$, and feed it to the dynamic sampler.
  Every time we are asked to output a sparsifier, we use the dynamic sampler to sample each edge independently with probability $\bp_e$, obtaining a subgraph $H$, in time $O(1 + |H|)$.
  By \cref{lemma:sampling}, $H$ is an $\widetilde{O}(\eps^{-2})$-sparse $\eps^\prime$-degree-sparsifier for $\bx^\prime$ w.h.p.
  This implies that $H$ is an $\eps$-degree-sparsifier for $\bx$ as well.
  Thus, we have $T_{\init}(n, m, \eps) = O(m)$, $T_{\update}(n, m, \eps) = O(1)$, and $T_{\output}(n, m, \eps) = O(1)$.
  The algorithm works against an output-adaptive adversary as \cref{lemma:dynamic-set-sampler} does.
\end{proof}

As alluded to at the beginning of the section, by including Property \labelcref{item:sparsifier-matching-size} into the definition, we get the following immediate corollary.
This can straightforwardly be generalized to the weighted setting by imposing a linear dependence on $W$, but we omit this part since it is irrelevant to this paper.

\begin{corollary}
  For $\eps \geq 1/\poly(n)$ there is a randomized output-adaptive $(\widetilde{O}(\eps^{-2}), O(m), O(1), O(1))$-sparsifier-maintainer for unweighted graphs that succeeds w.h.p.
  \label{cor:rand-sparsifier}
\end{corollary}

\subsection{Weighted Rounding for Entropy-Regularized Matching}\label{subsec:weighted rounding}

While the rounding algorithms presented previously work with weights, their dependence on $W$ is polynomial, which in general might incur a $\poly(n)$ or $\eps^{-O(1/\eps)}$ runtime overhead.
As such, here we further leverage the primal-dual properties of entropy-regularized matchings and provide a weighted rounding algorithm in the decremental setting that has near-optimal update time in dense graphs, proving Part (II) of \cref{thm:decremental-fractional}.
Before diving into the technical calculations that form the rest of the sections, we first give a short and intuitive explanation of this weighted rounding algorithm.
It leverages the following properties.

\begin{enumerate}[(1)]
  \item\label{item:delete-small-edges} For \emph{any} fractional matching $\bx_e$, deleting edges $e$ with $\bx_e = O(\eps/n)$ does not affect the weight of $\bx$ by more than an $O(\eps)$ factor.
  \item\label{item:degree-sparsification-suffices} For the \emph{entropy-regularized} fractional matching, a degree-sparsifier obtained from the unweighted rounding algorithm keeps most of its weight.
\end{enumerate}

As such, \labelcref{item:delete-small-edges} ensures that we can only consider edges with $\bx_e \geq \Omega(\eps/n)$.
This implies that any deletion to the integral matching that we maintain will also delete a comparable portion from the underlying fractional matching, which bounds the number of rebuilds of the integral matchings.
Rounding with nearly optimal dependence on $W$ is then made possible by \labelcref{item:degree-sparsification-suffices}.

With this intuition we now proceed to the proofs.
Let $G = (V, E)$ be the input graph to the decremental matching problem with edge weights $\bw \in \N^E$.
In the remainder of this section we consider a fixed set of input $G^{(t)} = (V, E^{(t)})$, $\mu$, and $\gamma$ to \cref{problem:entropy-matching}, where $E^{(t)} \subseteq E$ is the edge set in \cref{alg:meta-general-ds} for the $t$-th rebuild.
Fixing $t$, we let $\bx^{*} \defeq \bx^{\mu}_{E^{(t)},\gamma} \in \M_{G^{(t)}}$ be the unique optimal entropy-regularized fractional matching.
Note that we adopt the notations from \cref{sec:congestion-balancing} for the special case of $\X \defeq \M_G$, and thus the subscripts of certain notations are changed to $E$ from $S$.
For instance, $Z_{E^{(t)},\gamma}^{\mu}$ denotes the optimal value of the entropy-regularized matching problem on $G^{(t)}$ with parameters $\mu$ and $\gamma$.

\begin{lemma}\label{lemma:primal-dual relation}
    There exists a pair of dual solutions $(\by,\bz)\in\mathbb{R}^V\times\mathbb{R}^{\O_G}$ such that for every edge $e \in E^{(t)}$ it holds that
    \[\bx_{e}^*=2^{\frac{1}{\mu}-1-\frac{\byz_{e}}{\mu \cdot \bw_{e}}+\log \frac{\gamma}{\bw_e}},\]
    where $\byz_{e} \defeq \by_u+\by_v+\sum_{B \supseteq \{u, v\}}\bz_B$ for $e = \{u, v\}$. Further, the optimal objective value $Z_{E^{(t)},\gamma}^{\mu}$ satisfies
    \[Z_{E^{(t)},\gamma}^{\mu}=\mu \cdot \sum_{e \in E^{(t)}}\bw_e\bx^{*}_e+\sum_{v \in V}\by_v+\sum_{B\in\O_G}\bz_B \cdot \left\lfloor\frac{|B|}{2}\right\rfloor.\]
\end{lemma}
\begin{proof}
    Denote
    \[L(\bx,\by,\bz,\br)\defeq f_{E^{(t)},\gamma}^{\mu}(\bx)+\sum_{v\in V}\by_v\left(1-\sum_{e\in E^{(t)}_v}\bx_e\right)+\sum_{B\in\O_G}\bz_B\left(\left\lfloor\frac{|B|}{2}\right\rfloor-\sum_{e\in E^{(t)}[B]}\bx_e\right)+\sum_{e\in E^{(t)}}\br_e\bx_e\]
    as the Lagrangian of \cref{problem:entropy-matching} (see \eqref{eq:matching-polytope} for the definition of $\M_G$). Strong duality of \cref{problem:entropy-matching} (by Slater's condition) implies that
    \[Z_{E^{(t)},\gamma}^{\mu}=f_{E^{(t)},\gamma}^{\mu}(\bx^*)=\max_{\bx}\min_{\by,\bz,\br\geq \boldsymbol{0}}L(\bx,\by,\bz,\br)=\min_{\by,\bz,\br\geq \boldsymbol{0}}\max_{\bx} L(\bx,\by,\bz,\br).\]
    Suppose $(\by,\bz,\br)$ is the minimizer of the dual problem.
    From the stationarity of the KKT condition, the corresponding optimal primal solution $\bx^*$ satisfies that
    \begin{equation}
      \bx^*_{e}=2^{\frac{1}{\mu}-1-\frac{\byz_{e}}{\mu \cdot \bw_{e}}+\frac{\br_{e}}{\mu \cdot \bw_{e}}+\log\frac{\gamma}{\bw_e}}
      \label{eq:def-x-star}
    \end{equation}
    for every edge $e \in E^{(t)}$.
    From the complementary slackness of the KKT condition, we know $\br_e\bx^*_e=0$ for every $e \in E^{(t)}$. Since $\bx^* > \boldsymbol{0}$ from \eqref{eq:def-x-star}, we have $\br=\boldsymbol{0}$, and therefore
    \[\bx^*_{e}=2^{\frac{1}{\mu}-1-\frac{\byz_{e}}{\mu \cdot \bw_{e}}+\log\frac{\gamma}{\bw_e}}. \]
    Further, by plugging in $(\bx^{*}, \by, \bz, \br)$ to $L$ and expanding the formula, we can see that the optimal objective value can be written as
    \[Z_{E^{(t)},\gamma}^{\mu}=L(\bx^*,\by,\bz,\br)=\mu\cdot \sum_{e \in E^{(t)}}\bw_e \bx^*_e+\sum_{v \in V}\by_v+\sum_{B\in\O_G}\bz_B \cdot \left\lfloor\frac{|B|}{2}\right\rfloor.\]
\end{proof}

An immediate corollary of \cref{lemma:primal-dual relation} that will be used throughout is the following.

\begin{corollary}
  It holds that
  \[ Z_{E^{(t)},\gamma}^{\mu}\geq \sum_{v \in V}\by_v+\sum_{B\in\O_G}\bz_B \cdot \left\lfloor\frac{|B|}{2}\right\rfloor.\]
  \label{cor:lower-bound-opt}
\end{corollary}

For $\eps>0$ and $\bx \in \R_{\geq 0}^{E^{(t)}}$, let $E^{(t)}_{\eps}(\bx) \defeq \left\{e \in E^{(t)}\mid \bx_e \geq \frac{\eps}{3n}\right\}$.
Recall from \cref{tab:notation} that $\nu^*_{E^{(t)}} \defeq \max_{\bx\in \M_{G^{(t)}}}\bw^{\top}_{E^{(t)}}\bx$ which satisfies $\nu^{*}_{E^{(t)}} \leq nW$.

\begin{lemma}\label{lemma:covering of the edges}
    For $\eps>0$, $\nu^*_{E^{(t)}}\leq \gamma\leq m \cdot \nu^*_{E^{(t)}}$, and $\mu \leq \frac{\eps}{8\log(n^4W/\eps)}$, the optimal dual solution $(\by,\bz)$ approximately covers all edge weights, i.e., for every $e \in E^{(t)}$ it holds that $\byz_{e}\geq (1-\varepsilon/8)\bw_{e}$.
    Moreover, the covering is almost tight on the subgraph $E^{(t)}_{\eps}(\bx)$, i.e., $\byz_{e} \leq (1 + \eps/8)\bw_{e}$ holds for every $e \in E^{(t)}_{\eps}(\bx)$.
\end{lemma}

\begin{proof}
  From \cref{problem:entropy-matching} and the choice of $\gamma$, we have $\max_{e \in E^{(t)}}\bw_e \leq \nu^*_{E^{(t)}}\leq \gamma \leq m\cdot \nu^*_{E^{(t)}}\leq n^3 W$.
  \Cref{lemma:primal-dual relation} shows that 
  \[\bx^*_{e}=2^{\frac{1}{\mu}-1-\frac{\byz_{e}}{\mu \cdot \bw_{e}}+\log\frac{\gamma}{\bw_e}}.\]
  Since the optimal primal solution is feasible, we have $\bx^{*}_{e} \leq 1$, implying
  \[ \byz_{e}\geq (1-\mu)\bw_{e}+\mu \bw_e\log \frac{\gamma}{\bw_e}\geq \left(1-\frac{\eps}{8}\right)\bw_e\]
  for all $e \in E^{(t)}$ because $\gamma\geq \bw_e$ and $\mu\leq \eps/8$.
  Similarly, for all $e\in E^{(t)}_{\eps}(\bx)$, we have
  \[\byz_{e}\leq \left(1-\mu+\mu\log\left(\frac{n}{3\eps}\right)\right)\bw_{e}+\mu \bw_e\log \frac{\gamma}{\bw_e}\leq \left(1+\frac{\eps}{8}\right)\bw_{e},\]
  since $\gamma\leq n^3W$ and $\mu \leq \frac{\eps}{8\log(n^4W/\eps)}$.
\end{proof}

We argue that $E^{(t)}(\eps, \bx)$ for any fractional matching $\bx\in\M_{G^{(t)}}$ keeps most of the weight of $\bx$.

\begin{lemma}\label{lemma:structured subgraph}
    For any fractional matching $\bx \in \M_{G^{(t)}}$, accuracy parameter $\eps>0$, $\nu^*_{E^{(t)}}\leq \gamma\leq m \cdot \nu^*_{E^{(t)}}$, and $\mu\leq \frac{\eps}{8\log(n^4W/\eps)}$, we have
    \[\sum_{e \in E^{(t)}_{\eps}(\bx)}\bw_e\bx_e \geq \sum_{e \in E^{(t)}}\bw_e\bx_e- \frac{4\eps}{7} \cdot Z_{E^{(t)},\gamma}^{\mu}.\]
\end{lemma}

\begin{proof}
    For those edges in $E^{(t)}_{\eps}(\bx)$, we have
    \begin{align*}
        \sum_{e\in E^{(t)} \setminus E^{(t)}_{\eps}(\bx)}\bw_{e}\bx_{e}
        &\stackrel{(i)}{\leq} \frac{1}{1-\eps/8}\sum_{e\in E^{(t)} \setminus E^{(t)}_{\eps}(\bx)}\byz_{e} \cdot \bx_{e}\\
        &=\frac{1}{1-\eps/8}\sum_{v\in V}\by_v\norm*{\bx_{E^{(t)} _v\setminus E^{(t)}_{\eps}(\bx)}}_1+\frac{1}{1-\eps/8}\sum_{B\in\O_G}\bz_B\norm*{\bx_{E^{(t)} [B]\setminus E^{(t)}_{\eps}(\bx)}}_1\\
        &\stackrel{(ii)}{\leq} \frac{\eps}{3(1-\eps/8)}\sum_{v \in V}\by_v+\frac{\eps}{3(1-\eps/8)}\sum_{B\in\O_G}\bz_B\frac{|B|}{2}\\
        &\stackrel{(iii)}{\leq} \frac{\eps}{3(1-\eps/8)}\norm*{\by}_1+\frac{\eps}{2(1-\eps/8)}\sum_{B\in\O_G}\bz_B\left\lfloor\frac{|B|}{2}\right\rfloor\\
        &\stackrel{(iv)}{\leq} \frac{4\eps}{7} \cdot Z_{E^{(t)},\gamma}^{\mu},
    \end{align*}
    where (i) is by \Cref{lemma:covering of the edges} which shows that $\byz_{e}\geq (1-\eps/8)\bw_{e}$ for all $e \in E^{(t)}$, (ii) is because
    \[ \sum_{e \in E_v^{(t)} \setminus E^{(t)}_{\eps}(\bx)}\bx_e \leq \frac{\eps}{3}\quad\text{and}\quad\sum_{e \in E^{(t)}[B] \setminus E^{(t)}_{\eps}(\bx)}\bx_e \leq \frac{1}{2} \cdot \sum_{v \in B}\sum_{e \in E^{(t)}_v \setminus E^{(t)}_{\eps}(\bx)} \leq \frac{\eps}{3} \cdot \frac{|B|}{2}, \]
    (iii) is by $|B|\geq 3$, and (iv) is by \Cref{cor:lower-bound-opt}.
\end{proof}

Now, we consider a fractional matching $\bx \in \M_{G^{(t)}}$ that is close to $\bx^{*}$.
We show that a degree-sparsifier of $\bx$ preserves most of its weight.
This is perhaps surprising as the definition of degree-sparsifier (\cref{def:sparsifier}) is purely unweighted, while we can use it to sparsify this particular weighted matching.
\begin{lemma}\label{lemma:weighted sparsifier}
  For any accuracy parameter $\eps>0$, $\nu^*_{E}\leq \gamma\leq m \cdot\nu_{E}^*$, and $\mu \leq \frac{\eps}{64\log(8n^4W/\eps)}$,
  given an $\bx \in \M_G$ with
  \begin{equation}
  \sum_{e \in E^{(t)}}\bw_e\left|\bx_e - \bx^{*}_e\right| \leq \frac{\eps}{8}\cdot Z_{E^{(t)},\gamma}^{\mu},
  \label{eq:distance}
  \end{equation}
  for any edge subset $E^\prime\subseteq E^{(t)}$ such that
  \[
    \sum_{e \in E^\prime}\bw_e \bx_e \geq (1-\eps/8)\sum_{e \in E^{(t)}}\bw_e\bx_e,
  \]
  we have that any $\eps/8$-degree-sparsifier $H$ of $\bx_{E^\prime \cap E^{(t)}_{\eps/8}(\bx)}$ satisfies
  \[
    M^{*}_{\bw}(H) \geq \left(1-\frac{\eps}{2}\right)\cdot \sum_{e \in E^\prime}\bw_e\bx_e.
  \]
\end{lemma}
\begin{proof}
    Letting $\eps^\prime \defeq \eps/8$, we have $\mu \leq \frac{\eps^\prime}{8\log(n^4W/\eps^\prime)}$.
    Let $\bx^{(H)}$ be the certificate of $H$ being an $\eps/8$-degree-sparsifier.
    It suffices to prove that $\sum_{e \in E^\prime}\bw_e\bx_e^{(H)} \geq (1-\eps/2)\cdot \sum_{e \in E^\prime}\bw_e\bx_e$.
    From \Cref{lemma:structured subgraph} with accuracy parameter $\eps^\prime$, we have
    \begin{equation}
    \sum_{e \in E^\prime \setminus E^{(t)}_{\eps/8}(\bx)}\bw_e \bx_e \leq \frac{\eps}{14} \cdot Z_{E^\prime,\gamma}^{\mu} \leq \frac{\eps}{14} \cdot Z_{E^{(t)},\gamma}^{\mu}
    \quad\text{and}\quad
    \sum_{e \in E^\prime \setminus E^{(t)}_{\eps/8}(\bx^*)}\bw_e\bx^{*}_e \leq \frac{\eps}{14} \cdot Z_{E^{(t)},\gamma}^{\mu}.
    \label{eq:bound-small-edges}
    \end{equation}
    By \eqref{eq:distance}, \eqref{eq:bound-small-edges}, and triangle inequality we have
    \[
      \sum_{e \in E^\prime \setminus E^{(t)}_{\eps/8}(\bx^*)}\bw_e\bx_e \leq \sum_{e \in E^\prime \setminus E^{(t)}_{\eps/8}(\bx^*)}\bw_e\bx^{*}_e + \sum_{e \in E^\prime \setminus E^{(t)}_{\eps/8}(\bx^*)}\bw_e\left|\bx_e - \bx^{*}_e\right| \leq \frac{3\eps}{14} \cdot Z_{E^{(t)},\gamma}^{\mu}
    \]
    and thus
    \begin{equation}
      \sum_{e \in E^\prime \setminus \left(E^{(t)}_{\eps/8}(\bx) \cup E^{(t)}_{\eps/8}(\bx^*)\right)}\bw_e\bx_e \leq \frac{2\eps}{7} \cdot Z_{E^{(t)},\gamma}^{\mu}.
        \label{eq:norm}
    \end{equation}
    Applying \Cref{lemma:covering of the edges} with accuracy parameter $\eps^\prime$, for every edge $e\in E^{(t)} \cap E^{(t)}_{\eps/8}(\bx^*)$, we have $\byz_{e}\leq \left(1+\frac{\eps}{64}\right)\bw_{e}$.
    Therefore, letting $\widetilde{E} \defeq E^\prime \cap E^{(t)}_{\eps/8}(\bx)$ and $\widehat{E} \defeq \widetilde{E} \cap E^{(t)}_{\eps/8}(\bx^*)$ for clarity, we have
    \begin{align*}
      \sum_{e \in E^\prime}\bw_e\bx^{(H)}_e &\geq \sum_{e \in \widetilde{E}}\bw_e \bx^{(H)}_e \geq \left(1-\frac{\eps}{64}\right)\sum_{e\in\widetilde{E}}\byz_e \bx^{(H)}_e \\
      &= \left(1-\frac{\eps}{64}\right)\left(\sum_{v \in V}\by_v\norm*{\bx^{(H)}_{\widetilde{E}_v}}_1 + \sum_{B \in \O_G}\bz_B\norm*{\bx^{(H)}_{\widetilde{E}[B]}}_1\right) \\
      &\stackrel{(i)}{\geq} \left(1-\frac{\eps}{64}\right)\left(\sum_{v \in V}\by_v\norm*{\bx_{\widetilde{E}_v}}_1+\sum_{B \in \O_G}\bz_B\norm*{\bx_{\widetilde{E}[B]}}_1\right) - \frac{\eps}{8}\left(\norm*{\by}_1+\sum_{B \in \O_G}\bz_B \left\lfloor\frac{|B|}{2}\right\rfloor\right) \\
      &\geq \left(1-\frac{\eps}{64}\right)\left(\sum_{v \in V}\by_v\norm*{\bx_{\widehat{E}_v}}_1+\sum_{B \in \O_G}\bz_B\norm*{\bx_{\widehat{E}[B]}}_1\right) - \frac{\eps}{8}\left(\norm*{\by}_1+\sum_{B \in \O_G}\bz_B \left\lfloor\frac{|B|}{2}\right\rfloor\right) \\
      &\stackrel{(ii)}{\geq} \left(1-\frac{\eps}{64}\right)\left(1-\frac{\eps}{8}\right)\left(\sum_{e \in \widehat{E}}\bw_e\bx_e\right) - \frac{\eps}{8} Z_{E,\gamma}^{\mu} \\
      &\stackrel{(iii)}{\geq} \left(1-\frac{9\eps}{64}\right)\left(\sum_{e \in E^\prime}\bw_e\bx_e\right) - \frac{\eps}{6} Z_{E^{(t)},\gamma}^{\mu}\stackrel{(iv)}{\geq} \left(1-\frac{\eps}{2}\right)\left(\sum_{e \in E^\prime}\bw_e\bx_e\right),
    \end{align*}
    where (i) uses the assumption of $\bx^{(H)}$ being an $\eps/8$-degree-sparsifier of $\bx_{E^\prime \cap E^{(t)}_{\eps/8}(\bx)}$, (ii) uses \cref{lemma:covering of the edges,cor:lower-bound-opt}, and (iii) uses \eqref{eq:norm}.
    Finally, (iv) follows from
    \begin{align*}
      \sum_{e \in E^\prime} \bw_e \bx_e &\stackrel{(a)}{\geq} \left(1-\frac{\eps}{8}\right)\sum_{e \in E^{(t)}}\bw_e\bx_e \stackrel{(b)}{\geq} \left(1-\frac{\eps}{8}\right) \left(\sum_{e \in E^{(t)}}\bw_e \bx^{*}_e - \frac{\eps}{8} Z_{E^{(t)},\gamma}^{\mu}\right) \\
      &\stackrel{(c)}{\geq}\left(1-\frac{\eps}{8}\right)\left(\left(1-\frac{\eps}{16}\right)Z_{E^{(t)},\gamma}^{\mu}-\frac{\eps}{8}Z_{E^{(t)},\gamma}^{\mu} \right)
      \geq \left(1-\frac{5}{16}\eps\right)Z_{E^{(t)},\gamma}^{\mu},
    \end{align*}
    where (a) and (b) follow from the input assumption, and (c) is from \cref{lemma:absolute entropy bound} with accuracy parameter $\eps^\prime$.
    This concludes the proof.
\end{proof}

We are now ready to finish the proof of \cref{thm:weighted-rounding}.
For completeness we repeat the proof of \cref{thm:decremental-fractional} with a different set of parameters which allows us to perform rounding afterward.

\WeightedRounding*

\begin{proof}
  Let us assume $\eps \geq n^{-1/6}$ for general $G$ and $\eps \geq n^{-1}$ for bipartite $G$, as the update time of $\widetilde{O}(\eps^{-\expnear})$, $\widehat{O}(\eps^{-\expalmost})$, and $\widehat{O}(\eps^{-6})$ can be obtained by re-running the static algorithm of \cite{DuanP14} after every update.
  Letting $\X \defeq \M_G$,  \cref{thm:upper bound on reconstruction rounds approximate multiple phases} with accuracy parameter $\eps^\prime \defeq \eps/8$ and $\mu \defeq \frac{\eps^\prime}{128\log{((n^4W)/\eps^\prime)}}$ shows that by using \cref{thm:entropy-solver} with accuracy parameter $\frac{\mu{\eps^\prime}^2}{512}$ inside \cref{alg:rebuild} as the subroutine \texttt{Rebuild()}, there will be at most $\widetilde{O}(\eps^{-2})$ calls to $\texttt{Rebuild()}$ until the graph becomes empty.
  As such, we can maintain a $(1-\eps/8)$-approximate fractional maximum weight matching in amortized update times $\widetilde{O}(\eps^{-\expnear})$ and $\widehat{O}(\eps^{-\expalmost})$ for general $G$.
  As in the proof of \cref{thm:decremental-fractional}, for bipartite $G$, the update time improves to $\widehat{O}(\eps^{-2})$ by using \cite[Theorem 10.16]{ChenKLPGS22} instead of \cref{thm:entropy-solver}.

  Consider the $t$-th rebuild, where the current graph is $G^{(t)} = (V, E^{(t)})$.
  Until the next rebuild, let $E^\prime \subseteq E^{(t)}$ be the current edge set, i.e.,  the adversary has deleted $E^{(t)} \setminus E^\prime$ from the graph.
  To round the fractional matching into an integral matching, right after the $t$-th rebuild we run \cref{lemma:rand-sparsifier} to maintain an $\eps/8$-degree-sparsifier $H$ of $\bx_{E^\prime \cap E^{(t)}_{\eps/8}(\bx)}$, where $\bx \defeq \bx^{(t)}$ is the new fractional matching that we just got from $\texttt{Rebuild()}$.
  We then run the static algorithm of \cite{DuanP14} to get a $(1-\eps/4)$-approximate matching $M$ in $H$ in time $\widetilde{O}(n\eps^{-3})$, since $|H| \leq \widetilde{O}(n\eps^{-2})$ by the sparsity guarantee of \cref{lemma:rand-sparsifier}.
  We use $M$ as the output integral matching, and let $\nu$ be the value of $\bw(M)$ right after running \cite{DuanP14}.
  Until the next rebuild happens, for every deletion, we feed the update of $\bx_{E^\prime \cap E^{(t)}_{\eps/8}(\bx)}$ to \cref{lemma:rand-sparsifier} to maintain a degree sparsifier.
  Whenever the deletions make $\bw(M) < (1-\eps/8)\nu$, we query the degree-sparsifier-maintainer of \cref{lemma:rand-sparsifier} to get a new $H$ over which we reconstruct an integral $(1-\eps/4)$-approximate matching $M$ and its value $\nu$ using \cite{DuanP14}.

  We first analyze the quality of $M$.
  By \cref{lemma:approximate solution}, the $\bx^{(t)}$ returned by $\texttt{Rebuild()}$ satisfies\footnote{Note that there is a change of notation here, so $\sum_{e \in E^{(t)}}\bw_e\left|\bx_e^{(t)} - \left(\bx^{\mu}_{E^{(t)},\gamma}\right)_e\right|$ is the same as $\norm*{\bx^{(t)} - \bx^{\mu}_{S^{(t)},\gamma}}_{\bw,S^{(t)}}$ from \cref{sec:congestion-balancing}. See also \cref{tab:notation}.}
  \[ \sum_{e \in E^{(t)}}\bw_e\left|\bx^{(t)}_e - \left(\bx^{\mu}_{E^{(t)},\gamma}\right)_e\right| \leq \frac{\eps}{8}Z_{E^{(t)},\gamma}^{\mu}. \]
  By definition of \cref{alg:meta-general-ds}, until the next rebuild happens, we have
  \[ \sum_{e \in E^\prime}\bw_e\bx^{(t)}_e \geq \left(1-\frac{\eps^\prime}{2}\right)\sum_{e \in E^{(t)}}\bw_e\bx^{(t)}_e = \left(1 - \frac{\eps}{16}\right)\sum_{e \in E^{(t)}}\bw_e\bx^{(t)}_e. \]
  Since $\mu \leq \frac{\eps}{64\log(8n^4W/\eps)}$, the conditions of \cref{lemma:weighted sparsifier} are satisfied, showing that the $\eps/8$-degree-sparsifier $H$ of $\bx_{E^\prime \cap E^{(t)}_{\eps/8}(\bx)}$ we maintain indeed has $M^{*}_{\bw}(H) \geq \left(1-\frac{\eps}{2}\right)\left(\sum_{e \in E^\prime}\bw_e\bx^{(t)}_e\right)$.
  This implies that
  \begin{equation}
    \bw(M) \geq \left(1-\frac{\eps}{4}\right)\left(1-\frac{\eps}{2}\right)\left(\sum_{e \in E^\prime}\bw_e\bx^{(t)}_e\right) \geq \left(1-\frac{3\eps}{4}\right)\left(1-\frac{\eps}{16}\right)\sum_{e \in E^{(t)}}\bw_e\bx^{(t)}_e \geq \left(1-\frac{7\eps}{8}\right)M^{*}_{\bw}(G)
    \label{eq:quality-M}
  \end{equation}
  holds right after we run \cite{DuanP14} since $\bx^{(t)}$ was a $(1-\eps^\prime/2)$-approximate fractional matching after the rebuild.
  Because we re-run \cite{DuanP14} whenever $\bw(M) < \left(1 - \frac{\eps}{8}\right)\nu$, we have that $M$ is always a $(1-\eps)$-approximate maximum weight matching.

  We now analyze the additional update time spent in rounding.
  By \eqref{eq:quality-M}, whenever we have to re-run \cite{DuanP14}, we must have deleted a set of edges $D \subseteq E^{(t)} \setminus E^\prime$ from $M$ whose weights sum to at least
  $\sum_{e \in D}\bw_e \geq \frac{\eps}{8}\nu$.
  Because $M \subseteq E^{(t)}_{\eps/8}(\bx^{(t)})$, we have $\bx^{(t)}_e \geq \Omega(\eps/n)$ for all $e \in M$, and thus the adversary also deletes $\Omega(\nu \cdot \eps^2 / n)$ units of weight from $\bx^{(t)}$, i.e., $\sum_{e \in D}\bw_e\bx^{(t)}_e \geq \Omega(\nu \cdot \eps^2 / n)$.
  As a result, until the next rebuild of $\bx^{(t)}$ happens (i.e., when the weight of $\bx^{(t)}$ drops by an $\Theta(\eps)$ fraction), we will run \cite{DuanP14} at most $\widetilde{O}(n/\eps)$ times.
  Since there are $\widetilde{O}(\eps^{-2})$ rebuilds of $\bx^{(t)}$ by \cref{thm:upper bound on reconstruction rounds approximate multiple phases}, rounding incurs an $\widetilde{O}(n^2\eps^{-6})$ additional running time, which is amortized to $\widetilde{O}((n^2/m) \cdot \eps^{-6})$ time per update.
  This proves the update times of the algorithms.
  Finally, the algorithms are output-adaptive as \cref{lemma:rand-sparsifier} is.
\end{proof}

\section*{Acknowledgements}

Thank you to Aditi Dudeja for coordinating the posting to arXiv.
Thank you to Aaron Bernstein, Sayan Bhattacharya, Arun Jambulapati, Peter Kiss, Yujia Jin, Thatchaphol Saranurak, Kevin Tian, and David Wajc for helpful conversations at various stages of the project that ultimately led to this paper.
Part of the work for this paper was conducted while the authors were visiting the Simons Institute for the Theory of Computing.
Jiale Chen was supported in part by a Lawrence Tang Graduate Fellowship, a Microsoft Research Faculty Fellowship, and NSF CAREER Award CCF-1844855.
Aaron Sidford was supported in part by a Microsoft Research Faculty Fellowship, NSF CAREER Award CCF-1844855, NSF Grant CCF-1955039, a PayPal research award, and a Sloan Research Fellowship.
Ta-Wei Tu was supported in part by a Stanford School of Engineering Fellowship, a Microsoft Research Faculty Fellowship, and NSF CAREER Award CCF-1844855.

\bibliography{reference}

\appendix

\section{Generalizing \texorpdfstring{\cite{AhnG14}}{[AG14]}}\label{appendix:entropy}

As mentioned in \cref{sec:decremental-fractional}, \cite{AhnG14} can be generalized to prove \cref{lemma:main-general} in a fairly straightforward manner.
Their results, however, only claimed a $\poly(\eps^{-1})$ dependence in the runtime instead of an explicit one.
As we will swap out certain components in their algorithms with recently developed counterparts, for completeness we give a proof of \cref{lemma:main-general} with explicit dependence on $\eps^{-1}$.
Most of the proofs in this section are slight variants of analogous proofs in \cite{AhnG14}, and we do not claim novelty of them.

\paragraph{Overview.}
Here we outline the overall approach of \cite{AhnG14} to make it easier to understand what the subsequent lemmas and proofs are about.
To begin, note that by \Cref{fact:scaled-down} and concavity of the objective we can focus only on $\M_{G,\eps}$.
Since the constraints of the matching polytope $\mathcal{M}_{G,\eps}$ are linear and each $\bx \in \beta \cdot \P$ returned by the oracle only violates these constraints by a factor of $\kappa_{\P} \cdot \beta$, the main idea of \cite{AhnG14} is to use the multiplicative weights update (MWU) framework to produce a sequence of vectors in $\beta \cdot \P$ whose average only violates the constraints by a factor of $(1+O(\eps))$.
Here, the width of the MWU algorithms is $\kappa_{\P} \cdot \beta$.
As long as each vector in the sequence is an approximate maximizer of the objective, the average of them will be one as well (again we use the concavity of the objective).
However, classic MWU algorithms need to spend $\Omega(\nnz(\boldsymbol{A}))$ time in each iteration where $\boldsymbol{A}$ is the constraint matrix, which is exponential for the matching polytope even only considering small odd sets.
Thus, to circumvent to issue, \cite{AhnG14} showed that (1) a variation of the classic MWU algorithms where only constraints that are violated significantly are evaluated still works and (2) for (a slightly perturbed) $\M_{G,\eps}$, these maximally-violated constraints can be found and evaluated efficiently.
In what follows we adopt several lemmas from \cite{AhnG14}, reproving them analogously in the form we need if necessary.

\begin{lemma}[{\cite[Theorem 8(1)]{AhnG14}}]
  Suppose we are given an $(\beta, T_{\A}, \zeta)$-oracle $\A$ for $(\P, f)$, $\eps \geq \frac{1}{\poly(n)}$, concave function $f$, non-negative linear function $\ell$ where all $\boldsymbol{\ell}_e$'s are polynomially bounded, $\gamma_1 \geq \Omega(1)$ and $\gamma_2 \geq 0$.
  If
  \begin{equation}
    \{\bx \in \mathcal{P}: f(\bx) \geq \gamma_1\;\text{and}\;\ell(\bx) \leq \gamma_2\} \neq \emptyset,
    \label{eq:nonempty}
  \end{equation}
  then in $\widetilde{O}(T_{\A}(n,m))$ time we can find an $\bx \in \beta \cdot \mathcal{P}$ such that $f(\bx) \geq \zeta (1-\eps) \gamma_1$ and $\ell(\bx) \leq \gamma_2$.
  \label{lemma:bin-search}
\end{lemma}

\begin{proof}
  The proof strategy is the same as that of \cite[Lemma 6.1]{AhnG14}.
  Let $g_{\rho}(\bx) \defeq f(\bx) - \rho \cdot \ell(\bx)$.
  We can use a binary search to find $0 \leq \rho^{-} < \rho^{+} \leq \frac{\gamma_1}{\gamma_2}$ and $\bx_{\rho^{-}}, \bx_{\rho^{+}} \in \beta \cdot \P$ such that
  $\rho^{+} - \rho^{-} \leq \frac{\eps\gamma_1}{\gamma_2}$ with $\ell(\bx_{\rho^{-}}) > \gamma_2$ and $\ell(\bx_{\rho^{+}}) \leq \gamma_2$ in $\widetilde{O}\left(T_{\A}(n, m)\right)$ time as follows.
  For $\rho = 0$, the $\bx_0$ returned by $\A$ on input $g_{0}$ satisfies $f(\bx_0) \geq \zeta \gamma_1$.
  If $\ell(\bx_0) \leq \gamma_2$ then we can simply return $\bx_0$ as the solution.
  Similarly, for $\rho = \frac{\gamma_1}{\gamma_2}$, the vector $\bx_{\gamma_1/\gamma_2} \defeq \boldsymbol{0}$ satisfies $\ell(\bx_{\gamma_1/\gamma_2}) \leq \gamma_2$.
  Starting from the left endpoint $\rho^{-} = 0$ and the right endpoint $\rho^{+} = \frac{\gamma_1}{\gamma_2}$, we do a binary search on $\rho \in [\rho^{-}, \rho^{+}]$ using $\A$ to optimize $g_{\rho}$.
  Let $\bx_{\rho}$ be the returned vector on $g_{\rho}$.
  We maintain that $\ell(\bx_{\rho^{-}}) > \gamma_2$ and $\ell(\bx_{\rho^{+}}) \leq \gamma_2$ until $\rho^{+} - \rho^{-} \leq \frac{\eps\gamma_1}{\gamma_2}$, i.e., if $\ell(\bx_{\rho}) > \gamma_2$, then we set $\rho^{-} \gets \rho$; otherwise, we set $\rho^{+} \gets \rho$.
  This takes $\widetilde{O}(T_\A(n,m))$ time.
  To see that each $g_{\rho}$ is indeed a valid function for $\A$ to optimize, observe that by the guarantee of \eqref{eq:nonempty}, we have $\max_{\bx \in \P}g_{\rho}(\bx) \geq \gamma_1 - \rho\gamma_2$.
  Since the binary search is terminated when $\rho^{+} - \rho^{-} \leq \frac{\eps \gamma_1}{\gamma_2}$, we have $\gamma_1 - \rho\gamma_2 \geq \frac{\eps}{2}$ and therefore $\max_{\bx \in \P}g_{\rho}(\bx) \geq \frac{1}{\poly(n)}$ for each $\rho$ that we invoke the oracle.
  This verifies \eqref{eq:assumption-max-g}.
  In the end, we take a linear combination of $\bx_{\rho^{-}}$ and $\bx_{\rho^{+}}$ to get an $\bx \defeq (1-\alpha)\bx_{\rho^{-}} + \alpha\bx_{\rho^{+}}$ for some $\alpha \in [0, 1]$, which is in $\beta \cdot \P$ by convexity of $\P$, with $\ell(\bx) = \gamma_2$.
  This $\bx$ is then returned as the solution.
  
  To analyze the quality of $\bx$, we use by the guarantee of \eqref{eq:nonempty} that
  \begin{equation}
  g_{\rho^{-}}(\bx_{\rho^{-}}) \geq \zeta(\gamma_1 - \rho^{-}\gamma_2)\quad\text{and}\quad g_{\rho^{+}}(x_{\rho^{+}}) \geq \zeta(\gamma_1 - \rho^{+}\gamma_2).
  \label{eq:g}
  \end{equation}
  It then follows that
  \begin{align*}
    f(\bx)
    &\stackrel{(i)}{\geq} (1-\alpha)f(\bx_{\rho^{-}}) + \alpha f(\bx_{\rho^{+}})
    \stackrel{(ii)}{=} (1-\alpha)\left(g_{\rho^{-}}(\bx_{\rho^{-}}) + \rho^{-}\ell(\bx_{\rho^{-}})\right) + \alpha\left(g_{\rho^{+}}(\bx_{\rho^{+}}) + \rho^{+}\ell(\bx_{\rho^{+}})\right) \\
    &\stackrel{(iii)}{\geq} \zeta \left((1-\alpha)\left(\gamma_1 - \rho^{-}\gamma_2\right) + \alpha\left(\gamma_1 - \rho^{+}\gamma_2\right)\right) + (1-\alpha)\rho^{-}\ell(\bx_{\rho^{-}}) + \alpha \rho^{+}\ell(\bx_{\rho^{+}}) \\
    &\stackrel{(iv)}{\geq} \zeta \left(\gamma_1 - \alpha \cdot \eps \cdot \gamma_1 - \rho^{-}\gamma_2\right) + \ell(\left(1-\alpha)\rho^{-}\bx_{\rho^{-}} + \alpha\rho^{+}\bx_{\rho^{+}}\right) \\
    &\stackrel{(v)}{\geq} \zeta (1-\eps)\gamma_1 - \zeta \rho^{-}\gamma_2 + \rho^{-}\ell(\bx) + \alpha (\rho^{+} - \rho^{-})\ell(\bx_{\rho^{+}}) \geq \zeta (1-\eps)\gamma_1,
  \end{align*}
  where (i) is by concavity of $f$, (ii) is by definition of $g_{\rho}$, (iii) is by \eqref{eq:g}, (iv) is by $\rho^{+} - \rho^{-} \leq \frac{\eps\gamma_1}{\gamma_2}$ and linearity of $\ell$, and (v) is by definition of $\bx$ and again linearity of $\ell$.
  This proves the lemma.
\end{proof}

In addition to the oracle optimizing the objective, as described in the overview we also need to find constraints that are violated by the current solution significantly.
Let
\begin{equation}
\widetilde{\M_{G,\eps}} \defeq
\left\{
\begin{array}{ll}
\bx(v) \leq \widetilde{b}_v, & \forall\;v \in V \\
\bx(B) \leq \widetilde{b}_B, & \forall\;B \in \O_{G,\eps} \\
\end{array}
\right\} \cap \R_{\geq 0}^E, \quad\text{where}\;\widetilde{b}_v \defeq 1-4\eps\;\text{and}\;\widetilde{b}_B \defeq \left\lfloor\frac{|B|}{2}\right\rfloor - \frac{\eps^2 |B|^2}{4}
\end{equation}
be a slightly perturbed version of $\M_{G,\eps}$.
Formally, for $\bx \in \R_{\geq 0}^{E}$, let
\begin{equation}
  \lambda_{\bx} \defeq \max\left\{\max_{v \in V}\lambda_{\bx,v}, \max_{B \in \O_{G,\eps}}\lambda_{\bx,B}\right\},\;\text{where}\;\lambda_{\bx,v} \defeq \frac{\bx(v)}{\widetilde{b}_v}\;\text{and}\;\lambda_{\bx,B} \defeq \frac{\bx(B)}{\widetilde{b}_B}.
  \label{eq:lambda}
\end{equation}
An algorithm $\mathcal{B}$ is a \emph{$T_{\B}$-evaluator} for $\widetilde{\M_{G,\eps}}$ if given any vector $\bx \in \R_{\geq 0}^{E}$, when $\lambda_{\bx} > 1 + 8\eps$, it computes $\lambda_{\bx}$ and  the subset of small odd sets $\L_{\bx} \defeq \left\{B \in \O_{G,\eps}: \lambda_{\bx,B} \geq \lambda_{\bx} -\eps^{3}/10\right\}$ in $T_{\mathcal{B}}(n, m)$ time.
On the other hand, it correctly identifies when $\lambda_{\bx} \geq 1 + 8\eps$ in the same runtime.
Let $\lambda_{\bx,O_{G,\eps}} \defeq \max_{B \in \O_{G,\eps}}\left\{\lambda_{\bx,B}\right\}$.
The following theorem from \cite{AhnG14} characterizes the structure of $\L_{\bx}$.

\begin{proposition}[{\cite[Theorem 5]{AhnG14}}]
  If $\lambda_{\bx,\O_{G,\eps}} \geq 1 + 3\eps$, then $\L_{\bx}$ forms a laminar family.
\end{proposition}

We also give the following simple algorithm which helps us compute the multiplicative weights induced by a laminar family without any dependence on $\eps^{-1}$.

\begin{lemma}
  Given any $\bx \in \R_{\geq 0}^{E}$, a laminar family $\L \subseteq \O_{G,\eps}$, and value $p_B$ associate with each $B \in \L$, there is an $\widetilde{O}\left(m + \sum_{B \in \L}|B|\right)$ time algorithm that computes
  \[
    \boldsymbol{\ell}_e \defeq \sum_{B \in \L: u, v \in B}p_B
  \]
  for each edge $e = \{u, v\}$ in $E$.
  \label{lemma:compute-score}
\end{lemma}

\begin{proof}
  Let $\mathcal{T}$ be a tree on $\L \cup \{V\} \cup \{\{v\}: v \in V\}$ such that $U$ is an ancestor of $W$ if and only if $W \subseteq U$.
  By laminarity of $\L$ such a tree exists and can be constructed in $\widetilde{O}\left(m + \sum_{B \in \L}|B|\right)$ time.
  Observe that for each $e = \{u, v\} \in E$, the sets of $B \in \L$ containing both $u$ and $v$ correspond to a path from the root of $\mathcal{T}$ to the lowest common ancestor (LCA) of $\{u\}$ in $\{v\}$ in the tree.
  The LCAs of all edges can be computed in $\widetilde{O}(m)$ time \cite{GabowT83}, so can the sum of $p_B$'s of all root-to-vertex paths be populated in $O(n)$ time.
\end{proof}

The following claim is proven in \cite{AhnG14}.

\begin{claim}
  It holds that $\widetilde{\M_{G,\eps}} \subseteq \M_G \subseteq \frac{1}{1-4\eps} \cdot \widetilde{\M_{G,\eps}}$.
  \label{claim:perturbation}
\end{claim}

\begin{proof}
  Consider an $\bx \in \widetilde{\M_{G,\eps}}$, which clearly satisfies the degree constraints and odd set constraints of size at most $1/\eps$.
  For an odd set $B \in \O_G$ with $|B| > 1/\eps$, it follows that
  \[
    \bx(B) \leq \frac{1}{2}\sum_{v \in B}\bx(v) \leq \frac{(1-4\eps)|B|}{2} \leq \left\lfloor \frac{|B|}{2} \right\rfloor.
  \]
  That $\M_G \subseteq \frac{1}{1-4\eps} \cdot \widetilde{\M_{G,\eps}}$ follows from
  \[
    \left(1-4\eps\right) \cdot \left\lfloor\frac{|B|}{2}\right\rfloor \leq \left\lfloor \frac{|B|}{2}\right\rfloor - \frac{\eps^2|B|^2}{4}.
  \]
\end{proof}

\begin{algorithm2e}[!ht]
  \caption{Concave function optimization over matching polytope} \label{alg:mwu}
  
  \SetEndCharOfAlgoLine{}
  \SetKwInput{KwData}{Input}
  \SetKwInput{KwResult}{Output}
  \SetKwProg{KwProc}{function}{}{}

  \KwData{$n$-vertex $m$-edge graph $G=(V, E)$ and $\eps \in (3/\sqrt{n}, 1/16)$.}
  \KwData{convex, downward closed $\P \subseteq \R_{\geq 0}^{E}$ and concave $f: \R_{\geq 0}^{E} \to \R_{-\infty}$.}
  \KwData{$(\beta, T_{\A}, \xi)$-oracle $\A$ for $(\P, f)$ and $T_{\B}$-evaluator $\B$ for $\widetilde{\M_{G,\eps}}$.}

  \vspace{0.4em}
  Let $\delta \gets \frac{1}{8}$, $\alpha \gets 50\eps^{-3}\ln n$, and $\lambda_0 \defeq \kappa_{\P} \cdot \beta$.\;
  Use $\A$ on $f$ to find an $\bx \in \beta \cdot \P$ with $f(\bx) \geq \max_{\bx^{\prime} \in \P}f(\bx)$ and set $\OPT \defeq f(\bx)$. Note that $\lambda_{\bx} \leq \lambda_0$ by definition of $\kappa_\P$ and $\beta$. \label{line:initial}\;
  \While{true} {
    Compute $\lambda_{\bx}$ and $\L_{\bx}$ using the evaluator $\B$ in $T_{\B}(n, m)$ time.\footnotemark\;
    \lIf{$\lambda_{\bx} \leq 1 + 8\eps$} {
      \textbf{return} $\bx_f \defeq \frac{\bx}{1 + 8\eps}$.
    }
    Repeatedly set $\delta \gets \max\{2\delta/3, \eps\}$ until $\lambda_{\bx} > 1 + 8\delta$.\;
    Let $\bp_v \defeq \exp\left(\alpha \lambda_{\bx,v} - \alpha \lambda_{\bx}\right) / \widetilde{b}_v$ for $v \in V$ with $\lambda_{\bx, v} \geq \lambda_{\bx} - \eps^{3}/10$, and $\bp_v \defeq 0$ otherwise.\;
    Let $\bp_B \defeq \exp\left(\alpha \lambda_{\bx,B} - \alpha \lambda_{\bx}\right) / \widetilde{b}_B$ for $B \in \L$, and $\bp_B \defeq 0$ for $B \not\in \L_{\bx}$.\;  
    Let $\gamma_{\bx} \defeq \sum_{v \in V}\bp_v \widetilde{b}_v + \sum_{B \in \O_{G,\eps}}\bp_B \widetilde{b}_B$.\;
    Compute $\boldsymbol{\ell}_e \defeq p_u + p_v + \sum_{B \in \L_{\bx}: u, v \in B}p_B$ for each $e = \{u, v\}$ in $\widetilde{O}(m)$ time by \cref{lemma:compute-score}, using that $\L_{\bx}$ forms a laminar family.\label{line:compute-score}\;
    Use \cref{lemma:bin-search} with $f$, $\gamma_1 \defeq \OPT$, $\ell(\bx^\prime) \defeq \boldsymbol{\ell}^\top\bx^\prime$, and $\gamma_2 \defeq \gamma_{\bx}$ to find an $\widetilde{\bx} \in \beta \cdot \P$ such that $f(\widetilde{\bx}) \geq \zeta (1-\eps)\OPT$ and $\ell(\widetilde{\bx}) \leq \gamma_{\bx}$ in $\widetilde{O}(T_{\A}(n, m))$ time.\;
    \While{\cref{lemma:bin-search} failed to find such an $\widetilde{\bx}$} {
      $\OPT \gets (1-\eps)\OPT$ and re-run \cref{lemma:bin-search} with the new $\OPT$.\;
    }
    Update $\bx$ by $\bx_e \gets (1-\sigma)\bx_e + \sigma \widetilde{\bx}_e$ where $\sigma \defeq \frac{\delta}{4\alpha\lambda_0}$.\;
  }
\end{algorithm2e}

We now present the algorithm for concave optimization over the matching polytope.
\cref{alg:mwu} is a specialization of \cite[Algorithm 2]{AhnG14} to the case that the constraint matrices $\boldsymbol{A}$ and $\bb$ correspond to that of the perturbed matching polytope $\widetilde{\M_{G,\eps}}$, and then modified to work for concave function optimization.
The algorithm is divided into \emph{superphases}, where each superphase contains several \emph{phases}, and each phase contains several \emph{iterations}.
The only real difference between \cref{alg:mwu} and \cite[Algorithm 1]{AhnG14} is that we substitute the subroutine optimizing linear objective (e.g., the solver for LP8 in \cite{AhnG14}) to an oracle capable of optimizing general concave objective.
Other changes to the algorithm are for consistency with the notations and terminologies that we are using throughout the paper and to highlight when each superphase, phase, and iteration begins and ends.
What may appear confusing at first is that \cite{AhnG14} used $\delta$ for the input accuracy parameter and $\eps$ as the variable used in their algorithm that gradually decreases from $1/8$ to $\delta$.
Instead, we use $\eps$ for the accuracy parameter, and thus the roles of these two symbols are interchanged.
Also, \cite[Algorithm 2]{AhnG14} has other degrees of freedom where they can choose a function $f(\eps) < \eps$ (which is not to be confused with the concave function $f$ in our context) and $\alpha \leq \frac{1}{f(\eps)}\ln\frac{M\lambda_0}{\eps}$.
Nevertheless, for simplicity, we fix $f(\eps) \defeq \eps^3 - 10$ and $\alpha \defeq 50\eps^{-3}\ln n$ as in the uncapacitated $\bb$-matching algorithm in \cite{AhnG14}.

\begin{lemma}[{\cite[Algorithm 2, Theorem 16, and Lemma 18]{AhnG14}}]
  For any $\eps \geq \widetilde{\Omega}(n^{-1/2})$, downward closed, convex $\P \subseteq \R_{\geq 0}^{E}$, concave function $f: \R_{\geq 0}^E \to \R$ such that $\max_{\bx \in \P \cap \M_G}f(\bx) \geq 1$, given an $(\beta, T_{\A}, \zeta)$-oracle $\mathcal{A}$ for $(\P, f)$ and a $T_{\B}$-evaluator $\mathcal{B}$ for $\widetilde{\M_{G,\eps}}$, if $\kappa_{\P} \cdot \beta < n$, then \cref{alg:mwu} takes
  \[ \widetilde{O}\left(\left(T_{\A}(n, m) + T_{\B}(n, m)\right) \cdot \beta\kappa_\P\eps^{-4}\right) \]
  time and computes an $\bx_f \in \left(\beta \cdot \P\right) \cap \M_G$ such that $f(\bx_f) \geq \zeta(1-13\eps)\max_{\bx \in \P \cap \M_G}f(\bx)$.
  \label{lemma:main-appendix}
\end{lemma}

\begin{proof}
  Let us assume $\eps \leq 1/16$.
  Let $f^{*} \defeq \max_{\bx \in \P \cap \M_G}f(\bx)$ and $\bx^{*} \defeq \argmax_{\bx \in \P \cap \M_G}f(\bx)$.
  The vector $\bx$ in \cref{alg:mwu} is in $\beta \cdot \P$ at all times, and thus $\bx_f \in \beta \cdot \P$ as well.
  \cref{lemma:bin-search} never fails after $\OPT$ decreases to $(1-4\eps)f^{*}$: we can take $\bx^{*} \cdot (1-4\eps)$ with $f(\bx^{*} \cdot (1-4\eps)) \geq (1-4\eps) f^{*}$ and by \cref{claim:perturbation}, $\bx^{*} \cdot (1-4\eps) \in \P \cap \widetilde{\M_{G,\eps}}$ and in particular $\ell\left(\bx^{*} \cdot (1-4\eps)\right) \leq \gamma$ is guaranteed.
  This shows that indeed the value of $\gamma_1$ in \cref{lemma:bin-search} is $\Omega(1)$ as $f(\bx^{*}) \geq \Omega(1)$.
  Also, all $\widetilde{\bx}$'s returned by \cref{lemma:bin-search} have $f(\widetilde{\bx}) \geq \zeta (1-\eps)(1-4\eps)f^{*}$, proving $f(\bx) \geq \zeta (1-\eps)(1-4\eps)f^{*}$ and therefore $f(\bx_f) \geq \frac{\zeta (1-\eps)(1-4\eps)}{1+8\eps}f^{*} \geq \zeta (1-13\eps)f^{*}$ by concavity of $f$.
  At the end of the algorithm, we have that $\lambda_{\bx} \leq 1 + 8\eps$ and thus $\bx_{f} \in \widetilde{\M_{G,\eps}} \subseteq \M_G$ by \cref{claim:perturbation}.
  The numbers $\bp_v$'s are polynomially bounded since $\lambda_{\bx, v} \geq \lambda_{\bx}-\eps^3/10$ and thus $n^{-5} \leq \exp(\alpha \lambda_{\bx,v} - \alpha\lambda_{\bx}) \leq 1$.
  Similarly $\bp_B$'s are polynomially bounded.
  This ensures that $\ell(\bx) \defeq \boldsymbol{\ell}^\top \bx$ is a valid input to \cref{lemma:bin-search}.
  
  It thus remains to bound the number of iterations in \cref{alg:mwu} until it terminates since each iteration takes $O(T_{\A}(n,m) + T_{\B}(n, m))$ time.
  Let $\Phi_{\bx} \defeq \sum_{v \in V}\exp(\alpha \lambda_{\bx,v}) + \sum_{B \in \O_{G,\eps}}\exp(\alpha \lambda_{\bx,B})$ as defined in \cite[Definition 4]{AhnG14}.
  \cite[Theorem 16]{AhnG14} showed that \cite[Algorithm 2]{AhnG14} converges in $\widetilde{O}(\lambda_0 \cdot (\eps^{-2} + \alpha \eps^{-1}))$ iterations if we start with $\Phi_{\bx} \leq \gamma_{\bx} + \frac{\eps \gamma_{\bx}}{\lambda_0}$.
  This can be verbatim carried to \cref{alg:mwu} since the analysis is oblivious to the function $f$ we are optimizing: it applies as long as the vector $\widetilde{\bx}$ we found in each iteration satisfies $\boldsymbol{\ell}^\top \widetilde{\bx} \leq \gamma_{\bx}$, which is guaranteed by \cref{lemma:bin-search}.
  Roughly speaking, in \cite[Lemma 15]{AhnG14} they showed that $\Phi_{\bx}$ decreases after every iteration if the initial condition $\Phi_{\bx} \leq \gamma_{\bx} + \frac{\eps \gamma_{\bx}}{\lambda_0}$ is satisfied.
  The proof only used how $\Phi_{\bx}$ can change for the new $\bx$ given by $\bx_e \gets (1-\sigma)\bx_e + \sigma \widetilde{\bx}_e$ via the fact that $\boldsymbol{\ell}^\top \widetilde{\bx} \leq \gamma_{\bx}$.
  \cite[Theorem 16]{AhnG14} then used \cite[Lemma 15]{AhnG14} to argue the total decrease of $\Phi_{\bx}$ after $\widetilde{O}(\lambda_0 \cdot (\eps^{-2} + \alpha \eps^{-1}))$.
  Finally, we use \cite[Lemma 18]{AhnG14} which showed that indeed for the specific value of $\alpha$ that we choose, when $\lambda_0 < n$ the initial condition $\Phi_{\bx} \leq \gamma_{\bx} + \frac{\eps \gamma_{\bx}}{\lambda_0}$ is satisfied.
  This shows that the number of iterations \cref{alg:mwu} has is $\widetilde{O}(\kappa_{\P}\beta\eps^{-4})$.
\end{proof}

\subsection{Finding Maximally-Violated Constraints}

It remains to give an evaluator for $\widetilde{\M_{G,\eps}}$ that finds all the maximally-violated odd sets.
\cite[Lemma 17]{AhnG14} reduces this to the case where $\lambda_{\bx,\O_{G,\eps}} \geq 1 + 3\eps$.
In this regime, \cite[Theorem 5]{AhnG14} applies (recall that it says $\L_{\bx}$ forms a laminar family), and \cite[Theorem 6]{AhnG14} presents an $\widetilde{O}(m + n \cdot \poly(\eps^{-1}))$ algorithm.
We briefly sketch the algorithm and analyze its dependence on $\eps^{-1}$ in the remainder of the section.

The algorithm uses a binary search to find an estimate $\widetilde{\lambda}$ such that $\widetilde{\lambda} - \frac{\eps^3}{100} \leq \lambda_{\bx, \O_{G,\eps}} \leq \widetilde{\lambda}$, where $\lambda_{\bx, \O_{G,\eps}} \defeq \max_{B \in \O_{G,\eps}}\lambda_{\bx, B}$.
Fix a current value of $\widetilde{\lambda}$.
For each odd $3 \leq \ell \leq 1/\eps$, \cite{AhnG14} constructs an integral weighted graph $G_{\varphi}(\ell, \widetilde{\lambda})$ with $V(G_{\varphi}(\ell, \widetilde{\lambda})) = V(G) \cup \{s\}$ with $O(\min\{m, n\eps^{-5}\})$ edges whose weights sum up to $O(n \eps^{-5})$ in $O(m)$ time.
Let $\mathsf{cut}(B)$ for $B \subseteq V$ be the sum of weights of edges between $B$ and $(V(G) \cup \{s\}) \setminus B$ in $G_{\varphi}(\ell, \widetilde{\lambda})$.
If $\widetilde{\lambda}$ is an accurate estimate, i.e., $\widetilde{\lambda} - \frac{\eps^3}{100} \leq \lambda_{\bx, \O_{G,\eps}} \leq \widetilde{\lambda}$, then \cite[Property 1]{AhnG14} showed that
\begin{enumerate}[(1)]
  \item every $\ell$-sized odd set $B \in \L_{\bx}$ satisfies $\mathsf{cut}(B, \overline{B}) < \kappa(\ell)$, where $\kappa(\ell) \defeq \left\lfloor \varphi \widetilde{\lambda}\left(1 - \frac{\eps^2\ell^2}{2}\right)\right\rfloor + \frac{12\ell}{\eps} + 1 < 2\varphi$ for $\varphi \defeq 50\eps^{-4}$, and
  \item\label{item2} every $\ell$-sized odd set $B \in \O_{G,\eps}$ that satisfies $\mathsf{cut}(B, \overline{B}) < \kappa(\ell)$ belongs to the collection $\L_{\bx}^\prime \defeq \left\{B \in \O_{G,\eps}: \lambda_{\bx, B} \geq \lambda_{\bx, \O_{G,\eps}} - \eps^3\right\}$.
\end{enumerate}

They then applied the following algorithm of \cref{lemma:find-odd-sets} below with $\kappa \defeq \kappa(\ell) = O(\eps^{-4})$ to obtain a collection of vertex sets $\L$, for which they showed that $\L_{\bx} \subseteq \L$ and therefore we can extract all size-$\ell$ sets in $\L_{\bx}$ in $\widetilde{O}(m)$ time by simply checking their value of $\lambda_{\bx, B}$.\footnote{Since $\L = \L_1 \cup \cdots \cup \L_{O(\log{n})}$ where sets in $\L_i$ are disjoint by \cref{lemma:find-odd-sets}, for each $1 \leq i \leq O(\log{n})$ we can compute $\lambda_{\bx, B}$ for $B \in \L_i$ in $O(m)$ time.}
The value of $\widetilde{\lambda}$ can then be adjusted based on whether any size-$\ell$ set in $\L_{\bx}$ is found (recall that we are doing a binary search to determine an accurate estimate $\widetilde{\lambda}$).

To compute the collection $\L$, \cite{AhnG14} used the minimum odd-cut approach of \cite{PadbergR82} using the construct of partial Gomory-Hu trees.
A \emph{$\kappa$-partial Gomory-Hu tree} of a (possibly weighted) graph $G$ is a partition $\mathcal{U} = \{U_1, \ldots, U_k\}$ of $V(G)$ and a weighted tree $\mathcal{T}$ on $\mathcal{U}$ such that each $x, y \in U_i$ belonging to the same set has minimum cut value greater than $\kappa$ in $G$, and each $x \in U_i$ and $y \in U_j$ in different sets has minimum cut (both the value and the actual cut) equal to that induced by $\mathcal{T}$.

\begin{lemma}[{\cite[Lemma 19 and Algorithm 3]{AhnG14}}]
  Given a $T(n, m, \kappa)$-time algorithm $\mathcal{G}$ that constructs a $\kappa$-partial Gomory-Hu tree of an $n$-vertex $m$-edge graph, there is an algorithm that computes an $\widetilde{O}(n)$-sized collection $\L$ of odd sets in $G$ where $s \in V(G)$ in $\widetilde{O}\left(T(n, m, \kappa)\right)$ time such that (i) $s \not\in B$ for every $B \in \L$, (ii) $E_G(B, \overline{B}) \leq \kappa$ for every $B \in \L$, and (iii) every odd set $B^\prime$ not containing $s$ with $E_G(B^\prime, \overline{B^\prime}) \leq \kappa$ intersects with some $B \in \L$.
  Moreover, $\L$ is of the form $\L_{1} \cup \cdots \cup \L_{O(\log{n})}$, where each $\L_i$ contains disjoint vertex subsets.
  The algorithm is deterministic if $\mathcal{G}$ is.
  \label{lemma:find-odd-sets}
\end{lemma}

\begin{proof}
  The lemma is the same as \cite[Lemma 19]{AhnG14} modulo the additional guarantee that $\L$ is the union $O(\log{n})$ collections of disjoint sets.
  This is manifest from the implementation of \cite[Algorithm 3]{AhnG14}, using that it has $O(\log{n})$ iterations, and in each of them the sets it found are disjoint.
\end{proof}

We deliberately make the statement of \cref{lemma:find-odd-sets} flexible to the choice of partial Gomory-Hu tree algorithm, given that faster algorithms were developed recently.
\cite{AhnG14} used the following algorithm in the regime where $\kappa$ is small.

\begin{lemma}[\cite{HariharanKPB07,HariharanKP07}]
  There is a randomized $\widetilde{O}(m + n\kappa^2)$ time algorithm for constructing a $\kappa$-partial Gomory-Hu tree on an $m$-edge $n$-vertex unweighted graph.
  \label{lemma:near-linear-time-gomory-hu-tree}
\end{lemma}

Alternatively, we can use the recent almost-linear time Gomory-Hu tree algorithm which does not depend on $\kappa$ at the cost of having a subpolynomial factor $m^{o(1)}$.
Note that contracting edges greater than $\kappa$ in a Gomory-Hu tree trivially gives a $\kappa$-partial Gomory-Hu tree.

\begin{lemma}[\cite{AbboudLPS23}]
  There is a randomized $m^{1+o(1)}$ time algorithm that constructs a Gomory-Hu tree on an $m$-edge weighted graph w.h.p.
  \label{lemma:almost-linear-time-gomory-hu-tree}
\end{lemma}

\begin{lemma}[{\cite[Theorem 6 and Lemma 17]{AhnG14}}]
  There are randomized $\widetilde{O}(m\eps^{-1} + n\eps^{-9})$-evaluator and $\widehat{O}(m\eps^{-1})$-evaluator for $\widetilde{\mathcal{M}_{G,\eps}}$.
  \label{lemma:evaluators}
\end{lemma}

\begin{proof}
  The runtime of the algorithms can be analyzed as follows.
  The binary search takes $\widetilde{O}(1)$ iterations.
  In each iteration, $O(\eps^{-1})$ values of $\ell$ are enumerated, and for each of them we either use \cref{lemma:near-linear-time-gomory-hu-tree} for constructing partial Gomory-Hu tree in \cref{lemma:find-odd-sets}, resulting in an $\widetilde{O}(m + n\eps^{-8})$ time algorithm for computing size-$\ell$ sets in $\L_{\bx}$, or we use \cref{lemma:almost-linear-time-gomory-hu-tree} and get an $\widehat{O}(m)$ time algorithm.
  
  By \cite[Lemma 17]{AhnG14}, if $\lambda_{\bx} > 1+8\eps$ but $\lambda_{\bx,O_{G,\eps}} < 1+3\eps$, then $\L_{\bx} = \emptyset$.
  Since the above algorithms based on \cite[Theorem 6]{AhnG14} correctly identify $\L_{\bx}$ when $\lambda_{\bx,O_{G,\eps}} \geq 1 + 3\eps$, by checking whether there is indeed an odd set $B$ returned with $\lambda_{\bx,B} \geq 1 + 3\eps$ we can deduce whether we should return an empty set or not.
  Similarly, we can also deduce whether $\lambda_{\bx} \leq 1 + 8\eps$ by inspecting if there is a vertex or returned odd set which is violated by a $1 + 8\eps$ factor by $\bx$.
\end{proof}

With \cref{lemma:main-appendix,lemma:evaluators}, \cref{lemma:main-general} follow by appropriately adjusting the parameter $\eps$.

\begin{remark}
  From the discussion above in combination with the deterministic rounding algorithm in \cref{sec:rounding} we can also see that to obtain a deterministic $\widetilde{O}_{\eps}(1)$ update time decremental matching algorithm, it suffices to get a deterministic, $\widetilde{O}_{\kappa}(m)$ time algorithm for constructing $\kappa$-partial Gomory-Hu tree.
  We leave this as an interesting open question and future direction.
\end{remark}

\section{Omitted Proofs in \texorpdfstring{\cref{sec:rounding}}{Section 6}} \label{appendix:rounding-proofs}

In this section we provide proofs of the following lemmas.

\Support*

\begin{proof}
  By reverse induction on $i$ we show that $\supp(\bx^{(i)}) \subseteq \supp(\bx^{(i+1)}) \subseteq \supp(\bx)$.
  From \eqref{eq:def-x} we have $\supp(\bx^{(i)}) = F_i \cup \bigcup_{0 \leq j \leq L_{\min}}\supp_j(\bx) \cup \bigcup_{L_{\min} + 1 \leq j \leq i}E_j$ and thus it suffices to show that $F_i \subseteq \supp(\bx)$ and $E_j \subseteq \supp(\bx)$ for all $L_{\min} + 1 \leq j \leq i$ at all times.
  That $E_j \subseteq \supp(\bx)$ is immediate as $E_j$ is set to $\supp_j(\bx)$ in $\texttt{Rebuild()}$, and in each subsequent update we will remove $e$ from $E_j$ if $\bx_e$ is changed.
  By the property of $\texttt{degree-split}$, we have $F_i \subseteq E_{i+1} \cup F_{i+1}$, which by the inductive hypothesis satisfies $E_{i+1} \cup F_{i+1} \subseteq \supp(\bx^{(i+1)}) \subseteq \supp(\bx)$, and thus $F_i \subseteq \supp(\bx^{(i+1)}) \subseteq \supp(\bx)$ after a call to $\texttt{Rebuild()}$.
  Similarly, in each update after the rebuild, we will remove $e$ from $F_i$ whose $\bx_e$ is changed, and thus the containment is maintained.
  This proves the lemma.
\end{proof}

\SizeOfMatching*

\begin{proof}
  Choosing the subgraph returned by \texttt{degree-split} with large weight ensures that, right after a call to $\texttt{Rebuild(}i\texttt{)}$, we have $\bw(F_{i-1}) \geq \frac{\bw(E_i) + \bw(F_i)}{2}$ and thus $\bw^\top \bx^{(i-1)} \geq \bw^\top \bx^{(i)}$ by that
  \begin{equation}
    \bx^{(i-1)}_e - \bx^{(i)}_e = F_{i-1}(e) \cdot 2^{-(i-1)} - \left(F_{i}(e) + E_{i}(e)\right) \cdot 2^{-i}.
    \label{eq:diff}
  \end{equation}
  On the other hand, there are at most $2^{i-2} \cdot \frac{\eps \cdot \bw^\top \bx}{L_{\max} \cdot W}$ updates
  between two calls to $\texttt{Rebuild(}i\texttt{)}$ and thus at most $2^{i-1} \cdot \frac{\eps \cdot \bw^\top \bx}{L_{\max} \cdot W}$ edges are deleted from $F_{i-1}$, decreasing $\bw^\top \bx^{(i-1)}$ by at most $\frac{\eps}{L_{\max}} \cdot \bw^\top \bx$.
  By \eqref{eq:diff}, any change to $E_j$ for $j \neq i$ does not affect the difference between $\bw^\top \bx^{(i-1)}$ and $\bw^\top \bx^{(i)}$.
  Furthermore, $E_i$ and $F_i$ can only decrease until $\texttt{Rebuild(}i\texttt{)}$ is called (if $\texttt{Rebuild(}i+1\texttt{)}$ is called and $F_i$ increases consequently, then $\texttt{Rebuild(}i\texttt{)}$ will be called as well).
  This shows except for operations that remove edges from $F_{i-1}$, $\bw^\top \bx^{(i-1)} - \bw^\top \bx^{(i)}$ cannot decrease, and as such
  \begin{equation}
    \bw^\top \bx^{(i-1)} \geq \bw^\top \bx^{(i)} - \frac{\eps}{L_{\max}} \cdot \bw^\top \bx
    \label{eq:weight-decrease}
  \end{equation}
  holds at all times for each $i \in \{L_{\min} + 1, \ldots, L_{\max}\}$.
  Chaining \eqref{eq:weight-decrease} for all $i$  and noticing that $\bx = \bx^{(L_{\max})}$ and $\widetilde{\bx} \defeq \bx^{(L_{\min})}$ conclude the proof.
\end{proof}

\UpdateTime*

\begin{proof}
  We first show using a similar argument as the proof of \cref{lemma:degree} that $|F_i| \leq O(2^i \cdot \norm*{\bx}_1)$.
  By backward induction on $i$ we prove
  \begin{equation}
    |F_i| \leq \left\lceil 2^{i} \cdot \sum_{j=i+1}^{L_{\max}} |S_j| \cdot 2^{-j} \right\rceil.
    \label{eq:induction-2}
  \end{equation}
  For $i = L_{\max}$ this is trivially true.
  For $i < L_{\max}$, right after a call to $\texttt{Rebuild(}i+1\texttt{)}$ when $E_{i+1} \gets S_{i+1}$, by Property \labelcref{item:edge-set-size} of the subroutine $\texttt{degree-split}$ we have
  \[
    |F_i| \leq \left\lceil\frac{|S_{i+1}| + |F_{i+1}|}{2}\right\rceil \leq \left\lceil\frac{|S_{i+1}| + \left\lceil 2^{i+1} \cdot \sum_{j=i+2}^{L_{\max}}|S_j| \cdot 2^{-j}\right\rceil}{2}\right\rceil \leq \left\lceil 2^{i} \cdot \sum_{j=i+1}^{L_{\max}}|S_j| \cdot 2^{-j} \right\rceil.
  \]
  On the other hand, for an update after which $\texttt{Rebuild(}i+1\texttt{)}$ is not called, as in the proof of \cref{lemma:degree}, if the right-hand side decreases by one, then we will remove an extra edge from $F_i$ and thus the inequality still holds.
  This proves the \eqref{eq:induction-2} which in turn implies $|F_i| \leq O(2^i \cdot \norm*{\bx}_1)$.

  This shows that the call to $\texttt{degree-split}$ in $\texttt{Rebuild(}i\texttt{)}$ takes $O(2^i \cdot \norm*{\bx}_1)$ time.
  Since $\texttt{Rebuild(}i\texttt{)}$ causes $\texttt{Rebuild(}i-1\texttt{)}$, the total running time of $\texttt{Rebuild(}i\texttt{)}$ is $\sum_{j=0}^{i}O(2^j \cdot \norm*{\bx}_1) = O(2^i \cdot \norm*{\bx}_1)$.
  Because $\texttt{Rebuild(}i\texttt{)}$ is called once every $2^{i-2} \cdot \frac{\eps \cdot \bw^\top \bx}{L_{\max} \cdot W} \geq 2^{i-2} \cdot \frac{\eps \cdot \norm*{\bx}_1}{L_{\max} \cdot W}$ updates, the amortized update time of \cref{alg:det-rounding} is $\widetilde{O}(W \cdot \eps^{-2} \cdot L_{\max})$.
\end{proof}

\end{document}